\documentclass[a4paper,twoside,openright]{report}

\usepackage{illcmolthesis}
\title{Reactive Valuations}
\author{Bert Christiaan Regenboog}
\birthdate{April 11, 1983}
\birthplace{Rotterdam}
\defensedate{December 15, 2010}
\supervisor{Dr Alban Ponse}
\committeemember{Dr Alban Ponse}
\committeemember{Dr Inge Bethke}
\committeemember{Prof Dr Frank Veltman}
\committeemember{Dr Yde Venema}

\usepackage{amsmath,amsthm,amssymb,latexsym,stmaryrd}
\usepackage[bookman]{quotchap}
\usepackage[all]{xy}
\usepackage{fancyhdr}
\fancypagestyle{plain}{%
	\fancyhead{}
	
}
\usepackage[pdftex,colorlinks=true]{hyperref}

\newcommand{\sig}{\ensuremath{\Sigma_{CP}(A)}}
\newcommand{\dd}[1]{\frac{\partial}{\partial #1}}
\newcommand{\BF}{\ensuremath{\textbf{BF}}}
\newcommand{\syn}{\bumpeq}
\newcommand{\nsyn}{\not\bumpeq}

\newcommand{\CP}[1]{\ensuremath{\mathrm{CP#1}}}
\newcommand{\CPrp}[1]{\ensuremath{\mathrm{CPrp#1}}}
\newcommand{\CPcr}[1]{\ensuremath{\mathrm{CPcr#1}}}
\newcommand{\CPcontr}{\ensuremath{\mathrm{CPcontr}}}
\newcommand{\CPstat}{\ensuremath{\mathrm{CPstat}}}

\newcommand{\lef}{\ensuremath{\triangleleft}}
\newcommand{\rig}{\ensuremath{\triangleright}}

\newcommand{\leftand}{~
     \mathbin{\setlength{\unitlength}{1ex}
     \begin{picture}(1.4,1.8)(-.3,0)
     \put(-.6,0){$\wedge$}
     \put(-.53,-0.36){\circle{0.6}}
     \end{picture}
     }}

\newcommand{\rightand}{~
     \mathbin{\setlength{\unitlength}{1ex}
     \begin{picture}(1.4,1.8)(0,0)
     \put(-.8,0){$\wedge$}
     \put(.72,-0.36){\circle{0.6}}
     \end{picture}
     }}

\newcommand{\leftor}{~
     \mathbin{\setlength{\unitlength}{1ex}
     \begin{picture}(1.4,1.8)(-.3,0)
     \put(-.6,0){$\vee$}
     \put(-.53,1.7){\circle{0.6}}
     \end{picture}
     }}

\newcommand{\rightor}{~
     \mathbin{\setlength{\unitlength}{1ex}
     \begin{picture}(1.4,1.8)
     \put(-.8,0){$\vee$}
     \put(.72,1.7){\circle{0.6}}
     \end{picture}
     }}

 \newcommand{\leftimp}{~
     \mathbin{\setlength{\unitlength}{1ex}
     \begin{picture}(1.5,1.8)
     \put(-.1,0){$\rightarrow$}
     \put(-.3,0.57){\circle{0.6}}
     \end{picture}
     ~}}
     
\newcommand{\rightimp}{~
     \mathbin{\setlength{\unitlength}{1ex}
     \begin{picture}(1.5,1.8)
     \put(-.9,0){$\rightarrow$}
     \put(1.76,0.57){\circle{0.6}}
     \end{picture}
     ~}}
     
\newcommand{\leftbiimp}{~
     \mathbin{\setlength{\unitlength}{1ex}
     \begin{picture}(1.5,1.8)
     \put(-.1,0){$\leftrightarrow$}
     \put(-.44,0.57){\circle{0.6}}
     \end{picture}
     ~}}
     
\newcommand{\rightbiimp}{~
     \mathbin{\setlength{\unitlength}{1ex}
     \begin{picture}(1.5,1.8)
     \put(-.9,0){$\leftrightarrow$}
     \put(1.76,0.57){\circle{0.6}}
     \end{picture}
     ~}}

\newtheorem{theorem}{Theorem}[chapter]

\newtheorem{prop}[theorem]{Proposition}
\newtheorem{lem}[theorem]{Lemma}
\newtheorem{thm}[theorem]{Theorem}
\newtheorem{defn}[theorem]{Definition}
\newtheorem{cor}[theorem]{Corollary}

\begin{document}
\maketitle
\thispagestyle{empty}
~\newpage

\pagenumbering{roman}
\tableofcontents

\begin{abstract}
In sequential logic there is an order in which the atomic propositions
in an expression are evaluated. This order allows the same atomic proposition
to have different values depending on which atomic propositions have
already been evaluated. In the sequential propositional logic
introduced by Bergstra and Ponse in \cite{main}, such valuations are
called ``reactive'' valuations, in contrast to ``static'' valuations as
are common in e.g.\ ordinary propositional logic. There are many
classes of these reactive valuations e.g., we can define a class of
reactive valuations such that the value for each atomic proposition
remains the same until another atomic proposition is evaluated.

This Master of Logic thesis consists of a study of some of the properties of
this logic.

We take a closer look at some of the classes of reactive valuations mentioned in \cite{main}. We particularly focus on the relation between the axiomatization and the semantics. Consequently, the main part of this thesis focuses on proving soundness and completeness. Furthermore, we show that the axioms in the provided axiomatizations are independent i.e., there are no redundant axioms present. Finally, we show $\omega$-completeness for two classes of reactive valuations.

%Furthermore, the relation between reactive valuations and other logics
%needs to be investigated. Certainly, from a historical perspective it
%is important to know to what extent this kind of work has already been
%done. In addition, we can look at how other logics can be expressed in
%terms of reactive valuations. For example, in \cite{main} Bergstra and
%Ponse show such a connection between McCarthy's logical calculus and
%reactive valuations.
\end{abstract}

\thispagestyle{empty}
~\newpage

\pagenumbering{arabic}

\chapter{Introduction}

\section{Introduction}
In sentential logic (also called propositional calculus), sentences are build from atomic propositions, the constants true and false, and connectives such as $\neg$, $\wedge$, $\vee$, etc. The truth of such a sentence with respect to a model, is calculated using the interpretation function associated with that model. This function not only assigns meaning to the connectives and constants but also to the individual atomic propositions.

In sentential logic the interpretation of connectives and constants is given. Hence, a model in sentential logic is uniquely defined by the interpretation of the individual atomic propositions. These atomic propositions are assigned either the value true or the value false by the interpretation, indicating whether they are true or false in the model. Such an assignment is referred to as a valuation. In sentential logic, these valuations entirely depend on the atomic propositions they assign a value to and not on other external factors. Consequently, a valuation will give an atomic proposition the same valuation no matter its location within a sentence, and this valuation will never change. These valuations are, in a manner of speaking, static.

This static behaviour can be considered a severe limitation of sentential logic. For example, sentential logic is not sufficiently expressive for modelling logical conjunction as implemented in most programming languages because the conjunction in these cases is non-commutative\footnote{We explain in the next section why this is the case.}. In order for us to effectively model these and other kinds of connectives and sequential systems, we are required to extend our notion of valuation.

This thesis is based on the work by Bergstra and Ponse in \cite{main}. They introduce a logic that uses reactive valuations instead of normal valuations. Reactive valuations allow us to take previously evaluated atomic propositions into account. Thus the valuations are in a sense ``reactive''. The use of reactive valuations necessitates the need for expressions in this logic to be evaluated in some fixed order. Hence the resulting logic has a sequential interpretation. The same atomic proposition may have a different value depending on which atomic propositions have previously been evaluated. The reactive valuations have thus an additional dependence on a sequence of atomic propositions representing the history of evaluation.

The signature of this logic consists of a finite set of atomic propositions and the constants $T$ and $F$ plus a ternary operator $\_\lef\_\rig\_$. The constants $T$ and $F$ denote true and false, respectively. The ternary operator denotes conditional composition i.e., an \emph{if-then-else} operator. For example, $a\lef b\rig c$ translates to if $b$ then $a$ else $c$. This then clues us to the order in which expressions of this type are evaluated i.e., the antecedent is evaluated first. The question which of the two consequents is then evaluated first is irrelevant because their value depends only on the antecedent and not on each other.

For example, take the expression
\[
a\lef a\rig b
\]
The letters $a$ and $b$ represent atomic propositions. The above reads thus \emph{if a then a else b}. In sentential logic, it suffices to know the value of $a$ and $b$ to know whether the sentence is true or false, see the following table:
\[
\begin{array}{cc|c}
a & b & a\lef a\rig b\\
\hline
F & F & F\\
F & T & T\\
T & F & T\\
T & T & T
\end{array}
\]
However this is not the case if we are using reactive valuations. Keeping the if-then-else interpretation in mind, we intuitively begin by evaluating the middle $a$, the antecedent. The act of evaluating this $a$ can possibly have influence on the valuation of the left-hand $a$ and the right-hand $b$. We denote the value of $a$ given a valuation $H$ as $a/H$. Furthermore, the valuation obtained after evaluating $a$ is denoted as $\dd a(H)$. So $\dd a$ can be viewed as a function that maps reactive valuations to other reactive valuations. The value $(a\lef a\rig b)/H$ is determined by the values of $a/H$, $a/\dd a(H)$ and $b/\dd a(H)$, as illustrated in the following table: 
\[
\begin{array}{ccc|c}
a/H & a/\dd a(H) & b/\dd a(H) & (a\lef a \rig b)/H\\
\hline
F & F & F & F\\
F & F & T & T\\
F & T & F & F\\
F & T & T & T\\
T & F & F & F\\
T & F & T & F\\
T & T & F & T\\
T & T & T & T
\end{array}
\]
Compared to sentential logic, we have an extra parameter because $a/H\ne a/\dd a(H)$. Note that it is not possible for the valuation of $b$ to influence the valuation of either of the $a$'s. Nor is it possible that the left-hand $a$ has influence on the valuation of the middle $a$ or the right-hand $b$.

There are other limitations to reactive valuations. For example, take the expression $b\lef a\rig b$. This has the following truth table:
\[
\begin{array}{cc|c}
a/H & b/\dd a(H) & (b\lef a\rig b)/H\\
\hline
F & F & F\\
F & T & T\\
T & F & F\\
T & T & T\\
\end{array}
\]
The value of $(b\lef a\rig b)/H$ is in this case just computed using two values $a/H$ and $b/\dd a(H)$. It is not possible to assign different values to the left-hand $b$ and the right-hand $b$ using reactive valuations because reactive valuations do not take into account the value of previously observed atomic propositions. Only the act of evaluating atomic propositions influences reactive valuations, regardless of what those values might have been.

The class of all reactive valuations is referred to as the free reactive valuations. We can construct different logics by constraining the type of reactive valuations we allow. For example, if we take the class of reactive valuations that ignore the sequence of previously evaluated atomic propositions, we get the static valuations. Static valuations coincide with the classical valuations in sentential logic i.e., they always give the same value for an atomic proposition independent of context.

Another example of a class of reactive valuations are the contractive valuations. In these valuations the value of an atomic proposition, say $a$, remains the same as long as no atomic proposition other than $a$ is evaluated. This is in contrast with the free reactive valuation where each instance of atomic proposition $a$ in a sequence of $a$'s can have a different value. For example, if we are using free reactive valuations it is possible to assign different values to the $a$'s in the expression $F\lef a\rig a$. This is not possible if we are using contractive valuations because between the first $a$ and the second $a$ no other atomic proposition is evaluated, and thus the valuation of $a$ must remain the same.

We can formalize the idea of creating new logics using classes of reactive valuations. A class $K$ of reactive valuations gives rise to an equivalence relation, $K$-equivalence. Since reactive valuations are sensitive to the context, this $K$-equivalence is not necessarily congruent i.e., equivalence of a term $t$ need not be preserved when substituting equivalent subterms in the term $t$. For example, for every reactive valuation $H$ we have $(T\lef a\rig T)/H=T/H=T$. It is however not the case that $(b\lef T\rig F)/H=(b\lef(T\lef a\rig T)\rig F)/H$ because in the right-hand term the valuation of $b$ depends on $a$ which is not the case in the left-hand term. Since congruence is a necessary property, we therefore introduce $K$-congruence as the largest congruence contained in $K$-equivalence. $K$-congruence thus represents our semantics.

Besides a semantical characterization, each logic can be equationally specified using a number of axioms. For example, the axiom $x\lef T\rig y=x$, where $x$ and $y$ are arbitrary terms, is an axiom shared by every logic we present. Given these axiomatizations we can prove properties such as soundness and completeness.

In the rest of this chapter we further motivate why reactive valuations are relevant, and discuss some related work. In Chapter~2, a formal introduction is given to reactive valuations. In addition, four varieties are introduced and discussed. These varieties include the varieties of free, contractive and static valuations we mentioned earlier. This discussion includes proper axiomatizations and subsequent proofs of completeness and soundness of these varieties. In Chapter~3 a definition of $\omega$-completeness is given, and we explain why $\omega$-completeness is a nice property of an axiomatization. Subsequently, a proof of $\omega$-completeness for the variety of free reactive valuations and the variety of static valuations is presented. The axiomatizations given in Chapter~2 might contain redundant axioms. In Chapter~4, we show that this is mostly not the case. Finally, Chapter~5 contains a summary, and a few suggestions for further research.

\section{Motivation}
Static valuations, the type of valuations used in sentential logic, are inadequate to model many sequential systems. However, we can model those systems using different classes of reactive valuations.

Using reactive valuations we can model non-commutative logical connectives. For example, $\rightor$ is disjunction in which the right argument is evaluated first (notation is taken from \cite{connectives}). So in the signature of our logic $x\rightor y$ is defined as $T\lef y\rig x$. In similar fashion $x\leftor y$ is defined as $T\lef x\rig y$. In sentential logic these two definitions would coincide. This however is not the case if we use reactive valuations i.e., this allows us to distinguish between $x\rightor y$ and $x\leftor y$. Hence, reactive valuations are suited for modelling non-commutative connectives. One area where non-commutative connectives are commonplace, is that of programming languages

In most programming languages, it is possible that a function, in addition to producing a value, also does something else. It might for example raise an exception or modify a global variable. This kind of behaviour is called a side-effect of said function. Furthermore, it is also possible that the return value of a function might depend on some external factor. For example, a database or a random number generator. Finally, expressions are evaluated sequentially. This means that if want to evaluate $x\wedge y$, the interpreter has to decide whether to evaluate $x$ first or $y$ first.

Combining these facts, we could get a situation in which the value of the expression $f(x)\wedge g(y)$ depends on whether $f(x)$ is evaluated first or $g(y)$ because $f(x)$ might influence the value of $g(y)$, and vice versa. Admittedly it is limited to situations that can be translated to boolean formulas, but this is the kind of behaviour we can model with reactive valuations and not with sentential logic.

Short-circuit evaluation is a common feature of programming languages. Short-circuit evaluation is usually limited to the evaluation of a few specific operators. Using such evaluations only the arguments that have to be evaluated, are actually evaluated. The operator \verb+&&+, logical conjunction, in C/C++ is an example of a short-circuiting operator. Consider evaluating the expression \verb+x && y+. If \verb+x+ evaluates to false, the second argument \verb+y+ is not evaluated because regardless of its value \verb+x && y+ evaluates to false. If $x$ and $y$ do not have any side-effects and their values are limited to true and false, this operator is commutative. However, in practice it is possible that $x$ and $y$ represent some computation that does not necessarily terminate. Consequently, if one of the arguments does not terminate, the value of a short-circuiting operator like \verb+&&+ might be different depending on which argument we decide to evaluate first. Hence, this is another example of a non-commutative connective with a symmetric counterpart.

Arguably the programming language with the most direct connection with reactive valuations is Prolog. Prolog is originally designed to model language through computational models based on predicate logic. This paradigm of programming in terms of predicate logic is called logic programming. As a result, programs in Prolog almost read like logical formulas, and are referred to as predicates. In the early days of logic programming, the language did not have any instructions with explicit side-effects. However, for Prolog to have some practical value extra instructions are needed. For example, the database instructions \emph{assert} and \emph{retract}. These instructions can, perhaps not surprisingly, assert and retract facts to a Prolog program. Clearly, programs using these instructions have side-effects that might influence whether predicates evaluate to true or false. For example, predicate $P$ is true if fact $A$ is true. In addition, predicate $P$ retracts fact $B$. Hence, if there is another predicate $Q$, which is true if $B$ is true, and retracts $A$, the predicates influence each others value by their side-effects.

Staying within the field of computer science, the ``reactive behaviour'' illustrated by the previous examples does not limit itself to programming languages. On the more lower hardware level we have the term ``sequential logic'' in circuit theory. Here sequential logic refers to logic circuits that have a memory. The output of such a circuit does not only depend on the input, but also on the history of inputs. These circuits can be used to construct finite state machines such as Moore and Mealy machines. The output of these machines depends on an internal state, which in turn depends on the previous state and input. 

In everyday reasoning, so called common-sense reasoning, the assertion and retraction of facts is fairly common. For example, while it may be true that Jack is at home in the evening, it certainly does not have to be true that he is always at home. In addition, this is not limited to the physical world but can also include the beliefs of agents. For example, one might believe that all adult swans have white plumage, until one travels to Australia and sees that there are swans with black plumage, at which point the beliefs are revised. Ordinary classical logic is not equipped to model these reactive processes i.e., the validity of propositions remains the same.

Pragmatics is a subfield of linguistics in which the interaction between utterances i.e., speech acts, is studied. One example of such interaction is that of presupposition. Presupposition refers to implicit assumptions in sentences. Take for example the sentence ``Jack drives his car to the mall''. This sentence presupposes that Jack has a car. So modelling presupposition requires we are able to deal with side-effects of posing a proposition, another example of reactive evaluation. 

In this section we presented a number of processes that might be modelled using reactive valuations as motivation why reactive valuations are interesting. In the following section most of the aforementioned examples will be examined once more except this time in the context of related work i.e., we compare existing literature on these subjects with reactive valuations.

\section{Related work}
One of the defining properties of our logic is that the valuation of the atomic propositions changes depending on what atomic propositions have been evaluated. In this section we discuss some other logics and formalisms that also demonstrate this property.

As stated earlier in this thesis is based on the work done by Bergstra and Ponse in \cite{main}. In it they introduce reactive valuations and the varieties which we will study in the next chapters. They discuss a number of topics which will not be covered in this thesis. These topics include a method of modelling a three-valued logic using reactive valuations, expressivity results, the complexity of satisfiability, and a study of the properties of infinite propositions.

Since reactive valuations are a relatively new invention with no clearly defined predecessor, there is no related work that deals specifically with reactive valuations besides the one by Bergstra and Ponse. There is, however a very large body of literature dealing with sequential reasoning. This literature ranges from computer science to philosophy and linguistics. An exhaustive literature overview is however beyond the scope of this thesis and would in all likelihood constitute a thesis all on its own. This chapter, therefore, gives a very brief overview with a few specific examples, which will hopefully offer a starting point for a more detailed account of related work.

The previous section on motivation already gave a few examples of areas where reactive valuations might be applied and hence literature dealing with the phenomena described in that section can be considered related to the theory of reactive valuations.

For example we mentioned common-sense reasoning i.e., the type of reasoning we use in our daily lives. Common-sense reasoning has been studied in many fields but it has enjoyed renewed attention the past decades with the rise of the field of artificial intelligence where it is mostly referred to as non-monotonic reasoning.

In classical logic when a statement $\phi$ logically follows from a set $S$ of premises, it is the case that $\phi$ logically follows from a superset $S'\supseteq S$ of premises. Consequently, we call this logic monotonic. This means that once something is true it will remain so. In common-sense reasoning this is not the case. Hence this type of reasoning is called non-monotonic. Similarly reactive valuations are non-monotonic due to ever-changing valuations. See \cite{non-monotonic} and \cite{non-monotonic2} for an overview of non-monotonic reasoning.

Similarly in philosophy we have defeasible reasoning which deals with arguments that might be true but can be refuted at a later point by observing new facts, see \cite{defeasible}.

We can view the evaluation of an expression as the execution of a program. The atomic propositions would then correspond to single instructions or pieces of programs such as procedures or functions. There are many formalisms that are designed for reasoning about propositional properties of programs, e.g. Hoare logic, temporal logic of actions and propositional dynamic logic (PDL). 

For example, PDL (see \cite{PDL} for an overview) can be effectively used to model reactive valuations. In PDL we have a set of atomic propositions $P$, a set of basic actions $A$, a set of states $S$, and a binary relation $R$ on $S$. The connection between PDL and reactive valuations becomes evident if for every atomic proposition in $P$ there is a basic action in $A$ that signifies the evaluation of the respective atomic proposition. Then the states $S$ correspond to the various reactive valuations. In that context deterministic PDL i.e., the class of frames characterized by $\left<a\right>\phi\to\left[a\right]\phi$, is of particular interest because it illustrates one of the limitations of reactive valuations. Namely, that in the previously mentioned example $b\lef a\rig b$ the value of both $b$'s will have to be the same.

As mentioned in the previous section the programming language Prolog has special instructions for the assertion and retraction of facts.  Consider the following Prolog program
\begin{verbatim}
    p(a) :- p(b), retract(p(b)).
    p(a) :- assert(p(b)), fail.
\end{verbatim}
The statement \verb|fail| is a reserved keyword that automatically fails i.e., somewhat similar to the constant false. When repeatedly asking the interpreter \verb|p(a)| we get the output sequence 0101010101\ldots where 0 and 1 stand for \verb|no| and \verb|yes|, respectively.

In \cite{prolog_expr} the expressive power of the side effects of the assert and retract statements in Prolog is investigated. The authors main tool in this analysis are these output sequences. Much the same as we consider different varieties they consider different classes of output sequences e.g., constant sequences that represent programs with no side effect and ultimately periodic sequences where the sequence at some point starts to repeat itself. For a complete denotational semantics of Prolog, see \cite{prolog_semantics}.

Reactive valuations give rise to directed versions of connectives such as $\wedge$ and $\vee$. In Chapter~2 we define a number of these connectives. The notation for these is taken from \cite{connectives}, where a number of many-valued logics are described. In \cite{belnap_cond} an axiomatization of Belnap's four-valued logic is given using conditional composition, which as you might recall is the only connective in the signature of the logic described in this thesis. The notation for conditional composition is taken from a paper by Hoare, \cite{hoare}. In this paper Hoare describes what we call the variety with static valuations. This variety corresponds to boolean algebra, as we shall show in Chapter~3.

In the previous section we briefly mentioned pragmatics, and more specifically presupposition. The most commonly used formalism to describe presupposition and its effects is discourse representation theory (see \cite{DHT}). There are, however, different approaches. For example, in \cite{kracht} a many-valued logic with directed connectives is used to investigate some of the main problems in presupposition.

Besides the ones we just mentioned, there are many other research areas that deal with sequentiality that we did not mention here. For example, temporal logic, substructural logics and non-commutative logics. As mentioned before this section is but a brief overview, and we hope this will prove to be a useful point of departure for a more thorough investigation into related work.

Lastly, it took more than two years to write this thesis. This year a new paper on the subject of reactive valuations by Bergstra and Ponse appeared, see \cite{future}. The new results in that paper are not discussed here nor do the results in this thesis depend on those results.

%Coincidentally, Belnap et al wrote an interesting paper on indeterminism, \cite{belnap_indeter}

\chapter{Reactive valuations}

This chapter represents the main body of the thesis. Reactive valuations are formally introduced, which enables us to define a number of different logics. Subsequently, some basic properties such as soundness and completeness are proven. 

\section{Introduction}
In this section we introduce a sequential propositional theory starting with the language. The symbols of our language are as follows:
\begin{itemize}
\item the constants $T$ and $F$
\item the ternary operator $\_\lef\_\rig\_$, called conditional composition
\item a finite non-empty set $A$ of atomic propositions
\item an infinite set of variables $V$
\end{itemize}
The notation $\_\lef\_\rig\_$ for conditional composition was first introduced by Hoare in \cite{hoare}. We use the letters $x, y, z, u, v, w, \ldots$ to denote the variables, and the letters $a, b, c,\ldots$ to denote atomic propositions i.e., members of $A$.  Note that $A$ is finite and non-empty. These conditions on the set of atomic propositions are relevant because they affect the validity of certain theorems, particularly those dealing with independence and $\omega$-completeness as we shall see in the next chapter. We call this signature $\sig$.

The set $\mathbb{T}(\sig)$ of all terms over signature $\sig$ is defined as the smallest set such that
\begin{itemize}
\item $A\subseteq \mathbb{T}(\sig)$
\item $V\subseteq \mathbb{T}(\sig)$
\item $T,F\in \mathbb{T}(\sig)$
\item $t\lef r\rig s\in \mathbb{T}(\sig)$ with $t,r,s\in\mathbb{T}(\sig)$
\end{itemize}
Subsequently, the set $T(\sig)$ of all closed terms over signature $\sig$ is defined as the largest subset of $\mathbb{T}(\sig)$ such that no term in $T(\sig)$ contains a variable. The letters $t, r, s,\ldots$ are used to denote members of $\mathbb{T}(\sig)$ whereas we use $P, Q, R,\ldots$ to denote closed terms i.e., members of $T(\sig)$.

We postpone the discussion of the actual model construction until the next two sections. For now it suffices to recall the interpretation we offered in the introduction i.e., $T$ and $F$ are interpreted as true and false, respectively, and $\_\lef\_\rig\_$ is the if-then-else operator, where the middle argument is the antecedent and the left-most and right-most argument are the consequents.

%In this section we formally define the sequential proposition algebra. This algebra will make use of a set $A$ of atomic propositions. The letters $a,b,c,\ldots$ are used to denote the individual atomic propositions. During the evaluation of a term in this logic the value of these atomic propositions will depend on each other, and will change depending on where they are located. Phrased differently, the act of evaluating an atomic proposition has the side effect of influencing the values of forthcoming atomic propositions.

%In addition to the set $A$, the signature of this algebra contains the constants $T$ and $F$ corresponding to true and false, respectively, and a ternary operator $\_\lef\_\rig\_$. We call this ternary operator \emph{conditional composition}, and the intended interpretation of $x\lef y\rig z$ is \emph{if $y$ then $x$ else $z$}. Also keep in mind that since this is a sequential algebra, when interpreting an expression we assume an order in which the expression is evaluated i.e., when evaluating $x\lef y\rig z$ one starts with evaluating $y$ and then proceeds with either $x$ or $z$. The notation for conditional composition was introduced by Hoare in \cite{}. We call the signature consisting of $A$, $T$, $F$ and conditional composition, \sig.

Using the language we just introduced it is now possible to give the following axiomatization:
\[
\begin{array}{lrcl}
(\CP1)\qquad& x\lef T\rig y &=&x\\
(\CP2)\qquad& x\lef F\rig y &=&y\\
(\CP3)\qquad& T\lef x\rig F &=&x\\
(\CP4)\qquad& x\lef(y\lef z\rig u)\rig v &=& (x\lef y\rig v)\lef z\rig(x\lef u\rig v)
\end{array}
\]
We call this set of axioms CP. Hence when $t=r$, for terms $t$ and $r$ from signature \sig, can be derived from CP, we denote this as
\[
CP\vdash t=r
\]
Henceforth we often omit the ``$CP\vdash$'' part except when it is not clear from the context which set of axioms is used.

It is also important to note that the equality is in fact a congruence. Consequently, equality has besides the usual properties of reflexivity, symmetry, and transitivity
\[
\begin{array}{ll}
\text{(reflexivity)}\qquad & x = x\\
\text{(symmetry)}\qquad & x = y \rightarrow y=x\\
\text{(transitivity)}\qquad & x=y\wedge y=z \rightarrow x=z\\
\end{array}
\]
also the congruence property, which in this case will have the following form
\[
\begin{array}{lrcl}
\text{(congruence)}\qquad
\begin{array}{c}
x_1=y_1\quad x_2=y_2\quad x_3=y_3\\
\hline
x_1\lef x_2\rig x_3 = y_1\lef y_2\rig y_3
\end{array}
\end{array}
\]

Using this axiomatization and its intended interpretation, we can define versions of the classical connectives.
\[
\begin{array}{rclcrcl}
x\leftand y & = & y\lef x\rig F &\qquad& x\leftimp y & = & y\lef x\rig T\\
x\rightand y & = & x\lef y\rig F &\qquad& x\rightimp y & = & T\lef y\rig\neg x\\
x\leftor y & = & T\lef x\rig y &\qquad& x\leftbiimp y & = & y\lef x\rig\neg y\\
x\rightor y & = & T\lef y\rig x &\qquad& x\rightbiimp y & = & x\lef y\rig\neg x
\end{array}
\]
where $\neg x=F\lef x\rig T$. The notation of the operators is due to Bergstra, Bethke, and Rodenburg in \cite{connectives}. The circle in the connective indicates the order in which the expression is evaluated e.g., $x\leftor y$ indicates that we evaluate $x$ before looking at $y$. In the following sections we show that although e.g. $\leftand$ and $\rightand$ have the same interpretation in classical logic, they are not provably equal in CP. Also properties such as idempotency, commutativity, distributivity and absorption are not derivable in CP. However, there are some classical properties that are derivable in CP. For example we have the following property,
\begin{align*}
x\lef(F\lef y\rig T)\rig z
&=(x\lef F\rig z)\lef y\rig(x\lef T\rig z)\\
&=z\lef y\rig x
\end{align*}
which implies that 
\[
\neg\neg x = F\lef (F\lef x\rig T)\rig T = T\lef x\rig F = x
\]
Another example is based on De Morgan's laws,
\begin{align*}
\neg(x\rightor y)
&= F\lef(T\lef y\rig x)\rig T\\
&= (F\lef T\rig T)\lef y\rig(F\lef x\rig T)\\
&= F\lef y\rig\neg x\\
&= \neg x\lef \neg y\rig F\\
&= \neg x\rightand\neg y
\end{align*}

Using conditional composition, we can create the notion of sequential composition, denoted by $\circ$ i.e.,
\[
x\circ y = y\lef x\rig y
\]
By axiom CP4 it follows that sequential composition is associative,
\begin{align*}
x\circ(y\circ z) 
&= (z\lef y\rig z)\lef x\rig(z\lef y\rig z)\\
&= z\lef(y\lef x\rig y)\rig z\\
&= (x\circ y)\circ z
\end{align*}
In the following sections we not only give a model for the discussed axiomatization CP but also show that given the provided framework, it is easy to create variations on this model.

\section{Reactive valuations}
In the classic case a valuation determines the value of all the atomic propositions $a\in A$ i.e., it assigns either true or false to each atomic proposition. In the case of reactive valuations, this assignment can be dependent on atomic propositions previously evaluated. In this section we will formally define the notion of reactive valuations.

Let $B$ be the sort of boolean values with constants $T$ and $F$ and $RV$ be a sort of reactive valuations. Then for each $a\in A$ let there be a function
\[
y_a:RV\to B
\]
This function is called the yield of $a$ and it allows us to look up the value of $a$ using a specific reactive valuation. Furthermore, for each $a\in A$ there exists a function
\[
\dd a:RV\to RV
\]
called the $a$-derivative. With this function we can capture the dynamic nature of reactive valuations i.e., when we evaluate an atomic proposition $a$ the current reactive valuation can change. It is important to note that the elements in $RV$ do not just encode the value of the individual atomic propositions but also keep a history of atomic propositions previously evaluated. It is therefore possible that two reactive valuations $H$ and $H'$ assign the same values to each atomic proposition but $\dd a(H)\ne \dd a(H')$.

We define the signature $\Sigma_{ReVal}(A)$ to consist of the sorts $B$, $RV$, functions $y_a$ and $\dd a$ for each $a\in A$, and constants $T_{RV}$ and $F_{RV}$ of sort $RV$.

A structure $\mathbb{A}$ over $\Sigma_{ReVal}(A)$ is called a \emph{reactive valuation algebra} (RVA) if the following axioms are satisfied
\begin{align*}
y_a(T_{RV}) &= T\\
y_a(F_{RV}) &= F\\
\dd a(T_{RV}) &= T_{RV}\\
\dd a(F_{RV}) &= F_{RV}\\
\end{align*}
for each $a\in A$. So the constants $T_{RV}$ and $F_{RV}$ assign either $T$ to all the atomic propositions or $F$ to all the atomic propositions, respectively. Furthermore, these two valuations do not change while evaluating an expression.

The value of a proposition $P$ from signature $\sig$ according to a reactive valuation $H$ in a RVA $\mathbb{A}$ is denoted by
\[
P/H
\]
This value is determined as follows: for $a\in A$,
\begin{align*}
	T/H &= T\\
	F/H &= F\\
	a/H &= y_a(H)\\
	(P\lef Q\rig R)/H &= 
	\begin{cases}
		P/\dd Q(H)&\text{if $Q/H=T$}\\
		R/\dd Q(H)&\text{if $Q/H=F$}
	\end{cases}
\end{align*}
where $\dd P(H)$ is a generalized notion of the function $\dd a(H)$, and is defined as follows
\begin{align*}
	\dd T(H) &= H\\
	\dd F(H) &= H\\
	\dd{(P\lef Q\rig R)}(H) &=
	\begin{cases}
		\dd P(\dd Q(H))&\text{if $Q/H=T$}\\
		\dd R(\dd Q(H))&\text{if $Q/H=F$}\\
	\end{cases}
\end{align*}
There are a number of observations to be made here.

The propositions $T$ and $F$ are always evaluated as true and false, respectively, no matter which evaluation we use. Furthermore, evaluating $T$ and $F$ will not change the current valuation. So, it follows that $(a\lef T\rig b)/H=a/H$.

Also important to note is that in e.g. the proposition $a\lef (b\lef T\rig c)\rig d$ the values of $a$ and $d$ will not depend on $c$ as $c$ never gets evaluated.

Finally, let us look at a few examples.
\begin{align*}
((a\lef T\rig b)\lef a\rig c)/H
&= \begin{cases}
(a\lef T\rig b)/\dd a(H) &\text{if $a/H=T$}\\
c/\dd a(H) &\text{if $a/H=F$}\\
\end{cases}\\
&= \begin{cases}
a/\dd a(H) &\text{if $a/H=T$}\\
c/\dd a(H) &\text{if $a/H=F$}\\
\end{cases}\\
\end{align*}
Note that if we know that $a/H=T$ then it does not necessarily follow that $a/\dd a(H)$ is also true because the valuations $H$ and $\dd a(H)$ are different. To emphasize this point, look at propositions $a\lef a\rig a$ and $a$. Although, it certainly is true that in a classical setting these two are equivalent, it immediately follows that a valuation $H$ exists such that $(a\lef a\rig a)/H\ne a/H$.

The following example has instead of a constant or an atomic proposition as a condition, another conditional statement.
\begin{align*}
(a\lef(b\lef c\rig a)\rig F)/H
&=\begin{cases}
a/\dd{(b\lef c\rig a)}(H) &\text{if $(b\lef c\rig a)/H=T$}\\
F/\dd{(b\lef c\rig a)}(H) &\text{if $(b\lef c\rig a)/H=F$}\\
\end{cases}\\
&=\begin{cases}
a/\dd b(\dd c(H)) &\text{if $c/H=T$ and $b/\dd c(H)=T$}\\
a/\dd a(\dd c(H)) &\text{if $c/H=F$ and $a/\dd c(H)=T$}\\
F &\text{if $(b\lef c\rig a)/H=F$}\\
\end{cases}\\
&=\begin{cases}
a/\dd b(\dd c(H)) &\text{if $c/H=T$ and $b/\dd c(H)=T$}\\
a/\dd a(\dd c(H)) &\text{if $c/H=F$ and $a/\dd c(H)=T$}\\
F &\text{otherwise}\\
\end{cases}\\
\end{align*}
This example illustrates that the value of the leftmost $a$ does not only depend on $c$ being evaluated but also on the actual value of $c$ because this determines whether either $b$ or $a$ (the occurrence of $a$ right next to $c$ in the expression) is evaluated which in turn affects the value of the leftmost $a$.

\section{Reactive valuation varieties}
In the previous section we introduced the notion of reactive valuation algebra (RVA). In this section we define a number of specific classes of RVAs. Since the signature of all RVAs is the same, we refer to these classes as varieties. We define the following varieties of RVAs:
\begin{description}
\item[Free reactive valuations] This variety of RVAs consist of all possible RVAs. So there are no requirements posed on the RVAs in this variety other than that they are RVAs. Every other variety will be a subvariety of this one.
\item[Repetition-proof valuations] The variety with repetition-proof valuations consists of all RVAs that satisfy
\[
y_a(x) = y_a(\dd a(x))
\]
for all $a\in A$.
\item[Contractive valuations] The variety with contractive valuations is a subvariety of the variety with repetition-proof valuations i.e., every RVA in this variety will also be in the variety with repetition-proof valuations. In addition the RVAs here will satisfy
\[
\dd a(\dd a(x)) = \dd a(x)
\]
for all $a\in A$.
\item[Static valuations] The RVAs in the variety with static valuations satisfy the following equation
\[
y_a(\dd b(x))=y_a(x)
\]
for all $a,b\in A$.
\end{description}
The definitions of these varieties were taken from \cite{main}. In Appendix A we define and examine an additional variety of our own.

Given a variety $K$ we say that propositions $P$ and $Q$ are $K$-equivalent, which is denoted as
\[
P\equiv_K Q
\]
if $P/H=Q/H$ for all RVAs $\mathbb{A}$ in the variety $K$ and valuations $H\in\mathbb{A}$. Using the relation $\equiv_K$ we can define a congruence relation over propositions. We say that $P$ and $Q$ are $K$-congruent,
\[
P=_K Q
\]
if $=_K$ is the largest congruence contained in $\equiv_K$.

Given the four varieties we defined earlier we will use the abbreviations $fr$, $rp$, $cr$, $st$ for free, repetition-proof, contractive and static varieties, respectively

Bergstra and Ponse prove the following proposition.
\begin{prop}\label{variety relation prop}
$\equiv_{fr}\subsetneq\equiv_{rp}\subsetneq\equiv_{cr}\subsetneq\equiv_{st}$ and $=_K\subsetneq\equiv_K$ for $K\in\{fr,rp,cr\}$
\end{prop}
The first part of this proposition and  the differences between the varieties will become apparent in the following sections. The second part is best demonstrated using an example. If we take the term $F\lef a\rig F$ then clearly
\[
(F\lef a\rig F)/H=F/H
\] 
for all $H\in\mathbb{A}$ and $\mathbb{A}\in K$. However, it is not the case that for all varieties $K$ the following holds
\[
(b\lef (F\lef a\rig F)\rig b)/H=(b\lef F\rig b)/H
\]
because in the left-hand side the value of $b$ will depend on $a$ but on the right-hand side this dependency is gone. This means that although for all varieties $K$ we have $F\lef a\rig F\equiv_K F$, it does not follow that $F\lef a\rig F=_K F$. In the section on static valuations, we show that congruence and equivalence do happen to coincide for that variety.

The following proposition clarifies the relationship between congruence and equivalence for arbitrary variety $K$.
\begin{prop}\label{equiv congr prop}
If $P\equiv_KQ$ and for all $\mathbb{A}\in K$ and $H\in\mathbb{A}$, 
\[
\dd P(H)=\dd Q(H)
\]
then
\[
P=_KQ
\]
\end{prop}
\begin{proof}
Assume $P\equiv_KQ$ and $\forall\mathbb{A}\in K~\forall H\in\mathbb{A}$, $\dd P(H)=\dd Q(H)$. Since $=_K$ is defined as the largest congruence contained in $\equiv_K$, $P=_KQ$ if for all closed terms $S$ and $R$ the following three cases are true:
\begin{itemize}
\item[(1)] $P\lef S\rig R\equiv_K Q\lef S\rig R$
\item[(2)] $S\lef P\rig R\equiv_K S\lef Q\rig R$
\item[(3)] $S\lef R\rig P\equiv_K S\lef R\rig Q$
\end{itemize}
We continue by proving these three cases.

Since $P\equiv_KQ$, it follows that $\forall\mathbb{A}\in K~\forall H\in\mathbb{A}$, $P/H=Q/H$. Consequently, $P/\dd S(H)=Q/\dd S(H)$. Hence,
\begin{align*}
(P\lef S\rig R)/H
&=\begin{cases}
P/\dd S(H) & \text{if $S/H=T$}\\
R/\dd S(H) & \text{if $S/H=F$}
\end{cases}\\
&=\begin{cases}
Q/\dd S(H) & \text{if $S/H=T$}\\
R/\dd S(H) & \text{if $S/H=F$}
\end{cases}\\
&=(Q\lef S\rig R)/H
\end{align*}
So case (1) is true. Furthermore, the argument for case (3) is symmetric to the one give here. So case (3) is also true.

By assumption we know that $P/H=Q/H$ and $\dd P(H)=\dd Q(H)$. Thus,
\begin{align*}
(S\lef P\rig R)/H
&=\begin{cases}
S/\dd P(H) & \text{if $P/H=T$}\\
R/\dd P(H) & \text{if $P/H=F$}
\end{cases}\\
&=\begin{cases}
S/\dd Q(H) & \text{if $Q/H=T$}\\
R/\dd Q(H) & \text{if $Q/H=F$}
\end{cases}\\
&=(S\lef Q\rig R)/H
\end{align*}
Consequently, case (2) also holds, and $P=_KQ$.
\end{proof}

In the following sections we will further discuss the varieties we defined here. This discussion will include proper axiomatizations, and proofs of soundness and completeness.

\section{Notation and conventions}
Before continuing with the in-depth discussion of the varieties, we recap and introduce additional notation and conventions. We have encountered the following equality relations thus far:
\begin{itemize}
\item $\equiv_K$ denotes semantic equivalence with respect to variety $K$.
\item $=_K$ is the largest congruence contained in $\equiv_K$.
\item Plain $=$ is used to denote three different types of congruences. The first type is provable equality e.g. $CP\vdash x\lef T\rig y=x$. However we often omit the ``$CP\vdash$'' part if it is clear from the context which axiomatization we use. We also use $=$ in the interpretation of terms given some valuation $H$ e.g. $(P\lef F\rig Q)/H=Q/H$. Finally we use $=$ for equality between valuations e.g. $\dd T(H)=H$. Note that no ambiguity arises from these three different interpretations because they deal with equality over three distinct classes of objects and it will be immediately clear from the arguments or the context how $=$ is used.
\end{itemize}
Absent from this list is syntactic equality. We therefore introduce the symbol $\syn$ for syntactic equality\footnote{Often $\equiv$ is used to denote syntactic equality. However, since $\equiv$ is already used for semantic equivalence, we opted to use $\syn$ in order to avoid confusion.}.

We have the following conventions concerning symbols:
\begin{itemize}
\item The letters $a, b, c, \ldots$ denote atomic propositions. $A$ is the set of all atomic propositions.
\item The letters $x, y, z, u, v, w, \ldots$ denote variables. $V$ is the set of all variables.
\item The capital letters $P, Q, R,\ldots$ denote closed terms. $T(\sig)$ is the set of all closed terms from signature $\sig$.
\item The letters $t, s, r,\ldots$ denote terms that can possibly, but not necessarily, be open. $\mathbb{T}(\sig)$ is the set of all terms.
\end{itemize}

\section{Free reactive valuations}
$fr$-Congruence is axiomatized by CP. The variety with free reactive valuations is the variety on which all other varieties are based. This does not mean that there are no limitations. For example if we take the term $a\lef b\rig a$ it is not possible to distinguish between the two $a$'s i.e., we are forced to give them the same value no matter what valuation we choose. This limitation can be found in the definition of RVA where we define $y_a:RV\to B$ and $\dd a:RV\to RV$ as functions instead of relations.

\subsection{Soundness}
We have claimed that $fr$-congruence is axiomatized by CP. We, however, have not yet proven that the resulting theory is sound and complete with respect to our model. In this section we will show that CP is sound. Soundness of a theory with respect to a model means that whatever we derive from the axiomatization of the theory is also true in that model.
\begin{thm}\label{soundness cp}
For all closed terms $P$ and $Q$, 
\[
CP\vdash P=Q\qquad\Longrightarrow\qquad P=_{fr}Q
\]
\end{thm}
\begin{proof}
It suffices to show that the four axioms CP1, CP2, CP3 and CP4 are sound with respect to the variety with free reactive valuations. Let RVA $\mathbb{A}$ from the variety $fr$ with valuation $H\in\mathbb{A}$ be given. Then, starting with CP1 we proceed as follows.
\begin{align*}
(P\lef T\rig Q)/H
&=\begin{cases}
P/\dd T(H)&\text{if $T/H=T$}\\
Q/\dd T(H)&\text{if $T/H=F$}
\end{cases}\\
&=\begin{cases}
P/\dd T(H)&\text{if $T=T$}\\
Q/\dd T(H)&\text{if $T=F$}
\end{cases}\\
&=P/\dd T(H)\\
&=P/H
\end{align*}
We use the semantics we defined in the previous sections to evaluate the left-hand side of CP1 with an arbitrary valuation $H$, and end up with the right-hand side. In addition, observe that we do not pose any requirements on $H$. Consequently, the above derivation holds for all RVAs $\mathbb{A}$ and $H\in\mathbb{A}$. 
Furthermore, we also have the following:
\begin{align*}
\dd{(P\lef T\rig Q)}(H)
&=\begin{cases}
\dd P(\dd T(H)) & \text{if $T/H=T$}\\
\dd Q(\dd T(H)) & \text{if $T/H=F$}
\end{cases}\\
&=\dd P(H)
\end{align*}
Thus, by Proposition \ref{equiv congr prop}, we have proven that CP1 is sound.

Using the same strategy we prove that axioms CP2, CP3 and CP4 are sound.
\begin{align*}
(P\lef F\rig Q)/H
&=\begin{cases}
P/\dd F(H)&\text{if $F/H=T$}\\
Q/\dd F(H)&\text{if $F/H=F$}
\end{cases}\\
&=\begin{cases}
P/\dd F(H)&\text{if $F=T$}\\
Q/\dd F(H)&\text{if $F=F$}
\end{cases}\\
&=Q/\dd F(H)\\
&=Q/H
\end{align*}
\begin{align*}
\dd{(P\lef F\rig Q)}(H)
&=\begin{cases}
\dd P(\dd F(H)) & \text{if $F/H=T$}\\
\dd Q(\dd F(H)) & \text{if $F/H=F$}
\end{cases}\\
&=\dd Q(H)
\end{align*}

\begin{align*}
(T\lef P\rig F)/H
&=\begin{cases}
T/\dd P(H) & \text{if $P/H=T$}\\
F/\dd P(H) & \text{if $P/H=F$}\\
\end{cases}\\
&=\begin{cases}
T & \text{if $P/H=T$}\\
F & \text{if $P/H=F$}\\
\end{cases}\\
&=P/H
\end{align*}
\begin{align*}
\dd{(T\lef P\rig F)}(H)
&=\begin{cases}
\dd T(\dd P(H)) & \text{if $P/H=T$}\\
\dd F(\dd P(H)) & \text{if $P/H=F$}
\end{cases}\\
&=\begin{cases}
\dd P(H) & \text{if $P/H=T$}\\
\dd P(H) & \text{if $P/H=F$}
\end{cases}\\
&=\dd P(H)
\end{align*}
Showing that axiom CP4 is sound, is a bit more complicated than the previous three axioms because of the number of cases involved.
\begin{align*}
(P\lef(Q\lef R\rig S)\rig V)/H 
&=\begin{cases}
P/\dd{(Q\lef R\rig S)}(H) & \text{if $(Q\lef R\rig S)/H=T$}\\
V/\dd{(Q\lef R\rig S)}(H) & \text{if $(Q\lef R\rig S)/H=F$}\\
\end{cases}\\\\
&=\begin{cases}
P/\dd{(Q\lef R\rig S)}(H) & \text{if $Q/\dd R(H)=T$ and $R/H=T$}\\
P/\dd{(Q\lef R\rig S)}(H) & \text{if $S/\dd R(H)=T$ and $R/H=F$}\\
V/\dd{(Q\lef R\rig S)}(H) & \text{if $Q/\dd R(H)=F$ and $R/H=T$}\\
V/\dd{(Q\lef R\rig S)}(H) & \text{if $S/\dd R(H)=F$ and $R/H=F$}\\
\end{cases}\\\\
&=\begin{cases}
P/\dd Q(\dd R(H)) & \text{if $Q/\dd R(H)=T$ and $R/H=T$}\\
P/\dd S(\dd R(H)) & \text{if $S/\dd R(H)=T$ and $R/H=F$}\\
V/\dd Q(\dd R(H)) & \text{if $Q/\dd R(H)=F$ and $R/H=T$}\\
V/\dd S(\dd R(H)) & \text{if $S/\dd R(H)=F$ and $R/H=F$}\\
\end{cases}\\
\end{align*}
Since
\begin{align*}
(P\lef Q\rig V)/\dd R(H)
&= \begin{cases}
P/\dd Q(\dd R(H)) & \text{if $Q/\dd R(H)=T$}\\
V/\dd Q(\dd R(H)) & \text{if $Q/\dd R(H)=F$}
\end{cases}
\end{align*}
and
\begin{align*}
(P\lef S\rig V)/\dd R(H)
&= \begin{cases}
P/\dd S(\dd R(H)) & \text{if $S/\dd R(H)=T$}\\
V/\dd S(\dd R(H)) & \text{if $S/\dd R(H)=F$}
\end{cases}
\end{align*}
it is possible to e.g. replace $P/\dd S(\dd R(H))$ if $S/\dd R(H)=T$,
and $V/\dd S(\dd R(H))$ if $S/\dd R(H)=F$ by the expression $(P\lef
S\rig V)/\dd R(H)$. So continuing from where we left off,
\begin{align*}
&=\begin{cases}
(P\lef Q\rig V)/\dd R(H) & \text{if $R/H=T$}\\
(P\lef S\rig V)/\dd R(H) & \text{if $R/H=F$}\\
\end{cases}\\
&=((P\lef Q\rig V)\lef R\rig (P\lef S\rig V))/H
\end{align*}
\begin{align*}
\dd{(P\lef(Q\lef R\rig S)\rig V)}(H)
&=\begin{cases}
\dd P(\dd Q(\dd R(H))) & \text{if $R/H=T$ and $Q/\dd R(H)=T$}\\
\dd V(\dd Q(\dd R(H))) & \text{if $R/H=T$ and $Q/\dd R(H)=F$}\\
\dd P(\dd S(\dd R(H))) & \text{if $R/H=F$ and $S/\dd R(H)=T$}\\
\dd V(\dd S(\dd R(H))) & \text{if $R/H=F$ and $S/\dd R(H)=F$}
\end{cases}\\
&=\begin{cases}
\dd{(P\lef Q\rig V)}(\dd R(H)) & \text{if $R/H=T$}\\
\dd{(P\lef S\rig V)}(\dd R(H)) & \text{if $R/H=F$}\\
\end{cases}\\
&=\dd{((P\lef Q\rig V)\lef R\rig(P\lef S\rig V))}(H)
\end{align*}

So all axioms are sound with respect to $=_{fr}$.
\end{proof}

\subsection{Completeness}
In this section we prove completeness. The axiomatization $CP$ is complete with respect to the variety with free reactive valuations, if two closed terms are $fr$-congruent then these terms are also provably equal in $CP$.

Before proving completeness we first introduce basic forms, which are a class of closed terms. We will show that each closed term is provably equal to a basic form. The primary reason for introducing these basic forms is that they will greatly simplify most proofs by structural induction on closed terms because their structure is less complicated. This will be especially useful in proving completeness.

\begin{defn}\label{basic form def}
The set of basic forms $\BF$ is defined as the smallest set
such that $T,F\in\BF$, and if $P,Q\in\BF$ then $P\lef a\rig Q\in\BF$ for all atomic propositions $a\in A$.
\end{defn}

So e.g. $F\lef(a\lef b\rig c)\rig T$ is not a basic form but $(F\lef a\rig T)\lef b\rig(F\lef c\rig T)$ is. Similarly, $a$ is not a basic form but $T\lef a\rig F$ is. If conditional composition occurs in a basic form, the antecedent is always an atomic proposition.

An alternative way of looking at basic forms is to view them as labeled binary trees i.e., the basic form $P\lef a\rig Q$ corresponds to the tree
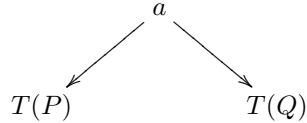
\[
\xymatrix{
&a\ar[dl]\ar[dr]\\
T(P)&&T(Q)
}
\]
where $T(P)$ and $T(Q)$ are the binary trees corresponding to $P$ and $Q$, respectively. Hence, the nodes of the tree consist of atomic propositions and the leaves of either $T$ or $F$. This illustrates the simplicity of basic forms because if we would try to similarly construct a binary tree for arbitrary closed terms, the nodes themselves would have to be trees because the antecedent of conditional composition occurring in such a term can itself be an arbitrary closed term.

As mentioned before we will prove that for each closed term there exists a basic form such that they are provably equal to each other.

\begin{lem}\label{basic form theorem}
For each closed term $P$ over signature $\sig$ there exists a
term $P'\in\BF$ such that $CP\vdash P=P'$.
\end{lem}
\begin{proof}
We proceed by structural induction on $P$. Suppose $P$ is either $T$ or $F$ then $P\in\BF$. Suppose $P\syn a$ for $a\in A$.\footnote{Recall that $\syn$ is used for syntactic equality.} Since $P=T\lef P\rig F$ and $T\lef P\rig F\in\BF$, $P'$ exists. Suppose $P$ is $P_1\lef P_2\rig P_3$. By the induction hypothesis there exist terms $P_1',P_2',P_3'\in\BF$ such that $P_1=P_1'$, $P_2=P_2'$, and $P_3=P_3'$. By congruence, it follows that $P=P_1'\lef P_2'\rig P_3'$. We show by structural induction on $P_2'$ that $P'$ exists.

If $P_2'\syn T$, then $P_1'\lef T\rig P_3'=P_1'$, and $P_1'$ is a basic form. Similarly, if $P_2'$ is $F$, then $P_1'\lef F\rig P_3'=P_3'$, and $P_3'\in\BF$. If $P_2'\syn P_{21}'\lef a\rig P_{22}'$ for some $a\in A$ then
\begin{align*}
P_1'\lef (P_{21}'\lef a\rig P_{22}')\rig P_3' &=
(P_1'\lef P_{21}'\rig P_3')\lef a\rig (P_1'\lef P_{22}'\rig P_3')\\
&=_{IH} V\lef a\rig W
\end{align*}
where $V,W\in\BF$ and $V=P_1'\lef P_{21}'\rig P_3'$, $W=P_1'\lef P_{22}'\rig P_3'$. Clearly, $V\lef a\rig W\in\BF$.
\end{proof}

The following lemma is needed in the Lemma \ref{completeness lemma} that shows that syntactic equality and $fr$-congruence coincide.
\begin{lem}\label{inv congr lemma}
For $P_1\lef a\rig P_2, Q_1\lef a\rig Q_2\in\BF$,
\[
P_1\lef a\rig P_2=_{fr}Q_1\lef a\rig Q_2\qquad\Longrightarrow\qquad P_1=_{fr} Q_1\wedge P_2=_{fr} Q_2
\]
\end{lem}
\begin{proof}
We prove the contraposition. Then either $P_1\ne_{fr}Q_1$ or $P_2\ne_{fr}Q_2$ implies $P_1\lef a\rig P_2\ne_{fr}Q_1\lef a\rig Q_2$. Assume without loss of generality that $P_1\ne_{fr}Q_1$. Then the following two cases can be distinguished.

In the first case, $P_1\not\equiv_{fr} Q_1$. Consequently, there exists an algebra $\mathbb{A}$ and valuation $H\in\mathbb{A}$ such that $P_1/H\ne Q_1/H$. Subsequently, we construct an algebra $\mathbb{A'}\supseteq\mathbb{A}$ with valuation $H'\in\mathbb{A'}$ such that $\dd a(H')=H$ and $y_a(H')=T$. Then
\begin{align*}
(P_1\lef a\rig P_2)/H'
&=\begin{cases}
P_1/\dd a(H') & \text{if $a/H'=T$}\\
P_2/\dd a(H') & \text{if $a/H'=F$}
\end{cases}\\
&=P_1/H\\
&\ne Q_1/H\\
&=(Q_1\lef a\rig Q_2)/H'
\end{align*}

In the second case, $P_1\equiv_{fr}Q_1$. So the congruence property does not apply to $P_1$ and $Q_1$. It follows that there are closed terms $S$ and $R$ such that one of the following is the case:
\begin{itemize}
\item[(1)] $S\lef R\rig P_1\not\equiv_{fr} S\lef R\rig Q_1$
\item[(2)] $P_1\lef S\rig R\not\equiv_{fr} Q_1\lef S\rig R$
\item[(3)] $S\lef P_1\rig R\not\equiv_{fr} S\lef Q_1\rig R$
\end{itemize}
In each of the three the cases there is an algebra $\mathbb{A}$ and valuation $H\in\mathbb{A}$ such that the left-hand side and the right-hand side are not equal using  valuation $H$. Using this valuation $H$ we know that in case (1) the following applies:
\begin{align*}
(S\lef R\rig P_1)/H
&=\begin{cases}
S/\dd R(H) & \text{if $R/H=T$}\\
P_1/\dd R(H) & \text{if $R/H=F$}
\end{cases}\\
&\ne\begin{cases}
S/\dd R(H) & \text{if $R/H=T$}\\
Q_1/\dd R(H) & \text{if $R/H=F$}
\end{cases}\\
&=(S\lef R\rig Q_1)/H
\end{align*}
Consequently, $P_1/\dd R(H)\ne Q_1/\dd R(H)$. However, this implies that $P_1\not\equiv_{fr} Q_1$. Since we already assumed that $P_1\equiv_{fr} Q_1$, we have a contradiction. Hence, case (1) cannot occur. Using a similar argument, we can show that this also applies to case (2).

This leaves us with case (3). Let $\mathbb{A}$ and $H$ be defined as before. Then, similarly to the case where $P_1\not\equiv_{fr}Q_1$, we construct a new algebra $\mathbb{A'}\supseteq\mathbb{A}$ and valuation $H'\in\mathbb{A'}$ such that $\dd a(H')=H$ and $y_a(H')=T$. Consequently,
\begin{align*}
(S\lef(P_1\lef a\rig P_2)\rig R)/H'
&=(S\lef P_1\rig R)/H\\
&\ne(S\lef Q_1\rig R)/H\\
&=(S\lef(Q_1\lef a\rig Q_2)\rig R)/H'
\end{align*}
So the congruence property does not hold for $P_1\lef a\rig P_2$ and $Q_1\lef a\rig Q_2$, and thus these terms are not congruent to each other, $P_1\lef a\rig P_2\ne_{fr}Q_1\lef a\rig Q_2$.
\end{proof}
It is perhaps interesting to observe that the other direction of the previous lemma follows from congruence i.e., if we know that $P_1=_{fr}Q_1$ and $P_2=_{fr}Q_2$ then $P_1\lef a\rig P_2=_{fr}Q_1\lef a\rig Q_2$.

The next lemma shows that syntactic equality and $fr$-congruence coincide.
\begin{lem}\label{completeness lemma}
For $P,Q\in\BF$,
\[
P=_{fr}Q\qquad\Longleftrightarrow\qquad P\syn Q
\]
\end{lem}
\begin{proof}
The direction from syntactic equality to $fr$-congruence is trivial. The other direction is proven by taking the contraposition and then proceeding by structural induction on both $P$ and $Q$. So assume $P\nsyn Q$. We omit the trivial cases and the cases that follow by symmetry. 

Suppose $P\syn T$ and $Q\syn Q_1\lef a\rig Q_2$ for some $a\in A$. If $P=_{fr}Q$ then $P\circ a=_{fr}Q\circ a$ by congruence. However,
\begin{align*}
(P\circ a)/H &= (a\lef T\rig a)/H\\
&=\begin{cases}
a/\dd T(H) & \text{if $T/H=T$}\\
a/\dd T(H) & \text{if $T/H=F$}
\end{cases}\\
&= a/H\\
&\ne\begin{cases}
a/\dd{Q_1}(\dd a(H)) & \text{if $a/H=T$}\\
a/\dd{Q_2}(\dd a(H)) & \text{if $a/H=F$}
\end{cases}\\
&= (a\lef (Q_1\lef a\rig Q_2)\rig a)/H\\
&= (Q\circ a)/H
\end{align*}
Note that we can always construct an algebra $\mathbb{A}$ and $H\in\mathbb{A}$ such that neither $a/H=a/\dd{Q_1}(\dd a(H))$ nor $a/H=a/\dd{Q_2}(\dd a(H))$.

Suppose $P\syn P_1\lef a\rig P_2$ and $Q\syn Q_1\lef a\rig Q_2$. Then we can assume without loss of generality that $P_1\nsyn Q_1$. By I.H., it follows that $P_1\ne_{fr}Q_1$. By Lemma \ref{inv congr lemma}, $P\ne_{fr}Q$.

Of course, if $P\syn P_1\lef a\rig P_2$ and $Q\syn Q_1\lef b\rig Q_2$, we can simply pick an algebra $\mathbb{A}$ and valuation $H\in\mathbb{A}$ such that $(P\circ a)/H\ne (Q\circ a)/H$.
\end{proof}

Now the stage has been set to prove completeness for not just basic forms but for all closed terms.

\begin{thm}\label{completeness thm}
For closed terms $P$ and $Q$,
\[
P=_{fr}Q\qquad\Longrightarrow\qquad CP\vdash P=Q
\]
\end{thm}
\begin{proof}
Let $P=_{fr}Q$. By Lemma \ref{basic form theorem}, there exist
terms $P',Q'\in\BF$ such that $P=P'$ and $Q=Q'$.
Furthermore, by soundness, it follows that $P=_{fr}P'$ and $Q=_{fr}Q'$, and
thus $P'=_{fr}Q'$. Finally, by Lemma \ref{completeness lemma}, we get
$P'=Q'$ which implies that $P=Q$.
\end{proof}

\section{Repetition-proof valuations}
Recall that the variety with repetition-proof valuations is characterized by the following equation:
\[
y_a(x)=y_a(\dd a(x))
\]
This restricts the type of valuations we allow in this variety. The consequences of introducing this restriction are perhaps best explained using an example. Take a look at the following evaluation of the term $(b\lef a\rig c)\lef a\rig d$ using a repetition-proof valuation $H$:
\begin{align*}
((b\lef a\rig c)\lef a\rig d)/H
&=\begin{cases}
b/\dd a(\dd a(H)) & \text{if $y_a(H)=T$ and $y_a(\dd a(H))=T$}\\
c/\dd a(\dd a(H)) & \text{if $y_a(H)=T$ and $y_a(\dd a(H))=F$}\\
d/\dd a(H) & \text{if $y_a(H)=F$}
\end{cases}\\
&=\begin{cases}
b/\dd a(\dd a(H)) & \text{if $y_a(H)=T$ and $y_a(\dd a(H))=T$}\\
d/\dd a(H) & \text{if $y_a(H)=F$}
\end{cases}\\
&=\begin{cases}
b/\dd a(\dd a(H)) & \text{if $y_a(H)=T$ and $y_a(\dd a(H))=T$}\\
e/\dd a(\dd a(H)) & \text{if $y_a(H)=T$ and $y_a(\dd a(H))=F$}\\
d/\dd a(H) & \text{if $y_a(H)=F$}
\end{cases}\\
&=((b\lef a\rig e)\lef a\rig d)/H
\end{align*}
Observe that since $y_a(H)=y_a(\dd a(H))$ it follows that the case where $y_a(H)=T$ and $y_a(\dd a(H))=F$ never occurs. Thus during the evaluation of $(b\lef a\rig c)\lef a\rig d$ the $c$ is never evaluated, and thus we can replace $c$ with any term we like which in this case is another atomic proposition $e$. However, this does not mean that $(b\lef a\rig c)\lef a\rig d=_{rp}b\lef a\rig d$ because the evaluation of $b$ depends on both $a$'s.

Similar to the free reactive valuations, we can define a corresponding axiomatization. In this case, the axiomatization consists of CP plus the following two axiom schemas,
\begin{align*}
(\CPrp1)\qquad(x\lef a\rig y)\lef a\rig z &= (x\lef a\rig x)\lef a\rig z\\
(\CPrp2)\qquad x\lef a\rig (y\lef a\rig z) &= x\lef a\rig (z\lef a\rig z)
\end{align*}
for all $a\in A$. We call the entire axiomatization $\text{CP}_{rp}$.
The axioms \CPrp1 and \CPrp2 combined with CP tell us that the value of an atomic proposition $a$ does not change unless there is another proposition in between them.

An example of repetition-proof behaviour can be found in programming. For example, an atomic proposition corresponds with a function that updates a global variable but its output does not depend on this variable.

As in the previous section we proceed by proving soundness and completeness, starting with soundness.

\subsection{Soundness}

\begin{thm}
For closed terms $P$ and $Q$,
\[
\text{CP}_{rp}\vdash P=Q\qquad\Longrightarrow\qquad P=_{rp}Q
\]
\end{thm}
\begin{proof}
According to Proposition \ref{variety relation prop}, we know that $=_{fr}\subseteq =_{rp}$. Since we already checked the soundness of the axioms in CP in the proof for soundness of free reactive valuations, it suffices to show soundness for $\CPrp1$ and $\CPrp2$, starting with CPrp1. Let $\mathbb{A}$ be a RVA from variety $rp$ and valuation $H\in\mathbb{A}$.
\begin{align*}
((P\lef a\rig Q)\lef a\rig R)/H
&=\begin{cases}
(P\lef a\rig Q)/\dd a(H) & \text{if $a/H=T$}\\
R/\dd a(H) & \text{if $a/H=F$}\\
\end{cases}\\
&=\begin{cases}
P/\dd a(\dd a(H)) & \text{if $a/H=T$ and $a/\dd a(H)=T$}\\
Q/\dd a(\dd a(H)) & \text{if $a/H=T$ and $a/\dd a(H)=F$}\\
R/\dd a(H) & \text{if $a/H=F$}\\
\end{cases}\\
&=\begin{cases}
P/\dd a(\dd a(H)) & \text{if $a/H=T$ and $a/H=T$}\\
Q/\dd a(\dd a(H)) & \text{if $a/H=T$ and $a/H=F$}\\
R/\dd a(H) & \text{if $a/H=F$}\\
\end{cases}\\
&=\begin{cases}
P/\dd a(\dd a(H)) & \text{if $a/H=T$ and $a/H=T$}\\
P/\dd a(\dd a(H)) & \text{if $a/H=T$ and $a/H=F$}\\
R/\dd a(H) & \text{if $a/H=F$}\\
\end{cases}\\
&=\begin{cases}
P/\dd a(\dd a(H)) & \text{if $a/H=T$ and $a/\dd a(H)=T$}\\
P/\dd a(\dd a(H)) & \text{if $a/H=T$ and $a/\dd a(H)=F$}\\
R/\dd a(H) & \text{if $a/H=F$}\\
\end{cases}\\
&=\begin{cases}
(P\lef a\rig P)/\dd a(H) & \text{if $a/H=T$}\\
R/\dd a(H) & \text{if $a/H=F$}\\
\end{cases}\\
&=((P\lef a\rig P)\lef a\rig R)/H
\end{align*}
\begin{align*}
\dd{((P\lef a\rig Q)\lef a\rig R)}(H)
&=\begin{cases}
\dd P(\dd a(\dd a(H))) & \text{if $a/H=T$ and $a/\dd a(H)=T$}\\
\dd Q(\dd a(\dd a(H))) & \text{if $a/H=T$ and $a/\dd a(H)=F$}\\
\dd R(\dd a(H)) & \text{if $a/H=F$}
\end{cases}\\
&=\begin{cases}
\dd P(\dd a(\dd a(H))) & \text{if $a/H=T$}\\
\dd R(\dd a(H)) & \text{if $a/H=F$}
\end{cases}\\
&=\begin{cases}
\dd P(\dd a(\dd a(H))) & \text{if $a/H=T$ and $a/\dd a(H)=T$}\\
\dd P(\dd a(\dd a(H))) & \text{if $a/H=T$ and $a/\dd a(H)=F$}\\
\dd R(\dd a(H)) & \text{if $a/H=F$}
\end{cases}\\
&=\dd{((P\lef a\rig P)\lef a\rig R)}(H)
\end{align*}
By Proposition \ref{equiv congr prop} CPrp1 is sound. Next we show soundness for CPrp2:
\begin{align*}
(P\lef a\rig (Q\lef a\rig R))/H
&=\begin{cases}
P/\dd a(H) & \text{if $y_a(H)=T$}\\
(Q\lef a\rig R)/\dd a(H) & \text{if $y_a(H)=F$}\\
\end{cases}\\
&=\begin{cases}
P/\dd a(H) & \text{if $y_a(H)=T$}\\
Q/\dd a(\dd a(H)) & \text{if $y_a(H)=F$ and $y_a(\dd a(H))=T$}\\
R/\dd a(\dd a(H)) & \text{if $y_a(H)=F$ and $y_a(\dd a(H))=F$}\\
\end{cases}\\
&=\begin{cases}
P/\dd a(H) & \text{if $y_a(H)=T$}\\
Q/\dd a(\dd a(H)) & \text{if $y_a(H)=F$ and $y_a(H)=T$}\\
R/\dd a(\dd a(H)) & \text{if $y_a(H)=F$ and $y_a(H)=F$}\\
\end{cases}\\
&=\begin{cases}
P/\dd a(H) & \text{if $y_a(H)=T$}\\
R/\dd a(\dd a(H)) & \text{if $y_a(H)=F$ and $y_a(H)=T$}\\
R/\dd a(\dd a(H)) & \text{if $y_a(H)=F$ and $y_a(H)=F$}\\
\end{cases}\\
&=\begin{cases}
P/\dd a(H) & \text{if $y_a(H)=T$}\\
R/\dd a(\dd a(H)) & \text{if $y_a(H)=F$ and $y_a(\dd a(H))=T$}\\
R/\dd a(\dd a(H)) & \text{if $y_a(H)=F$ and $y_a(\dd a(H))=F$}\\
\end{cases}\\
&=\begin{cases}
P/\dd a(H) & \text{if $y_a(H)=T$}\\
(R\lef a\rig R)/\dd a(H) & \text{if $y_a(H)=F$}\\
\end{cases}\\
&=(P\lef a\rig (R\lef a\rig R))/H
\end{align*}
\begin{align*}
\dd{(P\lef a\rig(Q\lef a\rig R))}(H)
&=\begin{cases}
\dd P(\dd a(H)) & \text{if $a/H=T$}\\
\dd Q(\dd a(\dd a(H))) & \text{if $a/H=F$ and $a/\dd a(H)=T$}\\
\dd R(\dd a(\dd a(H))) & \text{if $a/H=F$ and $a/\dd a(H)=F$}\\
\end{cases}\\
&=\begin{cases}
\dd P(\dd a(H)) & \text{if $a/H=T$}\\
\dd R(\dd a(\dd a(H))) & \text{if $a/H=F$}\\
\end{cases}\\
&=\begin{cases}
\dd P(\dd a(H)) & \text{if $a/H=T$}\\
\dd R(\dd a(\dd a(H))) & \text{if $a/H=F$ and $a/\dd a(H)=T$}\\
\dd R(\dd a(\dd a(H))) & \text{if $a/H=F$ and $a/\dd a(H)=F$}\\
\end{cases}\\
&=\dd{(P\lef a\rig(R\lef a\rig R))}(H)
\end{align*}
So CPrp2 is also sound.
\end{proof}

\subsection{Completeness}
Similar to the previous section we define a set of basic forms for this variety. Since we are working with a different variety the set of basic forms needs to change. If we were to use the set \BF\ i.e., the set of basic forms as defined in the section on free reactive valuations, as the basic forms of this variety then syntactic equality and $rp$-congruence would not coincide. For example, let the terms $(P\lef a\rig Q)\lef a\rig R$ and $(P\lef a\rig P)\lef a\rig R$ be in \BF\ and let $P\nsyn Q$ then these terms are $rp$-congruent but not syntactically equal. So we need to define a new set of basic forms.
\begin{defn}
The set of repetition-proof basic forms is the smallest set $\BF_{rp}$ such that $T,F\in\BF_{rp}$ and if $P,Q\in\BF_{rp}$ then
\begin{itemize}
\item if $P\syn P_1\lef a\rig P_2$ and $Q\syn Q_1\lef a\rig Q_2$ then $(a\circ P_1)\lef a\rig(a\circ Q_2)\in\BF_{rp}$
\item if $P\syn P_1\lef a\rig P_2$ and $Q\nsyn Q_1\lef a\rig Q_2$ then $(a\circ P_1)\lef a\rig Q\in\BF_{rp}$
\item if $P\nsyn P_1\lef a\rig P_2$ and $Q\syn Q_1\lef a\rig Q_2$ then $P\lef a\rig(a\circ Q_2)\in\BF_{rp}$
\item if $P\nsyn P_1\lef a\rig P_2$ and $Q\nsyn Q_1\lef a\rig Q_2$ then $P\lef a\rig Q\in\BF_{rp}$
\end{itemize}
for all $a\in A$.
\end{defn}
Clearly, the set $\BF_{rp}$ is a subset of $\BF$. The four cases mentioned in the definition are based on the axioms \CPrp1 and \CPrp2. 
\begin{lem}
For each closed term $P$ there exists a term $P'\in\BF_{rp}$ such
that $\text{CP}_{rp}\vdash P=P'$.
\end{lem}
\begin{proof}
We prove this theorem by structural induction on $P$. By Lemma \ref{basic form theorem} it follows that we can assume without loss of generality that $P$ is a basic form as defined in the section on free reactive valuations i.e., $P\in\BF$.

If $P$ is $T$ or $F$ then $P\in\BF_{rp}$. Suppose $P$ is $P_1\lef a\rig P_2$. By the induction hypothesis,
there exist terms $P_1',P_2'\in\BF_{rp}$ such that $P_1=P_1'$
and $P_2=P_2'$. Now suppose that $P_1'\syn P_{11}'\lef a\rig
P_{12}'$ and $P_2'\syn P_{21}'\lef a\rig P_{22}'$. Consequently,
\begin{align*}
P_1\lef a\rig P_2
&= P_1'\lef a\rig P_2'\\
&= (P_{11}'\lef a\rig P_{12}')\lef a\rig (P_{21}'\lef a\rig P_{22}')\\
&= (a\circ P_{11}')\lef a\rig (a\circ P_{22}')
\end{align*}
By definition of $\BF_{rp}$ we have $(a\circ P_{11}')\lef a\rig
(a\circ P_{22}')\in\BF_{rp}$.

Using similar reasoning we can show that there exists such a term
$P'\in\BF_{rp}$ for the remaining three cases:
\begin{itemize}
\item $P_1'\nsyn P_{11}'\lef a\rig P_{12}'$ and $P_2'\syn P_{21}'\lef a\rig P_{22}'$
\item $P_1'\syn P_{11}'\lef a\rig P_{12}'$ and $P_2'\nsyn P_{21}'\lef a\rig P_{22}'$
\item $P_1'\nsyn P_{11}'\lef a\rig P_{12}'$ and $P_2'\nsyn P_{21}'\lef a\rig P_{22}'$
\end{itemize}
\end{proof}

In the section on free reactive valuations we needed Lemmas \ref{inv congr lemma} and \ref{completeness lemma} in order to prove completeness. Similarly, we would like to prove these lemmas for this variety. However, observe that in the proofs of Lemmas \ref{inv congr lemma} and \ref{completeness lemma}, we construct a new valuation algebra based on another algebra. In the variety with free reactive valuations this is not a problem, but in this variety we have some restrictions on our RVAs, and thus cannot automatically assume that such a construction is possible. Therefore, in the proofs of the following two lemmas we focus on showing that such an algebra exists. We call an algebra from the variety with repetition-proof valuations an rp-algebra.

\begin{lem}
For $P_1\lef a\rig P_2,Q_1\lef a\rig Q_2\in\BF_{rp}$,
\[
P_1\lef a\rig P_2=_{rp}Q_1\lef a\rig Q_2\qquad\Longrightarrow\qquad P_1=_{rp}Q_1\wedge P_2=_{rp}Q_2
\]
\end{lem}
\begin{proof}
We prove the contraposition. We assume without loss of generality that $P_1\ne_{rp}Q_1$. Then either $P_1\not\equiv_{rp} Q_1$ or $P_1\equiv_{rp}Q_1$.

Suppose $P_1\not\equiv_{rp}Q_1$. Then there exists an rp-algebra $\mathbb{A}$ and valuation $H\in\mathbb{A}$ such that $P_1/H\ne Q_1/H$. We show that $P_1/H\ne Q_1/H$ holds whether $a/H=T$ or $a/H=F$. Consider the following four cases:
\begin{itemize}
\item Suppose $P_1\syn P_{11}\lef a\rig P_{12}$ and $Q_1\syn Q_{11}\lef a\rig Q_{12}$. Since $P_1\lef a\rig P_2,Q_1\lef a\rig Q_2\in\BF_{rp}$, it follows by definition of $\BF_{rp}$ that $P_{11}\syn P_{12}$ and $Q_{11}\syn Q_{12}$. Consequently, $P_1/H\ne Q_1/H$ whether $y_a(H)=T$ or $y_a(H)=F$.
\item Suppose $P_1\syn P_{11}\lef a\rig P_{12}$ and $Q_1\nsyn Q_{11}\lef a\rig Q_{12}$. By similar reasoning as before, we can conclude that $P_{11}\syn P_{12}$. Furthermore, the value of $Q_1/H$ does not depend on $a/H$. Consequently, $a/H$ can be either $T$ or $F$.
\item Suppose $P_1\nsyn P_{11}\lef a\rig P_{12}$ and $Q_1\syn Q_{11}\lef a\rig Q_{12}$. Argument is symmetric to the previous case.
\item Suppose $P_1\nsyn P_{11}\lef a\rig P_{12}$ and $Q_1\nsyn Q_{11}\lef a\rig Q_{12}$. Neither the value of $P_1/H$ nor that of $Q_1/H$ depends on $a/H$. Consequently, $P_1/H\ne Q_1/H$ is independent of the value of $a/H$.
\end{itemize}
Since $P_1/H\ne Q_1/H$ regardless of whether $a/H=T$ or $a/H=F$, we can assume without loss of generality that $a/H=T$. We construct an rp-algebra $\mathbb{A}'\supseteq\mathbb{A}$ with valuation $H'\in\mathbb{A'}$ such that $\dd a(H')=H$. Since $\mathbb{A'}$ is an rp-algebra we know that $y_a(x)=y_a(\dd a(x))$. Hence, it follows that $y_a(H')=T$ because if this were not the case then $y_a(\dd a(H'))=y_a(H)=a/H=F$ which is contrary to our assumption. It follows that $(P_1\lef a\rig P_2)/H'\ne(Q_1\lef a\rig Q_2)/H'$.

Suppose $P_1\equiv_{rp}Q_1$. Then the congruence property does not hold i.e., at least one of the following three cases is true,
\begin{itemize}
\item[(1)] $S\lef R\rig P_1\not\equiv_{rp} S\lef R\rig Q_1$
\item[(2)] $P_1\lef S\rig R\not\equiv_{rp} S\lef P_1\rig R$
\item[(3)] $S\lef P_1\rig R\not\equiv_{rp} S\lef P\rig R$
\end{itemize}
for closed terms $S$ and $R$. Using the same argument as in the proof of Lemma \ref{inv congr lemma}, it follows that cases (1) and (2) cannot occur. So suppose case (3) is true. Then there exists an rp-algebra $\mathbb{A}$ and valuation $H\in\mathbb{A}$ such that $(S\lef P_1\rig R)/H\ne(S\lef Q_1\rig R)/H$. Using similar reasoning as in the case for $P_1\not\equiv_{rp}Q_1$, we can assume without loss of generality that $a/H=T$. Thus we can construct an rp-algebra $\mathbb{A'}\supseteq\mathbb{A}$ and valuation $H'\in\mathbb{A'}$ such that $\dd a(H')=H$ and $y_a(H)=T$. Consequently, $(S\lef(P_1\lef a\rig P_2)\rig R)/H'\ne (S\lef(Q_1\lef a\rig Q_2)\rig R)/H'$. Thus the congruence property does not hold and $P_1\lef a\rig P_2\ne_{rp} Q_1\lef a\rig Q_2$.
\end{proof}

The following two lemmas have the perhaps odd condition that there are at least two atomic propositions. At the end of this section we examine what happens if there is only one atomic proposition. Note that by definition there is at least one atomic proposition i.e., $A$ is non-empty.

\begin{lem}\label{rp completeness lemma}
For $|A|\ge 2$ and $P,Q\in\BF_{rp}$,
\[
P=_{rp}Q\qquad\Longrightarrow\qquad P\syn Q
\]
\end{lem}
\begin{proof}
We use the same argument as in the proof of Lemma \ref{completeness lemma}. However, in the case of $P\syn T$ and $Q\syn Q_1\lef a\rig Q_2$, we claimed that we can always construct an algebra $\mathbb{A}$ and valuation $H$ such that neither $a/H=a/\dd a(\dd{Q_1}(H))$ nor $a/H=a/\dd a(\dd{Q_2}(H))$. This is not true in this variety. For example, take $Q_1\syn Q_2\syn T$. Then $\dd{Q_1}(H)=\dd{Q_2}(H)=H$, and by definition of this variety, $a/H=a/\dd a(H)$. We can solve this by instead of taking $P\circ a$ and $Q\circ a$ to show that $P$ and $Q$ are not congruent, we take $P\circ b$ and $Q\circ b$ where the existence of $b$ is guaranteed by the assumption that $|A|\ge 2$. Since it is possible to construct an algebra and corresponding valuation $H$ such that neither $b/H=b/\dd a(\dd{Q_1}(H))$ nor $b/H=b/\dd a(\dd{Q_2}(H))$.
\end{proof}

The argument for completeness is exactly the same as in the previous section, except that we use the lemmas proven in this section.
\begin{thm}
If $|A|\ge 2$ then for closed terms $P$ and $Q$,
\[
P=_{rp}Q\qquad\Longrightarrow\qquad \text{CP}_{rp}\vdash P=Q
\]
\end{thm}

Look at the following proposition to understand what happens when there is only one atomic proposition i.e., $|A|=1$.
\begin{prop}\label{prop one}
If $|A|=1$ then for all $P$, $Q$ and for all $\mathbb{A}\in rp$, $H\in\mathbb{A}$,
\[
P/\dd Q(H) = P/H
\]
\end{prop}
\begin{proof}
Proof by induction on $P$. If $P$ is either $T$ or $F$ then $P/\dd Q(H)=P/H$ follows immediately.

Suppose $P\syn a$ then we proceed by induction on $Q$. If $Q$ is either $T$ or $F$ then it is trivial. If $Q\syn a$ then
\begin{align*}
P/\dd Q(H) &= a/\dd a(H)\\
&= y_a(\dd a(H))\\
&= y_a(H)\\
&= P/H
\end{align*}
If $Q\syn Q_1\lef Q_2\rig Q_3$ then
\begin{align*}
a/\dd{(Q_1\lef Q_2\rig Q_3)}(H)
&=\begin{cases}
a/\dd{Q_1}(\dd{Q_2}(H)) & \text{if $Q_2/H=T$}\\
a/\dd{Q_3}(\dd{Q_2}(H)) & \text{if $Q_2/H=F$}
\end{cases}\\
&=_{IH}\begin{cases}
a/\dd{Q_2}(H) & \text{if $Q_2/H=T$}\\
a/\dd{Q_2}(H) & \text{if $Q_2/H=F$}
\end{cases}\\
&=_{IH} a/H
\end{align*}
Hence, $a/\dd Q(H)=a/H$ for all $Q$ and $H$.

Suppose $P\syn P_1\lef P_2\rig P_3$. Then
\begin{align*}
(P_1\lef P_2\rig P_3)/\dd Q(H)
&=\begin{cases}
P_1/\dd{P_2}(\dd Q(H)) & \text{if $P_2/\dd Q(H)=T$}\\
P_3/\dd{P_2}(\dd Q(H)) & \text{if $P_2/\dd Q(H)=F$}
\end{cases}\\
&=_{IH}\begin{cases}
P_1/(H) & \text{if $P_2/H=T$}\\
P_3/(H) & \text{if $P_2/H=F$}
\end{cases}\\
&=_{IH}\begin{cases}
P_1/\dd{P_2}(H) & \text{if $P_2/H=T$}\\
P_3/\dd{P_2}(H) & \text{if $P_2/H=F$}
\end{cases}\\
&=(P_1\lef P_2\rig P_3)/H
\end{align*}
\end{proof}
This proposition implies for example that
\begin{align*}
(T\lef a\rig a)/H
&=\begin{cases}
T/\dd a(H) & \text{if $a/H=T$}\\
a/\dd a(H) & \text{if $a/H=F$}
\end{cases}\\
&=\begin{cases}
T & \text{if $a/H=T$}\\
a/H & \text{if $a/H=F$}
\end{cases}\\
&= a/H
\end{align*}
In fact, Proposition \ref{prop one} implies that for $|A|=1$ we lose any kind of reactive behaviour, and we end up with static valuations. Furthermore, this proposition is clearly true for every variety in which all valuations are repetition-proof i.e., where $y_a(H)=y_a(\dd a(H))$ is true. Hence, we have the following corollary.
\begin{cor}\label{cor one}
If $|A|=1$ and the valuations in variety $K$ satisfy the equation $y_a(x)=y_a(\dd a(x))$ then
\[
P=_K Q\qquad\Longleftrightarrow\qquad P=_{st}Q
\]
for all closed terms $P$ and $Q$.
\end{cor}
This result will also be helpful in establishing completeness for the variety with contractive valuations.

\section{Contractive valuations}
Recall that the variety with contractive valuations is characterized by the following two equations:
\begin{align*}
y_a(x) &= y_a(\dd a(x))\\\\
\dd a(x) &= \dd a(\dd a(x))
\end{align*}
The first equation should be familiar since we encountered it in the previous section in the characterization of repetition-proof valuations. The second equation tells us that valuations remain constant through multiple $a$-derivatives. Consider the following example,
\begin{align*}
((b\lef a\rig c)\lef a\rig d)/H
&=\begin{cases}
(b\lef a\rig c)/\dd a(H) & \text{if $a/H=T$}\\
d/\dd a(H) & \text{if $a/H=F$}
\end{cases}\\
&=\begin{cases}
b/\dd a(\dd a(H)) & \text{if $a/H=T$ and $a/\dd a(H)=T$}\\
c/\dd a(\dd a(H)) & \text{if $a/H=T$ and $a/\dd a(H)=F$}\\
d/\dd a(H) & \text{if $a/H=F$}
\end{cases}\\
&=\begin{cases}
b/\dd a(H) & \text{if $a/H=T$ and $a/\dd a(H)=T$}\\
c/\dd a(H) & \text{if $a/H=T$ and $a/\dd a(H)=F$}\\
d/\dd a(H) & \text{if $a/H=F$}
\end{cases}\\
&=\begin{cases}
b/\dd a(H) & \text{if $a/H=T$}\\
d/\dd a(H) & \text{if $a/H=F$}
\end{cases}\\
&=(b\lef a\rig d)/H
\end{align*}
In the first two steps we expand the expression using the standard free reactive semantics. In the third step we replace $\dd a(\dd a(H))$ with $\dd a(H)$ using the definition of contractive valuations. Similarly, as in the example given in the section on repetition valuations we can eliminate the case where the example is equal to $c/\dd a(H)$ because $a/H$ must be equal to $a/\dd a(H)$.

By looking at the definition it becomes immediately apparent that the variety with contractive valuations is a subvariety of the variety with repetition-proof valuations i.e., if $P\equiv_{rp}Q$ then $P\equiv_{cr}Q$, and similarly if $P=_{rp}Q$ then $P=_{cr}Q$. Of course both varieties are subvarieties of the variety with free reactive valuations. This relation between the different varieties was previously also stated in Proposition \ref{variety relation prop}.

$cr$-Congruence is axiomatized by CP and the following axiom schemas.
\begin{align*}
(\CPcr1)\qquad(x\lef a\rig y)\lef a\rig z &= x\lef a\rig z\\
(\CPcr2)\qquad x\lef a\rig (y\lef a\rig z) &= x\lef a\rig z\\
\end{align*}
The entire axiomatization is called $\text{CP}_{cr}$. The axioms of \CPcr1 and \CPcr2 allow us to eliminate consecutive atomic propositions in our terms. So for example the terms $a\circ P$ and $a\circ a\circ P$ are provably equal. This is of course a stronger version of what we have seen in the previous section, which should not come as a surprise considering that this variety is defined in terms of the repetition-proof variety.

%In the previous section, we gave an example using a lookup table. Here we will also use a lookup table to demonstrate that contractive valuations naturally occur. Suppose that in this case the values of the lookup table are boolean. Let atomic proposition $a$ lookup the value $B$ indexed by specific key $K$. If $K$ exists $a$ returns the value $B$, otherwise $a$ returns false. Next we define atomic proposition $b$ to check to see if $K$ exists and if that is not the case, to create $K$ with corresponding value false. If $K$ exists $b$ returns true, otherwise false. Now consider the term $a\circ a\circ a\circ\ldots\circ a\circ b$. This term is now equal to $a\circ b$, in contrast to the example in the previous section. However, it still is not possible to remove an $a$ in the term $a\circ b\circ a$.

The following two sections show soundness and completeness for this variety.

\subsection{Soundness}
\begin{thm}
For closed terms $P$ and $Q$,
\[
\text{CP}_{cr}\vdash P=Q\qquad\Longrightarrow\qquad P=_{cr}Q
\]
\end{thm}
\begin{proof}
Since the variety with contractive valuations is a subvariety of the variety with free reactive valuations it suffices to show soundness for $\CPcr1$ and $\CPcr2$. Let algebra $\mathbb{A}$ and valuation $H\in\mathbb{A}$ be given.
\begin{align*}
((P\lef a\rig Q)\lef a\rig R)/H &= 
\begin{cases}
(P\lef a\rig Q)/\dd a(H) & \text{if $a/H=T$}\\
R/\dd a(H) & \text{if $a/H=F$}\\
\end{cases}\\
&=\begin{cases}
P/\dd a(\dd a(H)) & \text{if $a/H=T$ and $a/\dd a(H)=T$}\\
Q/\dd a(\dd a(H)) & \text{if $a/H=T$ and $a/\dd a(H)=F$}\\
R/\dd a(H) & \text{if $a/H=F$}\\
\end{cases}\\
&=\begin{cases}
P/\dd a(\dd a(H)) & \text{if $a/H=T$ and $a/H=T$}\\
Q/\dd a(\dd a(H)) & \text{if $a/H=T$ and $a/H=F$}\\
R/\dd a(H) & \text{if $a/H=F$}\\
\end{cases}\\
&=\begin{cases}
P/\dd a(\dd a(H)) & \text{if $a/H=T$}\\
R/\dd a(H) & \text{if $a/H=F$}\\
\end{cases}\\
&=\begin{cases}
P/\dd a(H) & \text{if $a/H=T$}\\
R/\dd a(H) & \text{if $a/H=F$}\\
\end{cases}\\
&=(P\lef a\rig R)/H
\end{align*}

\begin{align*}
\dd{((P\lef a\rig Q)\lef a\rig R)}(H)
&=\begin{cases}
\dd P(\dd a(\dd a(H))) & \text{if $a/H=T$ and $a/\dd a(H)=T$}\\
\dd Q(\dd a(\dd a(H))) & \text{if $a/H=T$ and $a/\dd a(H)=F$}\\
\dd R(\dd a(H)) & \text{if $a/H=F$}
\end{cases}\\
&=\begin{cases}
\dd P(\dd a(H)) & \text{if $a/H=T$}\\
\dd R(\dd a(H)) & \text{if $a/H=F$}\\
\end{cases}\\
&=\dd{(P\lef a\rig R)}(H)
\end{align*}
By Proposition \ref{equiv congr prop}, CPcr1 is sound. The proof of soundness for CPcr2 is similar to that of CPcr1.
\begin{align*}
(P\lef a\rig (Q\lef a\rig R))/H
&=\begin{cases}
P/\dd a(H) & \text{if $a/H=T$}\\
(Q\lef a\rig R)/\dd a(H) & \text{if $a/H=F$}\\
\end{cases}\\
&=\begin{cases}
P/\dd a(H) & \text{if $y_a(H)=T$}\\
Q/\dd a(\dd a(H)) & \text{if $a/H=F$ and $a/\dd a(H)=T$}\\
R/\dd a(\dd a(H)) & \text{if $a/H=F$ and $a/\dd a(H)=F$}\\
\end{cases}\\
&=\begin{cases}
P/\dd a(H) & \text{if $y_a(H)=T$}\\
Q/\dd a(\dd a(H)) & \text{if $a/H=F$ and $a/H=T$}\\
R/\dd a(\dd a(H)) & \text{if $a/H=F$ and $a/H=F$}\\
\end{cases}\\
&=\begin{cases}
P/\dd a(H) & \text{if $a/H=T$}\\
R/\dd a(\dd a(H)) & \text{if $a/H=F$}\\
\end{cases}\\
&=\begin{cases}
P/\dd a(H) & \text{if $a/H=T$}\\
R/\dd a(H) & \text{if $a/H=F$}\\
\end{cases}\\
&=(P\lef a\rig R)/H
\end{align*}

\begin{align*}
\dd{(P\lef a\rig(Q\lef a\rig R))}(H)
&=\begin{cases}
\dd P(\dd a(H)) & \text{if $a/H=T$}\\
\dd Q(\dd a(\dd a(H))) & \text{if $a/H=F$ and $a/\dd a(H)=T$}\\
\dd R(\dd a(\dd a(H))) & \text{if $a/H=F$ and $a/\dd a(H)=F$}\\
\end{cases}\\
&=\begin{cases}
\dd P(\dd a(H)) & \text{if $a/H=T$}\\
\dd R(\dd a(H)) & \text{if $a/H=F$}\\
\end{cases}\\
&=\dd{(P\lef a\rig R)}(H)
\end{align*}
Hence, CPcr2 is sound.
\end{proof}

\subsection{Completeness}
As we are working with a new variety we are required to define a new set of basic forms. Otherwise syntactic equality and $cr$-congruence will not coincide.
\begin{defn}
The set of contractive basic forms is the smallest set $\BF_{cr}$ such that $T,F\in\BF_{cr}$ and if $P,Q\in\BF_{cr}$ then for all $a\in A$
\begin{itemize}
\item if $P\syn P_1\lef a\rig P_2$ and $Q\syn Q_1\lef a\rig Q_2$ then $P_1\lef a\rig Q_2\in\BF_{cr}$
\item if $P\syn P_1\lef a\rig P_2$ and $Q\nsyn Q_1\lef a\rig Q_2$ then $P_1\lef a\rig Q\in\BF_{cr}$
\item if $P\nsyn P_1\lef a\rig P_2$ and $Q\syn Q_1\lef a\rig Q_2$ then $P\lef a\rig Q_2\in\BF_{cr}$
\item if $P\nsyn P_1\lef a\rig P_2$ and $Q\nsyn Q_1\lef a\rig Q_2$ then $P\lef a\rig Q\in\BF_{cr}$
\end{itemize}
\end{defn}
This definition differs from the one for repetition-proof basic forms. For example, $(T\lef a\rig T)\lef a\rig F$ is a valid repetition-proof basic form but it is not a contractive basic form. In fact, $\BF_{cr}$ is a subset of $\BF_{rp}$.
\begin{prop}
$\BF_{cr}\subseteq\BF_{rp}$
\end{prop}
\begin{proof}
Let $P\in\BF_{cr}$. Then by structural induction on $P$ we show that $P\in\BF_{rp}$. If $P\in\{T,F\}$, it follow immediately that $P\in\BF_{rp}$.

Suppose $P\syn P_1\lef a\rig P_2$. By definition, $P_1,P_2\in\BF_{cr}$. Hence, by I.H., it follows that $P_1,P_2\in\BF_{rp}$. By definition of contractive basic forms, $P_1\nsyn P_{11}\lef a\rig P_{12}$ and $P_2\nsyn P_{21}\lef a \rig P_{22}$. Consequently, $P_1\lef a\rig P_2\in\BF_{rp}$.
\end{proof}

The following lemma shows that the set $\BF_{cr}$ is indeed the set of basic forms we want.
\begin{lem}
For each closed term $P$ there exists a term $P'\in\BF_{cr}$ such
that $\text{CP}_{cr}\vdash P=P'$.
\end{lem}
\begin{proof}
By structural induction on $P$. We can assume that $P\in\BF$ because Lemma \ref{basic form theorem} is applicable.

If $P$ is $T$ or $F$ then $P\in\BF_{cr}$. Suppose $P$ is $P_1\lef a\rig P_2$. By the induction hypothesis,
there exists terms $P_1',P_2'\in\BF_{cr}$ such that $\vdash P_1=P_1'$
and $\vdash P_2=P_2'$. Now suppose that $P_1'\syn P_{11}'\lef a\rig
P_{12}'$ and $P_2'\syn P_{21}'\lef a\rig P_{22}'$. Consequently,
\begin{align*}
P_1\lef a\rig P_2
&= P_1'\lef a\rig P_2'\\
&= (P_{11}'\lef a\rig P_{12}')\lef a\rig (P_{21}'\lef a\rig P_{22}')\\
&= P_{11}'\lef a\rig P_{22}'
\end{align*}
By definition of $\BF_{cr}$ we have $P_{11}'\lef a\rig
P_{22}'\in\BF_{cr}$.

Using similar reasoning we can show that there exists such a term
$P'\in\BF_{cr}$ for the remaining three cases:
\begin{itemize}
\item $P_1'\nsyn P_{11}'\lef a\rig P_{12}'$ and $P_2'\syn P_{21}'\lef a\rig P_{22}'$
\item $P_1'\syn P_{11}'\lef a\rig P_{12}'$ and $P_2'\nsyn P_{21}'\lef a\rig P_{22}'$
\item $P_1'\nsyn P_{11}'\lef a\rig P_{12}'$ and $P_2'\nsyn P_{21}'\lef a\rig P_{22}'$
\end{itemize}
\end{proof}
The following two lemmas are needed for proving completeness for all closed terms.
\begin{lem}
For $P_1\lef a\rig P_2,Q_1\lef a\rig Q_2\in\BF_{cr}$,
\[
P_1\lef a\rig P_2=_{cr}Q_1\lef a\rig Q_2\qquad\Longrightarrow\qquad P_1=_{cr}Q_1\wedge P_2=_{cr}Q_2
\]
\end{lem}
\begin{proof}
We use a similar proof as the one for Lemma \ref{inv congr lemma}. Note that the problems that occurred in the variety with repetition-proof valuations from having either $P_1\syn P_{11}\lef a\rig P_{12}$ or $Q_1\syn Q_{11}\lef a\rig Q_{12}$ cannot occur here by construction of the contractive basic forms i.e., if either $P_1$ or $Q_1$ are syntactically equal to those terms then $P_1\lef a\rig P_2,Q_1\lef a\rig Q_2\in\BF_{cr}$ would not be true.
\end{proof}
\begin{lem}\label{cr completeness lemma}
For $|A|\ge 2$ and $P,Q\in\BF_{cr}$,
\[
P=_{cr}Q\qquad\Longleftrightarrow\qquad P\syn Q
\]
\end{lem}
\begin{proof}
We use a similar proof as the one for Lemma \ref{rp completeness lemma}.
\end{proof} 
Similar to the previous varieties, now that we have proven these lemmas, completeness for all closed terms follows.
\begin{thm}
If $|A|\ge 2$ then for closed terms $P$ and $Q$,
\[
P=_{cr}Q\qquad\Longrightarrow\qquad \text{CP}_{cr}\vdash P=Q
\]
\end{thm}
By definition there is at least on atomic proposition. In the section on repetition-proof valuation we proved Corollary \ref{cor one}. Clearly, this corollary also applies for this variety because $y_a(H)=y_a(\dd a(H))$ is also true in this variety. Hence, if $|A|=1$, contractive congruence $=_{cr}$ coincides with static congruence $=_{st}$.
%Consequently, the axiomatization is not complete when there is only one atomic proposition.

\section{Static valuations}
Static valuations correspond with classical propositional logic. As such the value of atomic propositions does not depend on other atomic propositions.

$st$-Congruence is axiomatized by CP and the following two axioms, due to Hoare in \cite{hoare}:
\begin{align*}
(\CPstat)\qquad (x\lef y\rig z)\lef u\rig v &=
(x\lef u\rig v)\lef y\rig(z\lef u\rig v)\\
(\CPcontr)\qquad (x\lef y\rig z)\lef y\rig u &= x\lef y\rig u
\end{align*}
The first axiom \CPstat{} tells us that the value of atomic propositions remain the same despite their relative position in the term. The second axiom \CPcontr{} is a generalization of the axioms for contractive valuations in the previous section i.e., this axiom allows contraction of not only atomic propositions but also for terms in general. We call the axiomatization of $st$-congruence $\text{CP}_{st}$.

The symmetric versions of the aforementioned axioms, listed below, follow from $\text{CP}_{st}$.
\begin{align*}
(\CPstat')\qquad x\lef y\rig(z\lef u\rig v) &=
(x\lef y\rig z)\lef u\rig(x\lef y\rig v)\\
(\CPcontr')\qquad x\lef y\rig(z\lef y\rig u) &=
x\lef y\rig u
\end{align*}
The key in deriving $\CPstat'$ and $\CPcontr'$ is the equality $y\lef x\rig z=z\lef(F\lef x\rig T)\rig y$ which we proved in the free reactive valuation section. For example,
\begin{align*}
x\lef y\rig(z\lef u\rig v)
&= x\lef(F\lef(F\lef y\rig T)\rig T)\rig(z\lef u\rig v)\\
&= (z\lef u\rig v)\lef(F\lef y\rig T)\rig x\\
&= (z\lef(F\lef y\rig T)\rig x)\lef u\rig (v\lef(F\lef y\rig T)\rig x)\\
&= (x\lef y\rig z)\lef u\rig(x\lef y\rig v)
\end{align*}
One can prove $\CPcontr'$ using the same technique.

Looking at the axiomatization it might not be immediately clear that this variety corresponds to classical propositional logic. One of the major differences between classical propositional logic and the varieties we have studied in the previous sections is the fact that the values of atomic propositions in a given term do not change depending on where they occur in the term. The following equality illustrates that the variety with static valuations also has this property.
\begin{align*}
x
&= (x\lef y\rig z)\lef F\rig x\\
&= (x\lef F\rig x)\lef y\rig (z\lef F\rig x)\\
&= x\lef y\rig x\\
&= y\circ x
\end{align*}
So appending an arbitrary term before $x$ will not change the value of $x$. In the next subsection we will prove its semantical counterpart.

In the following subsections we show soundness and completeness for the static valuations, and examine the relation between static valuations and boolean algebras.

\subsection{Soundness}
In contrast to the previous sections we cannot immediately start proving soundness but first need the following lemma which will not only be useful for proving soundness for this variety but also provides some additional insight to the correspondence between this logic and classical logic. It is worth noting that although this lemma can be viewed as a generalization of the way this variety is defined i.e., $y_a(\dd bx)=y_a(x)$, it in fact follows from the definition.
\begin{lem}\label{static gen lemma}
For all $P$, $Q$ and for all $\mathbb{A}\in st$, $H\in\mathbb{A}$,
\[
P/\dd Q(H) = P/H
\]
\end{lem}
\begin{proof}
Proof by structural induction on $P$. If $P\syn T$ or $P\syn F$ then $P/\dd Q(H) = P/H$ is trivially true. 

Suppose $P\syn a$ for $a\in A$. Then proceed by induction on $Q$. If $Q\syn T$ or $Q\syn F$ then $a/\dd Q(H) = a/H$ is trivially true. If $Q\syn b$ for $b\in A$. Then
\begin{align*}
P/H &= y_a(H)\\
&= y_a(\dd b(H))\\
&= P/\dd Q(H)
\end{align*}
Suppose $Q\syn Q_1\lef Q_2\rig Q_3$. Then
\begin{align*}
P/\dd Q(H)
&=a/\dd{Q_1\lef Q_2\rig Q_3}(H)\\
&=\begin{cases}
a/\dd {Q_1}(\dd {Q_2}(H)) & \text{if $Q_2/H=T$}\\
a/\dd {Q_3}(\dd {Q_2}(H)) & \text{if $Q_2/H=F$}\\
\end{cases}\\
&=_{IH}\begin{cases}
a/\dd {Q_2}(H) & \text{if $Q_2/H=T$}\\
a/\dd {Q_2}(H) & \text{if $Q_2/H=F$}\\
\end{cases}\\
&=a/\dd Q_2(H)\\
&=_{IH}a/H
\end{align*}
This concludes the case for $P\syn a$.

Suppose $P\syn P_1\lef P_2\rig P_3$. Then
\begin{align*}
P/\dd Q(H)
&=(P_1\lef P_2\rig P_3)/\dd Q(H)\\
&=\begin{cases}
P_1/\dd {P_2}(\dd Q(H)) & \text{if $P_2/\dd Q(H)=T$}\\
P_3/\dd {P_2}(\dd Q(H)) & \text{if $P_2/\dd Q(H)=F$}\\
\end{cases}\\
&=_{IH}\begin{cases}
P_1/\dd Q(H) & \text{if $P_2/H=T$}\\
P_3/\dd Q(H) & \text{if $P_2/H=F$}\\
\end{cases}\\
&=_{IH}\begin{cases}
P_1/H & \text{if $P_2/H=T$}\\
P_3/H & \text{if $P_2/H=F$}\\
\end{cases}\\
&=(P_1\lef P_2\rig P_3)/H
\end{align*}
\end{proof}
The previous lemma shows that the valuation of terms is independent of the context in which they appear. It is directly related to the equality $x=y\circ x$ which we derived in the introduction to this section, in that it similarly illustrates that the value of the atomic propositions is not dependent on what other atomic propositions might have occurred during the valuation of a term, and consequently it cannot change during the valuation.

In the previously observed varieties there is a clear distinction between equivalence and congruence. This difference was proven by using the example where $a\circ T\equiv_KT$ but $a\circ T\ne_KT$ for $K\in\{fr,rp,cr\}$ because $a\circ T\circ b\not\equiv_Kb$. However, as the previous lemma shows we cannot apply this example for the static valuations. In fact the following lemma shows that in this variety equivalence and congruence coincide.
\begin{lem}\label{equiv congr lem}
For closed terms $P$ and $Q$,
\[
P\equiv_{st}Q\qquad\Longleftrightarrow\qquad P=_{st}Q
\]
\end{lem}
\begin{proof}
Congruence is by definition an equivalence. So it will suffice to show that the equivalence $\equiv_{st}$ also has the congruence property. Suppose $P_1\equiv_{st}Q_1$, $P_2\equiv_{st}Q_2$, and $P_3\equiv_{st}Q_3$. Then
\begin{align*}
(P_1\lef P_2\rig P_3)/H
&=\begin{cases}
P_1/\dd{P_2}(H) & \text{if $P_2/H=T$}\\
P_3/\dd{P_2}(H) & \text{if $P_2/H=F$}
\end{cases}\\
&=_*\begin{cases}
P_1/H & \text{if $P_2/H=T$}\\
P_3/H & \text{if $P_2/H=F$}
\end{cases}\\
&=\begin{cases}
Q_1/H & \text{if $Q_2/H=T$}\\
Q_3/H & \text{if $Q_2/H=F$}
\end{cases}\\
&=_*\begin{cases}
Q_1/\dd{Q_2}(H) & \text{if $Q_2/H=T$}\\
Q_3/\dd{Q_2}(H) & \text{if $Q_2/H=F$}
\end{cases}\\
&=(Q_1\lef Q_2\rig Q_3)/H
\end{align*}
At both the *-marked steps in the derivation we apply Lemma \ref{static gen lemma}.
\end{proof}
Note that the application of Lemma \ref{static gen lemma} in the previous proof is necessary because otherwise $P_1/\dd{P_2}(H)$ and $P_3/\dd {P_2}(H)$ can only be replaced by $Q_1/\dd {P_2}(H)$ and $Q_3/\dd{P_2}(H)$ at which point the expression cannot be further reduced. So this line of reasoning will not work for the other varieties where we do not have this lemma. 

Now that we have proven these lemmas, soundness is relatively easy.
\begin{thm}
For closed terms $P$ and $Q$,
\[
\text{CP}_{st}\vdash P=Q\qquad\Longrightarrow\qquad P=_{st}Q
\]
\end{thm}
\begin{proof}
By Lemma \ref{equiv congr lem}, we only need to show that
\[
\text{CP}_{st}\vdash P=Q\qquad\Longrightarrow\qquad P\equiv_{st}Q
\]
It suffices to show soundness for only $\CPstat$ and $\CPcontr$. Take note of the frequent use of Lemma \ref{static gen lemma} in the derivations below.
\begin{align*}
((P\lef Q\rig R)\lef S\rig V)/H
&=\begin{cases}
(P\lef Q\rig R)/\dd S(H) & \text{if $S/H=T$}\\
V/\dd S(H) & \text{if $S/H=F$}\\
\end{cases}\\
&=\begin{cases}
P/\dd Q(\dd S(H)) & \text{if $S/H=T$ and $Q/\dd S(H)=T$}\\
R/\dd Q(\dd S(H)) & \text{if $S/H=T$ and $Q/\dd S(H)=F$}\\
V/\dd S(H) & \text{if $S/H=F$}\\
\end{cases}\\
&=\begin{cases}
P/H & \text{if $S/H=T$ and $Q/H=T$}\\
R/H & \text{if $S/H=T$ and $Q/H=F$}\\
V/H & \text{if $S/H=F$}\\
\end{cases}\\
&=\begin{cases}
P/H & \text{if $Q/H=T$ and $S/H=T$}\\
V/H & \text{if $Q/H=T$ and $S/H=F$}\\
R/H & \text{if $Q/H=F$ and $S/H=T$}\\
V/H & \text{if $Q/H=F$ and $S/H=F$}\\
\end{cases}\\
&=\begin{cases}
P/\dd S(\dd Q(H)) & \text{if $Q/H=T$ and $S/\dd Q(H)=T$}\\
V/\dd S(\dd Q(H)) & \text{if $Q/H=T$ and $S/\dd Q(H)=F$}\\
R/\dd S(\dd Q(H)) & \text{if $Q/H=F$ and $S/\dd Q(H)=T$}\\
V/\dd S(\dd Q(H)) & \text{if $Q/H=F$ and $S/\dd Q(H)=F$}\\
\end{cases}\\
&=\begin{cases}
(P\lef S\rig V)/\dd Q(H) & \text{if $Q/H=T$}\\
(R\lef S\rig V)/\dd S(H) & \text{if $Q/H=F$}\\
\end{cases}\\
&=((P\lef S\rig V)\lef Q\rig(R\lef S\rig V))/H
\end{align*}
\begin{align*}
((P\lef Q\rig R)\lef Q\rig S)/H
&=\begin{cases}
(P\lef Q\rig R)/\dd Q(H) & \text{if $Q/H=T$}\\
S/\dd Q(H) & \text{if $Q/H=F$}\\
\end{cases}\\
&=\begin{cases}
P/\dd Q(\dd Q(H)) & \text{if $Q/H=T$ and $Q/\dd Q(H)=T$}\\
R/\dd Q(\dd Q(H)) & \text{if $Q/H=T$ and $Q/\dd Q(H)=F$}\\
S/\dd Q(H) & \text{if $Q/H=F$}\\
\end{cases}\\
&=\begin{cases}
P/\dd Q(\dd Q(H)) & \text{if $Q/H=T$ and $Q/H=T$}\\
R/\dd Q(\dd Q(H)) & \text{if $Q/H=T$ and $Q/H=F$}\\
S/\dd Q(H) & \text{if $Q/H=F$}\\
\end{cases}\\
&=\begin{cases}
P/\dd Q(\dd Q(H)) & \text{if $Q/H=T$}\\
S/\dd Q(H) & \text{if $Q/H=F$}\\
\end{cases}\\
&=\begin{cases}
P/\dd Q(H) & \text{if $Q/H=T$}\\
S/\dd Q(H) & \text{if $Q/H=F$}\\
\end{cases}\\
&= (P\lef Q\rig S)/H
\end{align*}
Thus CPstat and CPcontr are sound. Soundness of the rest of the axioms follows by the soundness of the variety with free reactive valuations.
\end{proof}

\subsection{Completeness}
Proving completeness follows the same strategy as we have seen before i.e., we define basic forms, and prove completeness for the basic forms. However, the individual lemmas will differ significantly from what we have seen up to this point because the construction of the static basic forms is more complicated.  

In order to define the static basic forms, we first need to enumerate the members of $A$:
\[
a_1,a_2,\ldots,a_n
\]
Recall that a basic form i.e., a member of \BF, corresponds to a labeled binary tree. A static basic form is a member of \BF, and is a full binary tree with $n+1$ levels. At level $i$ only atomic proposition $a_i$ occurs, and at level $n+1$ at each leaf there is either a $T$ or $F$. The resulting tree is pictured below:
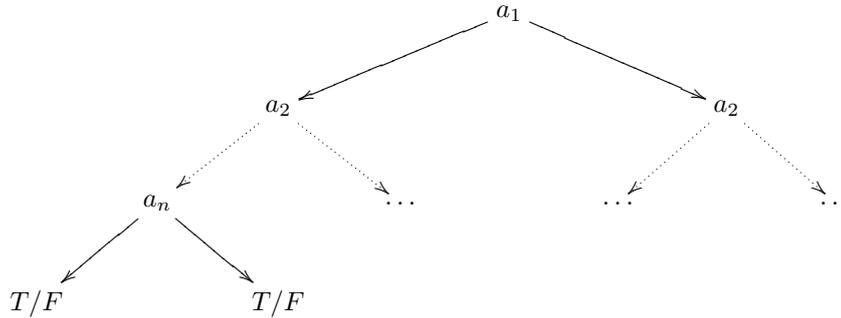
\[
\xymatrix{
&&&&a_1\ar[dll]\ar[drr]\\
&&a_2\ar@{.>}[dl]\ar@{.>}[dr]&&&&a_2\ar@{.>}[dl]\ar@{.>}[dr]\\
&a_n\ar[dl]\ar[dr] &&\ldots&&\ldots&&\ldots\\
T/F && T/F\\
}
\]
So an atomic proposition $a_i$ occurs $2^{i-1}$ times and there are $2^n$ leaves. The set of static basic forms is called $\BF_{st}$. 

The following two lemmas are needed to prove that there exists a static basic form for each closed term.
\begin{lem}\label{red lemma}
For $P,R\in\BF_{st}$ and $a_i\in A$, there exists a term $Q\in\BF_{st}$ such that
\[
\text{CP}_{st}\vdash P\lef a_i\rig R=Q
\]
\end{lem}
\begin{proof}
We proceed by induction on the number of atomic propositions. Note that since $P,R\in\BF_{st}$ it follows that $P\syn P_1\lef a_1\rig P_2$ and $R\syn R_1\lef a_1\rig R_2$ with $P_1,P_2,R_1,R_2\in\BF_{st}$.

If $|A|=1$ then $a_i=a_1$, and
\begin{align*}
P\lef a_i\rig R
&=(P_1\lef a_1\rig P_2)\lef a_1\rig(R_1\lef a_1\rig R_2)\\
&=_*P_1\lef a_1\rig R_2\\
\end{align*}
(*) is obtained by applying \CPcontr. Since $P_1\lef a_1\rig R_2$ is a static basic form, we have shown the existence of $Q$.

Suppose $|A|=n+1$ and $i\ge 2$ (if $i=1$ apply $\CPcontr$ and $\CPcontr'$). Then
\begin{align*}
P\lef a_i\rig R
&=(P_1\lef a_1\rig P_2)\lef a_i\rig(R_1\lef a_1\rig R_2)\\
&=(P_1\lef a_i\rig(R_1\lef a_1\rig R_2))\lef a_1\rig(P_2\lef a_i\rig(R_1\lef a_1\rig R_2))\\
\end{align*}
Next take the left consequent,
\begin{align*}
P_1\lef a_i\rig(R_1\lef a_1\rig R_2)
&=(P_1\lef a_i\rig R_1)\lef a_1\rig(P_1\lef a_i\rig R_2)\\
&=_{IH}Q_1\lef a_1\rig Q_2
\end{align*}
where $Q_1=P_1\lef a_i\rig R_1$, $Q_2=P_1\lef a_i\rig R$, and both $Q_1$ and $Q_2$ are static basic forms given the set of atomic propositions $A\setminus\{a_1\}$ (but otherwise the same enumeration). Note that since $|A\setminus\{a_1\}|=n$, we could use the I.H. in the above derivation. 

We can apply the same argument for the right consequent $P_2\lef a_i\rig(R_1\lef a_1\rig R_2)$ and obtain $Q_2$ and $Q_3$ such that they are static basic forms for this set, and $Q_3=P_2\lef a_i\rig R_1$ and $Q_4=P_2\lef a_i\rig R_2$. Consequently,
\begin{align*}
P\lef a_i\rig R
&=(Q_1\lef a_1\rig Q_2)\lef a_1\rig(Q_3\lef a_1\rig Q_4)\\
&=Q_1\lef a_1\rig Q_4
\end{align*}
Clearly, the term $Q_1\lef a_1\rig Q_4$ is a static basic form for the set $A$.
\end{proof}

The rest of this section resembles the previous sections. So we start by showing that there is a static basic form for each closed term.
\begin{thm}
For each term closed term $P$ there exists a $P'\in\BF_{st}$ such
that 
\[
\text{CP}_{st}\vdash P=P'
\]
\end{thm}
\begin{proof}
By structural induction on $P$. Since we already know that there exists a provably equal basic form for each closed term $P$ (as opposed to a static basic form), we can assume $P$ is a basic form. If $P\in\{T,F\}$ then simply construct a static basic forms where all the leaves are either $T$ or $F$.
Suppose $P\syn P_1\lef a\rig P_2$. By I.H. there exist $P_1',P_2'\in\BF_{st}$ such that $P_1=P_1'$ and $P_2=P_2'$. Then by the previous lemma we know $P'$ exists.
\end{proof}
Next we show that for static basic forms the congruence $=_{st}$ coincides with syntactic equality.
\begin{lem}
For static basic forms $P$ and $Q$,
\[
P=_{st}Q\qquad\Longleftrightarrow\qquad P\syn Q
\]
\end{lem}
\begin{proof}
Since one direction is trivial, it suffices to prove that $P=_{st}Q$ implies $P\syn Q$. Assume $P\nsyn Q$ for static basic forms $P$ and $Q$. By definition of static basic forms, it follows that there is at least one leaf that differs in value for $P$ and $Q$. For example the leftmost leaf for $P$ has value $T$ and the leftmost leaf for $Q$ has value $F$. It is then trivial to construct a static valuation $H$ such that $P/H\ne Q/H$. In the example we just mentioned this valuation would assign true to all atomic propositions. Since there is a valuation $H$ such that $P/H\ne Q/H$, it follows that $P\not\equiv_{st}Q$. Hence, $P\ne_{st} Q$.
\end{proof}
Using the same reasoning as in the previous sections we obtain completeness for all closed terms.
\begin{thm}
For closed terms $P$ and $Q$,
\[
P=_{st}Q\qquad\Longrightarrow\qquad \text{CP}_{st}\vdash P=Q
\]
\end{thm}

\chapter{$\omega$-Completeness}
In this chapter we discuss $\omega$-completeness of the different axiomatizations we encountered thus far. The following definition of $\omega$-completeness is taken from \cite{fokkink}.
\begin{defn}
An axiomatization $\mathcal{E}$ over a signature $\Sigma$ is $\omega$-complete if an equation $s=t$ with $s,t\in\mathbb{T}(\Sigma)$ can be derived from $\mathcal{E}$ if $\sigma(s)=\sigma(t)$ can be derived from $\mathcal{E}$ for all closed substitutions $\sigma$.
\end{defn}
The set $\mathbb{T}(\Sigma)$ is the set of all terms over signature $\Sigma$. $\omega$-Completeness is also know as \emph{inductive completeness} since we do not need an additional induction theorem to prove that $s=t$ can be derived if $\sigma(s)=\sigma(t)$ can be derived for all closed substitutions $\sigma$.

An example of an axiomatization that is not $\omega$-complete is the following axiomatization of the natural numbers with addition and multiplication, taken from \cite{omega_bergstra}:
\begin{align*}
x + 0 &= x\\
x + S(y) &= S(x+y)\\
x\cdot 0 &= 0\\
x\cdot S(y) &= x + (x\cdot y)
\end{align*}
In this axiomatization every closed instance of e.g., $x+y=y+x$ can be derived. However, the theorem itself cannot be derived from the above axioms.

See \cite{omega} for a more thorough introduction to $\omega$-completeness.

\section{$\omega$-Completeness of CP}

We begin with proving $\omega$-completeness for CP. First, however, it is necessary to distinguish between a few unique cases based on the number of elements in the set $A$ of atomic propositions. As it turns out CP is not $\omega$-complete for $|A|<2$.

If the set $A$ is empty, and there are no atomic propositions, it follows that every closed closed substitution $\sigma$ replaces the variables by terms build up from $T$, $F$ and $\_\lef \_\rig \_$. Using axioms CP1 and CP2, this is the same as replacing the variables by either $T$ or $F$. If we now consider the equation $x\circ T=T$, we see that for every closed substitution $\sigma$, $\sigma(x\circ T)=\sigma(T)$ follows from CP1 and CP2. However, $x\circ T=T$ cannot be derived from CP. If it could be derived then it could also be derived in the case where $A\ne\emptyset$ because the derivations are independent of the number of atomic propositions in $A$. So if this is the case then by soundness $a\circ T=_{fr}T$, which is clearly not true.

Similarly, when there is only one atomic proposition i.e., $A=\{a\}$, we take the equation $x\circ a\circ T=a\circ x\circ T$. For every closed substitution $\sigma$, $\sigma(x\circ a\circ T)=\sigma(a\circ x\circ T)$, where we can assume without loss of generality that $\sigma$ replaces variables either by $T$, $F$ or a sequence of $a$'s i.e., $a\circ a\circ\ldots\circ a$.  However, $CP\nvdash x\circ a\circ T=a\circ x\circ T$.

Therefore, in the remainder of the discussion of $\omega$-completeness for CP, we assume that $A$ has at least two atomic propositions, usually referred to as $a$ and $b$.

Similar to the sections where we showed completeness for the various varieties, we define a set of basic forms. However, this time the basic forms can also be open terms.
\begin{defn}\label{open basic form def}
Let the set of open basic forms $\mathbb{T}$ be the smallest set such that $T, F\in\mathbb{T}$, and if $s, t\in\mathbb{T}$ then $s\lef a\rig t\in\mathbb{T}$ for all $a\in A$ and $s\lef x\rig t\in\mathbb{T}$ for all $x\in V$.
\end{defn}
The set $V$ is the set of variables. The terminology ``open basic form'' is a bit misleading because an open basic form can possibly be a closed term. This definition differs from the one for the set of closed basic forms \BF{} in that terms of the form $s\lef x\rig t$ are also included. In addition, observe that if we substitute all the variables by atomic propositions in an arbitrary open basic form, the resulting term will be a member of \BF.

\begin{lem}\label{open basic form lem}
For all terms $t$ there exists an open basic form $t'$ such that $CP\vdash t=t'$.
\end{lem}
\begin{proof}
Proof by structural induction, very similar to the proof of Theorem \ref{basic form theorem}.
\end{proof}

The following lemma tells us that when dealing with open basic forms, it suffices to use closed substitutions that map variables to atomic propositions instead of arbitrary closed terms.

\begin{lem}\label{substitution lemma}
For open basic forms $s$ and $t$, if for all closed substitutions $\sigma:V\to A$,\footnote{Note that with the notation $\sigma:V\to A$ we do not imply that this substitution only works on variables but that the substitution replaces the variables in a term with members of $A$.}
\[
\sigma(s)=_{fr}\sigma(t)
\]
Then for all closed substitutions $\sigma:V\to T(\sig)$,
\[
\sigma(s)=_{fr}\sigma(t)
\]
\end{lem}
\begin{proof}
Assume there exists a closed substitution $\sigma:V\to T(\sig)$ such that $\sigma(s)\ne_{fr}\sigma(t)$ for open basic forms $s$ and $t$. We prove by induction on $s$ and $t$ that there exists a closed substitution $\tau:V\to A$ such that $\tau(s)\ne_{fr}\tau(t)$. We omit the cases where both $s$ and $t$ are closed terms, which are trivial, and the cases that follow by symmetry.

Suppose $s\in\{T,F\}$ and $t\syn t_1\lef x\rig t_2$ or $t\syn t_1\lef a\rig t_2$. Let $\tau$ be the substitution that maps all variables to atomic proposition $a$. Then $\tau(s),\tau(t)\in\BF$, and clearly $\tau(s)\nsyn\tau(t)$. By Lemma \ref{completeness lemma}, $\tau(s)\ne_{fr}\tau(t)$.

Suppose $s\syn s_1\lef a\rig s_2$ and $t\syn t_1\lef b\rig t_2$. Then it is trivial to see that no matter what closed substitution $\tau:V\to A$ we choose, $\tau(s)\ne_{fr}\tau(t)$. Similarly, we can reduce the case where $s\syn s_1\lef x\rig s_2$ and $t\syn t_1\lef y\rig t_2$, to this case by letting $\tau(x)=a$ and $\tau(y)=b$. Furthermore, the same argument can be used when one of the terms starts with a variable and the other with an atomic proposition.

Suppose $s\syn s_1\lef x\rig s_2$ and $t\syn t_1\lef x\rig t_2$. Since $\sigma(s)\ne_{fr}\sigma(t)$, we can assume without loss of generality that $\sigma(s_1)\ne_{fr}\sigma(t_1)$. By the induction hypothesis, it follows that there exists a closed substitution $\tau:V\to A$ such that $\tau(s_1)\ne_{fr}\tau(t_1)$. Consequently, by Lemma \ref{inv congr lemma}, $\tau(t)\ne_{fr}\tau(s)$. We can use the same argument for $s\syn s_1\lef a\rig s_2$ and $t\syn t_1\lef a\rig t_2$.
\end{proof}

Note that this lemma is trivially true when $A=\emptyset$. The above lemma cannot be proven for $|A|=1$. Take for example,
\begin{align*}
s &=(T\lef x\rig F)\lef a\rig(T\lef x\rig F)\\
t &=(T\lef a\rig F)\lef x\rig(T\lef a\rig F)
\end{align*}
where we assume that $A=\{a\}$. Then the only closed substitution $\tau:V\to A$ is the one that maps all variables to $a$. Hence $\tau(s)\syn\tau(t)$, and thus $\tau(s)=_{fr}\tau(t)$. However, if we pick closed substitution $\sigma:V\to T(\sig)$ such that $\sigma(x)=F$, then clearly $\sigma(s)\ne_{fr}\sigma(t)$.

\begin{lem}\label{omega lem}
For open basic forms $s$ and $t$, if for all closed substitutions $\sigma$, $\sigma(s)=_{fr}\sigma(t)$ then $s\syn t$.
\end{lem}
\begin{proof}
We prove the contraposition i.e., given that $s\nsyn t$ we show that there exists a closed substitution $\sigma$ such that $\sigma(s)\ne_{fr}\sigma(t)$. Note that if $s$ and $t$ are closed terms then $s, t\in\BF$, the set of closed basic forms. Since we already showed in Lemma \ref{completeness lemma} that for terms in \BF{} $fr$-congruence and syntactical equality coincide, we are done. We proceed by induction on $s$ and $t$, omitting the cases that follow by symmetry.

Suppose one of the following cases,
\begin{itemize}
\item $s\in\{T,F\}$ and $t\syn t_1\lef a\rig t_2$
\item $s\in\{T,F\}$ and $t\syn t_1\lef x\rig t_2$
\item $s\syn s_1\lef a\rig s_2$ and $t\syn t_1\lef b\rig t_2$
\item $s\syn s_1\lef x\rig s_2$ and $t\syn t_1\lef b\rig t_2$
\item $s\syn s_1\lef x\rig s_2$ and $t\syn t_1\lef y\rig t_2$
\end{itemize}
Let $\sigma:V\to A$ be a closed substitution such that $\sigma(x)=a$ and $\sigma(y)=b$. Regardless of which case we pick, it follows that $\sigma(s),\sigma(t)\in\BF$. Furthermore, $\sigma(s)\nsyn\sigma(t)$. Consequently, $\sigma(s)\ne_{fr}\sigma(t)$ by Lemma \ref{completeness lemma}.

Suppose $s\syn s_1\lef x\rig s_2$ and $t\syn t_1\lef x\rig t_2$. Since $s\nsyn t$, we can assume without loss of generality that $s_1\nsyn t_1$. By the induction hypothesis, there exists a closed substitution $\sigma$ such that $\sigma(s_1)\ne_{fr}\sigma(t_1)$. Consequently, by Lemma \ref{substitution lemma}, there exists a closed substitution $\tau:V\to A$ such that $\tau(s_1)\ne_{fr}\tau(t_1)$. Hence, by Lemma \ref{inv congr lemma}, $\tau(s)\ne_{fr}\tau(t)$. Using a similar argument, we can prove this for $s\syn s_1\lef a\rig s_2$ and $t\syn t_1\lef a\rig t_2$.
\end{proof}

The following theorem shows that CP is $\omega$-complete.

\begin{thm}
Let $s$ and $t$ be open terms such that for all closed substitutions $\sigma$, $\CP\vdash\sigma(s)=\sigma(t)$, then $\CP\vdash s=t$.
\end{thm}
\begin{proof}
Assume that $\forall\sigma:\CP\vdash\sigma(s)=\sigma(t)$. By Lemma \ref{open basic form lem}, there exist open basic forms $s'$ and $t'$ such that $\CP\vdash s=s',t=t'$. Hence, $\forall\sigma:\CP\vdash\sigma(s')=\sigma(t')$. By soundness, it follows that $\forall\sigma:\sigma(s')=_{fr}\sigma(t')$. Consequently, by Lemma \ref{omega lem}, $s'\syn t'$, and thus $\CP\vdash s'=t'$. It follows that $\CP\vdash s=t$.
\end{proof}

We were not able to prove $\omega$-completeness for $\text{CP}_{rp}$ and $\text{CP}_{cr}$. Instead we note that using the same examples as before, it follows that both $\text{CP}_{rp}$ and $\text{CP}_{cr}$ are not $\omega$-complete for $|A|<2$. Furthermore, proving $\omega$-completeness for the case $|A|\ge 2$ suggests that we adjust the set of open basic forms accordingly and then prove the same lemmas as before only this time using the altered set of open basic forms. 

\section{$\omega$-Completeness of $\text{CP}_{st}$}

The proof of $\omega$-completeness for the axiomatization $\text{CP}_{st}$ is quite different from the previous proof. We will use a back-and-forth translation between $\text{CP}_{st}$ and an axiomatization of boolean algebra of which we know that it is $\omega$-complete, to show $\omega$-completeness of $\text{CP}_{st}$.

In \cite{omega_bergstra} a proof of $\omega$-completeness for the axiomatization of an $n$-valued Post algebra is given. If we take $n=2$ the axiomatization is that of a Boolean algebra. We obtain the following axiomatization by taking $n=2$, and removing some of the redundant axioms.
\begin{align*}
x\vee y &= y\vee x\\
x\wedge y &= y\wedge x\\
x\vee(y\vee z) &= (x\vee y)\vee z\\
x\wedge(x\vee y) &= x\\
x\vee(x\wedge y) &= x\\
x\vee(y\wedge z) &= (x\vee y)\wedge (x\vee z)\\
F\vee x &= x\\
x\wedge T &= x\\
x\wedge\neg x &= F\\
\neg x\vee x &= T\\
\neg(x\wedge y) &= \neg x\vee\neg y
\end{align*}
We call this axiomatization BA. The signature of BA consists of $T$, $F$, $\neg$, $\vee$ and $\wedge$. Expanding the signature with a set of atomic propositions $A$ does not affect the $\omega$-completeness of this axiomatization. We call the resulting signature $\Sigma_{BA}$. We use the following two translations between $\sig$ and $\Sigma_{BA}$, starting with the translation of $\sig$ to $\Sigma_{BA}$.
\begin{align*}
T' &= T\\
F' &= F\\
a' &= a\\
x' &= x\\
(t\lef r\rig s)' &= (\neg r'\vee t')\wedge(r'\vee s')
\end{align*}
Note that $(\neg r'\vee t')\wedge(r'\vee s')=(r'\to t')\wedge(\neg r'\to s')$ which is perhaps more intuitive. The translation from $\Sigma_{BA}$ to $\sig$ looks as follows.
\begin{align*}
T^* &= T\\
F^* &= F\\
a^* &= a\\
x^* &= x\\
(\neg t)^* &= F\lef t^*\rig T\\
(t\vee r)^* &= T\lef t^*\rig r^*\\
(t\wedge r)^* &= r^*\lef t^*\rig F
\end{align*}
Note that these are translations over all terms including open terms. The next two lemmas show that the translations are sound i.e., if two terms are provably equal in either $\text{CP}_{st}$ or BA then their respective translations are also provably equal.
\begin{lem}\label{CPst in BA lem}
For all terms $s$ and $t$,
\[
\text{CP}_{st}\vdash s=t\qquad\Longrightarrow\qquad\text{BA}\vdash s'=t'
\]
\end{lem}
\begin{proof}
It suffices to show that the translations of axioms CP1-4, CPstat and CPcontr can be derived in BA.
\begin{align*}
(x\lef T\rig y)' &= (\neg T\vee x)\wedge(T\vee y)\\
&=x'\\\\
(x\lef F\rig y)' &= (\neg F\vee x)\wedge(F\vee y)\\
&=y'\\\\
(T\lef x\rig F)' &= (\neg x\vee T)\wedge(x\vee F)\\
&=x'
\end{align*}
Using a truth table we can check that the translations of CP4, CPstat and CPcontr are correct because BA is both sound and complete.
\end{proof}
\begin{lem}\label{BA in CPst lem}
For all terms $s$ and $t$,
\[
\text{BA}\vdash s=t\qquad\Longrightarrow\qquad\text{CP}_{st}\vdash s^*=t^*
\]
\end{lem}
\begin{proof}
We just need to check that the translations of axioms of BA are derivable in $\text{CP}_{st}$. We omit the trivial derivations.
\begin{align*}
(x\vee y)^* &= T\lef x\rig y\\
&= T\lef x\rig (T\lef y\rig F)\\
&= (T\lef x\rig T)\lef y\rig(T\lef x\rig F)\\
&= T\lef y\rig x\\
&= (y\lor x)^*\\\\
(x\land y)^*
&= y\lef x\rig F\\
&= (T\lef y\rig F)\lef x\rig F\\
&= (T\lef x\rig F)\lef y\rig (F\lef x\rig F)\\
&= x\lef y\rig F\\
&= (y\land x)^*\\\\
(x\lor(y\lor z))^*
&= T\lef x\rig(T\lef y\rig z)\\
&= (T\lef t\rig z)\lef x\rig (T\lef y\rig z)\\
&= T\lef(T\lef x\rig y)\rig z\\
&= ((x\lor y)\lor z)^*\\\\
(x\land (x\lor y))^*
&= (T\lef x\rig y)\lef x\rig F\\
&= T\lef x\rig F\\
&= x^*\\\\
(x\lor(x\land y))^*
&= T\lef x\rig(y\lef x\rig F)\\
&= T\lef x\rig F\\
&= x^*\\\\
(x\vee(y\wedge z))^*
&=T\lef x\rig(z\lef y\rig F)\\
&=(T\lef x\rig z)\lef y\rig(T\lef x\rig F)\\
&=(T\lef x\rig (T\lef x\rig z))\lef y\rig(T\lef x\rig F)\\
&=T\lef x\rig((T\lef x\rig z)\lef y\rig F)\\
&=(T\lef x\rig z)\lef x\rig((T\lef x\rig z)\lef y\rig F)\\
&=((T\lef x\rig z)\lef T\rig F)\lef x\rig((T\lef x\rig z)\lef y\rig F)\\
&=(T\lef x\rig z)\lef(T\lef x\rig y)\rig F\\
&=((x\vee y)\wedge(x\vee z))^*\\
\end{align*}

\begin{align*}
(x\wedge\neg x)^*
&=(F\lef x\rig T)\lef x\rig F\\
&=F\lef x\rig F\\
&=F^*\\\\
(\neg x\vee x)^*
&= T\lef (F\lef x\rig T)\rig x\\
&= x\lef x\rig T\\
&= (T\lef x\rig F)\lef x\rig T\\
&= T\lef x\rig T\\
&= T^*\\\\
(\neg(x\wedge y))^*
&=F\lef(y\lef x\rig F)\rig T\\
&=(F\lef y\rig T)\lef x\rig T\\
&=T\lef(F\lef x\rig T)\rig(F\lef y\rig T)\\
&=(\neg x\vee\neg y)^*
\end{align*}
\end{proof}

The following two lemmas show that the translations are invariant for each logic i.e., if a term $t$ is translated first to one logic and then back to the original it is still provably equal to $t$.
\begin{lem}\label{trans invariance CPst lem}
For all terms $s$, 
\[
\text{CP}_{st}\vdash(s')^*=s
\]
\end{lem}
\begin{proof}
Proof by structural induction on $s$. If $s\in\{T,F\}\cup A\cup V$ then it is trivially true. Suppose $s\syn s_1\lef s_2\rig s_3$. By the induction hypothesis, it follows that
\[
\text{CP}_{st}\vdash(s_1')^*=s_1,\quad(s_2')^*=s_2,\quad(s_3')^*=s_3
\]
Then
\begin{align*}
((s_1\lef s_2\rig s_3)')^*
&=((\neg s_2'\vee s_1')\wedge(s_2'\vee s_3'))^*\\
&=(T\lef (s_2')^*\rig(s_3')^*)\lef(T\lef(F\lef(s_2')^*\rig T)\rig (s_1')^*)\rig F\\
&=_{IH}(T\lef s_2\rig s_3)\lef(T\lef(F\lef s_2\rig T)\rig s_1)\rig F\\
&=(T\lef s_2\rig s_3)\lef(s_1\lef s_2\rig T)\rig F\\
&=((T\lef s_2\rig s_3)\lef s_1\rig F)\lef s_2\rig ((T\lef s_2\rig s_3)\lef T \rig F)\\
&=((T\lef s_2\rig s_3)\lef s_1\rig F)\lef s_2\rig (T\lef s_2\rig s_3)\\
&=((T\lef s_2\rig s_3)\lef s_1\rig F)\lef s_2\rig s_3\\
&=((T\lef s_1\rig F)\lef s_2\rig(s_3\lef s_1\rig F)\lef s_2\rig s_3\\
&=(T\lef s_1\rig F)\lef s_2\rig s_3\\
&=s_1\lef s_2\rig s_3
\end{align*}
\end{proof}

\begin{lem}\label{trans invariance BA lem}
For all terms $s$, 
\[
\text{BA}\vdash(s^*)'=s
\]
\end{lem}
\begin{proof}
Proof by structural induction on $s$. If $s\in\{T,F\}\cup A\cup V$ then it is trivially true. Suppose $s\syn\neg s_1$. Then
\begin{align*}
((\neg s_1)^*)' &= (F\lef s_1^*\rig T)'\\
&=(\neg(s_1^*)'\vee F)\wedge((s_1^*)'\vee T)\\
&=\neg(s_1^*)'\\
&=_{IH}\neg s_1
\end{align*}
Suppose $s\syn s_1\vee s_2$. Then
\begin{align*}
((s_1\vee s_2)^*)' &= (T\lef s_1^*\rig s_2^*)'\\
&=(\neg(s_1^*)'\vee T)\wedge((s_1^*)'\vee(s_2^*)')\\
&=(s_1^*)'\vee(s_2^*)'\\
&=_{IH}s_1\vee s_2
\end{align*}
Suppose $s\syn s_1\wedge s_2$. Then
\begin{align*}
((s_1\wedge s_2)^*)' &= (s_2^*\lef s_1^*\rig F)'\\
&= (\neg(s_1^*)'\vee(s_2^*)')\wedge((s_1^*)'\vee F)\\
&= (\neg(s_1^*)'\vee(s_2^*)')\wedge(s_1^*)'\\
&= ((s_1^*)'\wedge\neg(s_1^*)')\vee((s_1^*)'\wedge(s_2^*)')\\
&= (s_1^*)'\wedge(s_2^*)'\\
&=_{IH} s_1\wedge s_2
\end{align*}
\end{proof}

The last lemma before proving $\omega$-completeness of $\text{CP}_{st}$ is a variation of Lemma \ref{substitution lemma}.
\begin{lem}\label{trans subst lem}
If for all closed substitutions $\sigma:V\to T(\sig)$,
\[
\text{BA}\vdash\sigma(s)'=\sigma(t)'
\]
Then for all closed substitutions $\tau:V\to T(\Sigma_\text{BA})$,
\[
\text{BA}\vdash\tau(s')=\tau(t')
\]
\end{lem}
\begin{proof}
Assume that $\forall\sigma:\text{BA}\vdash\sigma(s)'=\sigma(t)'$ with $\sigma:V\to T(\sig)$ closed substitutions. Let $\sigma'(x)=_{def}\sigma(x)'$. Since for all $x\in V$, $x'=x$, it follows that $\sigma(s)'\syn\sigma'(s')$. Hence,
\[
\forall\sigma:\text{BA}\vdash\sigma'(s')=\sigma'(t')\qquad(*)
\]
Let $\tau:V\to T(\Sigma_\text{BA})$ be a closed substitution. Furthermore, define $\sigma:V\to\ T(\sig)$ to be the substitution such that $\sigma(x)=\tau(x)^*$ for all $x\in V$. By Lemma \ref{trans invariance BA lem}, it follows that $\text{BA}\vdash\tau(x)=(\tau(x)^*)'=\sigma(x)'=\sigma'(x)$ for all $x$. Hence, $\text{BA}\vdash\tau(s')=\sigma'(s')$ and $\text{BA}\vdash\tau(t')=\sigma'(t')$. Consequently, by (*),
\[
\text{BA}\vdash\tau(s')=\sigma'(s')=\sigma'(t')=\tau(t')
\]
Thus, $\forall\tau:\text{BA}\vdash\tau(s')=\tau(t')$.
\end{proof}
Having done the groundwork, it is now possible to prove that $\text{CP}_{st}$ is $\omega$-complete.
\begin{thm}
Let $s$ and $t$ be terms over $\sig$ such that for all closed substitutions $\sigma$, $\text{CP}_{st}\vdash\sigma(s)=\sigma(t)$, then $\text{CP}_{st}\vdash s=t$.
\end{thm}
\begin{proof}
Assume that
\[
\forall\sigma:\text{CP}_{st}\vdash\sigma(s)=\sigma(t)
\]
where $\sigma$ is closed. By Lemma \ref{CPst in BA lem},
\[
\forall\sigma:\text{BA}\vdash\sigma(s)'=\sigma(t)'
\]
By Lemma \ref{trans subst lem},
\[
\forall\tau:\text{BA}\vdash\tau(s')=\tau(t')
\]
where $\tau$ is closed. By $\omega$-completeness of BA,
\[
\text{BA}\vdash s'=t'
\]
By Lemma \ref{BA in CPst lem},
\[
\text{CP}_{st}\vdash (s')^*=(t')^*
\]
By Lemma \ref{trans invariance CPst lem}
\[
\text{CP}_{st}\vdash s=t
\]
\end{proof}

In this chapter we have shown $\omega$-completeness for $\text{CP}$ and $\text{CP}_{st}$. $\omega$-Completeness of the remaining two axiomatizations, $\text{CP}_{rp}$ and $\text{CP}_{cr}$, remains an open issue with the exception of the case where there are less than two atomic propositions in which neither $\text{CP}$, $\text{CP}_{rp}$ nor $\text{CP}_{cr}$ are $\omega$-complete.

\chapter{Independence of the axioms}
In this chapter we prove that the axioms are independent from each other. An axiom is independent with respect to a set of axioms if it cannot be derived from the other axioms e.g., \CP1 is independent in CP if $\CP2,\CP3,\CP4\nvdash \CP1$. This is a nice property for a set of axioms to have, as it shows that there are no redundant axioms.

The standard strategy for proving that an axiom is independent consists of constructing a model such that every axiom except the one we are trying to prove independence for, is true in this model. In other words, if we want to prove that CP1 is independent in CP, we show there is a model $\mathcal{M}$ and interpretation $\phi$ such that
\begin{itemize}
\item[(1)] $\text{CP2-4}\vdash s=t$ implies that $\mathcal{M}\models \phi(s)=\phi(t)$
\item[(2)] $\mathcal{M}\not\models\phi(\text{CP1})$
\end{itemize}
Hence, by contraposition of (1) it follows that $\text{CP2-4}\nvdash\text{CP1}$. Note that this only applies if $\mathcal{M}$ is a model of equational logic i.e., reflexivity, symmetry, transitivity and congruence are all true in $\mathcal{M}$.

In the following sections independence of axioms is shown for the different varieties of RVAs.

\section{Independence of CP}
Recall that $fr$-congruence is axiomatized by the axioms in CP, listed here again for the reader's convenience:
\[
\begin{array}{lrcl}
(\CP1) & x\lef T\rig y & = & x\\
(\CP2) & x\lef F\rig y & = & y\\
(\CP3) & T\lef x\rig F & = & x\\
(\CP4) & x\lef(y\lef z\rig u)\rig v & = & (x\lef y\rig v)\lef z\rig(x\lef u\rig v)\\
\end{array}
\]
We start by proving independence of \CP1. In order to do this we need to construct a model such that \CP2, \CP3 and \CP4 are true in this model but \CP1 is not. Take a look at the following model:
\begin{align*}
\phi_1(T) &= T\\
\phi_1(F) &= F\\
\phi_1(a) &= T\\
\phi_1(P\lef Q\rig R) &= \phi_1(Q)\vee\phi_1(R)
\end{align*}
for all $a\in A$ with $\vee$ as in sentential logic and as domain $\{T,F\}$. Observe that this is a model of equational logic, in particular congruence is true. The next step in proving independence for \CP1 is to show that \CP1 does not hold under this interpretation. We do this by giving a counterexample i.e., we take a specific instantiation of this axiom and show that the left-hand and the right-hand side of \CP1 are not equal. If CP1 were true in this model then $\phi_1(F\lef T\rig F)=\phi_1(F)$. However, $\phi_1(F\lef T\rig F)=T$, and hence unequal to $\phi_1(F)=F$. Therefore, CP1 does not hold in this model. Now we have to check whether CP2-4 do hold:
\begin{align*}
\phi_1(P\lef F\rig Q) &= F\vee\phi_1(Q)\\
&= \phi_1(Q)\\\\
\phi_1(T\lef P\rig F) &= \phi_1(P)\vee F\\
&=\phi_1(P)\\\\
\phi_1(P\lef(Q\lef R\rig S)\rig V)
&= \phi_1(Q\lef R\rig S)\vee\phi_1(V)\\
&= (\phi_1(R)\vee\phi_1(S))\vee\phi_1(V)\\
&= \phi_1(R)\vee(\phi_1(S)\vee\phi_1(V))\\
&= \phi_1(R)\vee(\phi_1(P\lef S\rig V))\\
&= \phi_1((P\lef Q\rig V)\lef R\rig(P\lef S\rig V))
\end{align*}
Since CP2-4 are true in this model but CP1 is not, we can conclude that CP1 is independent with respect to CP.

Proving independence for the remaining axioms requires that we repeat these steps for each axiom. So let us continue with proving independence of axiom \CP2. This time we construct a model such that it models \CP1, \CP3 and \CP4 but not \CP2.
\begin{align*}
\phi_2(T) &= T\\
\phi_2(F) &= F\\
\phi_2(a) &= T\\
\phi_2(P\lef Q\rig R) &= \phi_2(P)\wedge\phi_2(Q)
\end{align*}
for all $a\in A$. Then we check whether \CP2 does not hold. Since $\phi_2(T\lef F\rig T)=F\ne \phi_2(T)$, this is true. It is easy to see that \CP1 and \CP3 are true in this model. That leaves us with checking \CP4:
\begin{align*}
\phi_2(P\lef(Q\lef R\rig S)\rig V)
&=\phi_2(P)\wedge\phi_2(Q\lef R\rig S)\\
&=\phi_2(P)\wedge(\phi_2(Q)\wedge\phi_2(R))\\
&=(\phi_2(P)\wedge\phi_2(Q))\wedge\phi_2(R)\\
&=\phi_2(P\lef Q\rig V)\wedge\phi_2(R)\\
&=\phi_2((P\lef Q\rig V)\lef R\rig(P\lef S\rig V))
\end{align*}
So \CP4 also holds under this interpretation, and thus we know that \CP2 is independent.

The previous two models looked quite similar, in particular the models share the same domain i.e., $\{T,F\}$. In proving the independence of \CP3 we will show that this is not always the case. The model we will be constructing here has a finite set of natural numbers as domain, and as a result differs quite a bit from the standard semantics.

The construction of this model requires that we first enumerate the atomic propositions in the set $A$:
\[
a_1,a_2,\ldots,a_n
\]
Using this enumeration we can define our model: 
\begin{align*}
\phi_3(T) &= 0\\
\phi_3(F) &= n+1\\
\phi_3(a_i) &= i\\
\phi_3(P\lef Q\rig R) &= \begin{cases}
\phi_3(P) & \text{if $\phi_3(Q)\le 1$}\\
\phi_3(R) & \text{if $\phi_3(Q)> 1$}
\end{cases}
\end{align*}
Note that congruence is trivially true in this model. By definition there is at least one atomic proposition i.e., $A\ne\emptyset$. So, we can always assume that $a_1$ exists. Observe that $\phi_3(T\lef a_1\rig F)=\phi_3(T)\ne\phi_3(a_1)$. Hence, \CP3 does not follow. Checking to see that \CP1 and \CP2 are true is trivial. That leaves \CP4:
\begin{align*}
\phi_3(P\lef(Q\lef R\rig S)\rig U)
&=\begin{cases}
\phi_3(P)&\text{if $\phi_3(Q\lef R\rig S)\le 1$}\\
\phi_3(U)&\text{if $\phi_3(Q\lef R\rig S)>1$}
\end{cases}\\
&=\begin{cases}
\phi_3(P)&\text{if $\phi_3(R)\le 1$ and $\phi_3(Q)\le 1$}\\
\phi_3(P)&\text{if $\phi_3(R)> 1$ and $\phi_3(S)\le 1$}\\
\phi_3(U)&\text{if $\phi_3(R)\le 1$ and $\phi_3(Q)>1$}\\
\phi_3(U)&\text{if $\phi_3(R)> 1$ and $\phi_3(S)>1$}\\
\end{cases}\\
&=\begin{cases}
\phi_3(P\lef Q\rig U)&\text{if $\phi_3(R)\le 1$ and $\phi_3(Q)\le 1$}\\
\phi_3(P\lef S\rig U)&\text{if $\phi_3(R)> 1$ and $\phi_3(S)\le 1$}\\
\phi_3(P\lef Q\rig U)&\text{if $\phi_3(R)\le 1$ and $\phi_3(Q)>1$}\\
\phi_3(P\lef S\rig U)&\text{if $\phi_3(R)> 1$ and $\phi_3(S)>1$}\\
\end{cases}\\
&=\begin{cases}
\phi_3(P\lef Q\rig U)&\text{if $\phi_3(R)\le 1$}\\
\phi_3(P\lef S\rig U)&\text{if $\phi_3(R)> 1$}\\
\end{cases}\\
&=\phi_3((P\lef Q\rig U)\lef R\rig(P\lef S\rig U))
\end{align*}
So \CP3 is also independent.

Let $A=\{a_1,\ldots,a_n\}$. Take the following interpretation:
\begin{align*}
\phi_4(T) &=1\\
\phi_4(F) &=0\\
\phi_4(a_i) &=i+1\\
\phi_4(P\lef Q\rig R) &=\begin{cases}
\phi_4(P) & \text{if $\phi_4(Q)=1$}\\
\phi_4(R) & \text{if $\phi_4(Q)=0$}\\
\phi_4(P)\cdot\phi_4(Q) & \text{otherwise}
\end{cases}
\end{align*}
Clearly, congruence and CP1-3 are true in this model. Furthermore, it follows that $\phi_4(F\lef a_1\rig T)=\phi_4(F)\phi(a_1)=0$. Hence,
\begin{align*}
\phi_4(T\lef(F\lef a_1\rig T)\rig T)
&=\phi_4(T)\\
&=1
\end{align*}
However,
\begin{align*}
\phi_4((T\lef F\rig T)\lef a_1\rig (T\lef T\rig T))
&=\phi_4(T\lef F\rig T)\phi_4(a_1)\\
&=2
\end{align*}
So CP4 is not true using this interpretation. Hence axiom CP4 is also independent with respect to CP. Since this is the last axiom in CP, we have now shown independence for all the axioms in CP. Consequently, there are no redundant axioms in CP. In the next section we will be looking at an extension of CP i.e., the axiomatization of $rp$-congruence.

\section{Independence of $\text{CP}_{rp}$}
The axiomatization of $rp$-congruence is an extension of CP with the following axioms:
\[
\begin{array}{lrcl}
(\CPrp1) & (x\lef a\rig y)\lef a\rig z & = & (x\lef a\rig x)\lef a\rig z\\
(\CPrp2) & x\lef a\rig (y\lef a\rig z) & = & x\lef a\rig(z\lef a\rig z)
\end{array}
\]
Note that these are actually axiom schemes i.e., for each $a\in A$ we have an axiom CPrp1 and CPrp2. Since we have a new set of axioms, we are required, in addition to proving the independence of the two new axioms \CPrp1 and \CPrp2, to prove that CP1-4 is independent with respect to this new set of axioms. Fortunately, it is possible to reuse the models used in the previous section i.e., for the cases CP1-4 the same models are taken. It then suffices for these cases to show that \CPrp1 and \CPrp2 are true in these models.

We start by taking the same model as we did in the previous section for proving the independence of axiom \CP1, and then checking if it models \CPrp1 and \CPrp2.
\begin{align*}
\phi_1((P\lef a\rig Q)\lef a\rig R)
&=\phi_1(a)\vee\phi(R)\\
&=\phi_1((P\lef a\rig P)\lef a\rig R)\\\\
\phi_1(P\lef a\rig(Q\lef a\rig R))
&=\phi_1(a)\vee\phi_1(Q\lef a\rig R)\\
&=\phi_1(a)\vee(\phi_1(a)\vee\phi_1(R))\\
&=\phi_1(a)\vee\phi_1(R\lef a\rig R)\\
&=\phi_1(P\lef a\rig(R\lef a\rig R))
\end{align*}

Repeat this procedure for the models given for the independence-models of \CP2, \CP3 and \CP4.
\begin{align*}
\phi_2((P\lef a\rig Q)\lef a\rig R)
&=\phi_2(P\lef a\rig Q)\wedge\phi_2(a)\\
&=(\phi_2(P)\wedge\phi_2(a))\wedge\phi_2(a)\\
&=\phi_2(P\lef a\rig P)\wedge\phi_2(a)\\
&=\phi_2((P\lef a\rig P)\lef a\rig R)\\\\
\phi_2(P\lef a\rig(Q\lef a\rig R))
&=\phi_2(P)\wedge\phi_2(a)\\
&=\phi_2(P\lef a\rig(R\lef a\rig R))
\end{align*}

We only show that \CPrp1 holds in the model given for case \CP3 because the proof that \CPrp2 is true in this model is symmetric to that of \CPrp1.
\begin{align*}
\phi_3((P\lef a\rig Q)\lef a\rig R)
&=\begin{cases}
\phi_3(P\lef a\rig Q) & \text{if $\phi_3(a)\le 1$}\\
\phi_3(R) & \text{if $\phi_3(a)> 1$}
\end{cases}\\
&=\begin{cases}
\phi_3(P) & \text{if $\phi_3(a)\le 1$}\\
\phi_3(R) & \text{if $\phi_3(a)> 1$}
\end{cases}\\
&=\begin{cases}
\phi_3(P\lef a\rig P) & \text{if $\phi_3(a)\le 1$}\\
\phi_3(R) & \text{if $\phi_3(a)> 1$}
\end{cases}\\
&=\phi_3((P\lef a\rig P)\lef a\rig R)
\end{align*}

The following shows that CP4 is also independent in $\text{CP}_{rp}$.
\begin{align*}
\phi_4((P\lef a\rig Q)\lef a\rig R)
&= \phi_4(P\lef a\rig Q)\phi_4(a)\\
&= \phi_4(P)\phi_4(a)\phi_4(a)\\
&= \phi_4(P\lef a\rig P)\phi_4(a)\\
&= \phi_4((P\lef a\rig P)\lef a\rig R)\\\\
\phi_4(P\lef a\rig(Q\lef a\rig R))
&= \phi_4(P)\phi(a)\\
&= \phi_4(P\lef a\rig(Q\lef a\rig R))
\end{align*}

The model for showing independence of CPrp1 is based on the reactive valuation variety that satisfies
\[
y_a(x)=F\qquad\Longrightarrow\qquad y_a(\dd a(x))=F
\]
We call this variety $rp1$. By definition, this is a subvariety of the variety with free reactive valuations. Thus, by soundness of free reactive valuations, it follows that the resulting congruence $=_{rp1}$ (constructed similarly as $=_{fr}$, $=_{rp}$, etc.) is a model for CP. Hence, we do not have to check whether CP1-4 are true in this model.

If $a/H=T$, it is possible that $a/H\ne a/\dd a(H)$ for some $\mathbb{A}\in rp1$ and $H\in\mathbb{A}$. Consequently,
\begin{align*}
((T\lef a\rig F)\lef a\rig F)/H
&=\begin{cases}
T & \text{if $a/H=T$ and $a/\dd a(H)=T$}\\
F & \text{if $a/H=T$ and $a/\dd a(H)=F$}\\
F & \text{if $a/H=F$}
\end{cases}\\
&\ne\begin{cases}
T & \text{if $a/H=T$ and $a/\dd a(H)=T$}\\
T & \text{if $a/H=T$ and $a/\dd a(H)=F$}\\
F & \text{if $a/H=F$}
\end{cases}\\
&=((T\lef a \rig T)\lef a\rig F)/H
\end{align*}
Thus, CPrp1 is not true. However, CPrp2 is.
\begin{align*}
(P\lef a\rig(Q\lef a\rig R))/H
&=\begin{cases}
P/\dd a(H) & \text{if $a/H=T$}\\
(Q\lef a\rig R)/\dd a(H) & \text{if $a/H=F$}\\
\end{cases}\\
&=\begin{cases}
P/\dd a(H) & \text{if $a/H=T$}\\
Q/\dd a(\dd a(H)) & \text{if $a/H=F$ and $a/\dd a(H)=T$}\\
R/\dd a(\dd a(H)) & \text{if $a/H=F$ and $a/\dd a(H)=F$}
\end{cases}\\
&=\begin{cases}
P/\dd a(H) & \text{if $a/H=T$}\\
Q/\dd a(\dd a(H)) & \text{if $a/H=F$ and $F=T$}\\
R/\dd a(\dd a(H)) & \text{if $a/H=F$ and $F=F$}
\end{cases}\\
&=\begin{cases}
P/\dd a(H) & \text{if $a/H=T$}\\
R/\dd a(\dd a(H)) & \text{if $a/H=F$ and $F=T$}\\
R/\dd a(\dd a(H)) & \text{if $a/H=F$ and $F=F$}
\end{cases}\\
&=\begin{cases}
P/\dd a(H) & \text{if $a/H=T$}\\
R/\dd a(\dd a(H)) & \text{if $a/H=F$ and $a/\dd a(H)=T$}\\
R/\dd a(\dd a(H)) & \text{if $a/H=F$ and $a/\dd a(H)=F$}
\end{cases}\\
&=(P\lef a\rig(R\lef a\rig R)/H
\end{align*}
\begin{align*}
\dd{(P\lef a\rig(Q\lef a\rig R))}(H)
&=\begin{cases}
\dd P(\dd a(H)) & \text{if $a/H=T$}\\
\dd Q(\dd a(\dd a(H))) & \text{if $a/H=F$ and $a/\dd a(H)=T$}\\
\dd R(\dd a(\dd a(H))) & \text{if $a/H=F$ and $a/\dd a(H)=F$}\\
\end{cases}\\
&=\begin{cases}
\dd P(\dd a(H)) & \text{if $a/H=T$}\\
\dd R(\dd a(\dd a(H))) & \text{if $a/H=F$ and $a/\dd a(H)=T$}\\
\dd R(\dd a(\dd a(H))) & \text{if $a/H=F$ and $a/\dd a(H)=F$}\\
\end{cases}\\
&=\dd{(P\lef a\rig(R\lef a\rig R))}(H)
\end{align*}
Note that both $a/H=F$ and $a/\dd a(H)=T$ is impossible. Hence, we can replace $\dd Q(\dd a(\dd a(H)))$ with $\dd R(\dd a(\dd a(H)))$ in the above derivation.

The proof for showing independence of CPrp2 is symmetric to the one in CPrp1, using the reactive valuation variety that satisfies
\[
y_a(x)=T\qquad\Longrightarrow\qquad y_a(\dd a(x))=T
\]
We call this variety $rp2$ and we will use this variety in the next section.

\section{Independence of $\text{CP}_{cr}$}
The axiomatization of $cr$-congruence consists of CP and the following axioms
\[
\begin{array}{lrcl}
(\CPcr1) & (x\lef a\rig y)\lef a\rig z & = & x\lef a\rig z\\
(\CPcr2) & x\lef a\rig (y\lef a\rig z) & = & x\lef a\rig z
\end{array}
\]
Like in the previous section, it is possible to reuse the models given for CP, and just show that \CPcr1 and \CPcr2 are true in these models. Starting with \CP1:
\begin{align*}
\phi_1((P\lef a\rig Q)\lef a\rig R)
&=\phi_1(a)\vee\phi_1(R)\\
&=\phi_1(P\lef a\rig R)\\\\
\phi_1(P\lef a\rig(Q\lef a\rig R))
&=\phi_1(a)\vee\phi_1(Q\lef a\rig R)\\
&=\phi_1(a)\vee(\phi_1(a)\vee\phi_1(R))\\
&=\phi_1(a)\vee\phi_1(R)\\
&=\phi_1(P\lef a\rig R)
\end{align*}
For \CP2:
\begin{align*}
\phi_2((P\lef a\rig Q)\lef a\rig R)
&=\phi_2(P\lef a\rig Q)\wedge\phi_2(a)\\
&=(\phi_2(P)\wedge\phi_2(a))\wedge\phi_2(a)\\
&=\phi_2(P)\wedge\phi_2(a)\\
&=\phi_2(P\lef a\rig R)\\\\
\phi_2(P\lef a\rig(Q\lef a\rig R))
&=\phi_2(P)\wedge\phi_2(a)\\
&=\phi_2(P\lef a\rig R)
\end{align*}

Similar to the previous section we omit the proof that \CPcr2 is true in the independence-model for \CP3 as it is symmetric to that of \CPcr1.
\begin{align*}
\phi_3((P\lef a\rig Q)\lef a\rig R)
&=\begin{cases}
\phi_3(P\lef a\rig Q) & \text{if $\phi_3(a)\le 1$}\\
\phi_3(R) & \text{if $\phi_3(a)>1$}
\end{cases}\\
&=\begin{cases}
\phi_3(P) & \text{if $\phi_3(a)\le 1$}\\
\phi_3(R) & \text{if $\phi_3(a)>1$}
\end{cases}\\
&=\phi_3(P\lef a\rig R)
\end{align*}

Unfortunately, we have not been able to find a model that demonstrates the independence of CP4 in $\text{CP}_{cr}$. The model we used in previous sections does not work in this variety. For example, we have the following.
\begin{align*}
\phi_4((T\lef a_1\rig F)\lef a_1\rig T)
&= 4\\
&\ne 2\\
&=\phi_4(T\lef a\rig T)
\end{align*}
Hence, CPcr1 is not true in the resulting model, and thus $\phi_4$ does not suffice. The question whether CP4 is independent in $\text{CP}_{cr}$ remains therefore open.

In order to construct the model that shows independence of CPcr1, we take the variety of all algebras from variety $rp1$ that satisfy
\[
\dd a(\dd a(x)) =\dd a(x)
\]
for all $a\in A$. We call this variety $cr1$. It follows that there is an $\mathbb{A}\in cr1$ and $H\in\mathbb{A}$ such that 
\begin{align*}
((T\lef a\rig F)\lef a\rig T)/H
&=\begin{cases}
T/H & \text{if $a/H=T$ and $a/\dd a(H)=T$}\\
F/H & \text{if $a/H=T$ and $a/\dd a(H)=F$}\\
F/H & \text{if $a/H=F$}
\end{cases}\\
&\ne\begin{cases}
T/H & \text{if $a/H=T$}\\
F/H & \text{if $a/H=F$}
\end{cases}\\
&=(T\lef a\rig T)
\end{align*}
Checking CPcr2:
\begin{align*}
(P\lef a\rig(Q\lef a\rig R))/H
&=\begin{cases}
P/\dd a(H) & \text{if $a/H=T$}\\
Q/\dd a(\dd a(H)) & \text{if $a/H=F$ and $a/\dd a(H)=T$}\\
R/\dd a(\dd a(H)) & \text{if $a/H=F$ and $a/\dd a(H)=F$}
\end{cases}\\
&=\begin{cases}
P/\dd a(H) & \text{if $a/H=T$}\\
R/\dd a(H) & \text{if $a/H=F$}
\end{cases}\\
&=(P\lef a\rig R)/H
\end{align*}
\begin{align*}
\dd{(P\lef a\rig(Q\lef a\rig R))}(H)
&=\begin{cases}
\dd P(\dd a(H)) & \text{if $a/H=T$}\\
\dd Q(\dd a(\dd a(H))) & \text{if $a/H=F$ and $a/\dd a(H)=T$}\\
\dd R(\dd a(\dd a(H))) & \text{if $a/H=F$ and $a/\dd a(H)=F$}
\end{cases}\\
&=\begin{cases}
\dd P(\dd a(H)) & \text{if $a/H=T$}\\
\dd R(\dd a(\dd a(H))) & \text{if $a/H=F$}
\end{cases}\\
&=\begin{cases}
\dd P(\dd a(H)) & \text{if $a/H=T$}\\
\dd R(\dd a(H)) & \text{if $a/H=F$}
\end{cases}\\
&=\dd{(P\lef a\rig R)}(H)
\end{align*}

A proof of independence for CPcr2, starts by taking the variety of all algebras from variety $rp2$ (see previous section) that satisfy
\[
\dd a(\dd a(x)) = \dd a(x)
\]
for $a\in A$. The proof for showing independence of CPcr2 using this variety is symmetric to the previous proof of independence for CPcr1.

\section{Independence of $\text{CP}_{st}$}
In this section we show independence of the axioms CP2, CP3, CPstat and CPcontr.
\[
\begin{array}{lrcl}
(\CPstat) & (x\lef y\rig z)\lef u\rig v & = & (x\lef u\rig v)\lef y\rig(z\lef u\rig v)\\
(\CPcontr) & (x\lef y\rig z)\lef y\rig u &= & x\lef y\rig u
\end{array}
\]
The models we used in the previous sections to show independence of CP1 and CP4 cannot be used here because CPstat and CPcontr are not both true in these models. We give two counterexamples to show this. The first counterexample shows that CPstat is not true in the independence-model for CP1. By CPstat the terms $(T\lef T\rig T)\rig F\rig F$ and $(T\lef F\rig F)\lef T\rig (T\lef F\rig F)$ should be equal. However, this is not the case.
\begin{align*}
\phi_1((T\lef T\rig T)\lef F\rig F)
&=\phi_1(F)\vee\phi_1(F)\\
&= F
\end{align*} 
\begin{align*}
\phi_1((T\lef F\rig F)\lef T\rig(T\lef F\rig F))
&=\phi_1(T)\vee\phi_1(F)\vee\phi_1(F)\\
&=T
\end{align*}
The second counterexample shows that CPcontr is not true in the model we used for showing independence of CP4.
\begin{align*}
\phi_4(T\lef a_1\rig F)\lef a_1\rig T)
&= 4\\
&\ne 2\\
&=\phi_4(T\lef a_1\rig T)
\end{align*}
Proving independence for both CP1 and CP4 in $\text{CP}_{st}$ remains an open issue.

We can use the same models we used in the previous section for showing independence of CP2 and CP3.
\begin{align*}
\phi_2((P\lef Q\rig R)\lef S\rig V)
&=\phi_2(P)\wedge\phi_2(Q)\wedge\phi_2(S)\\
&=\phi_2(P)\wedge\phi_2(S)\wedge\phi_2(Q)\\
&=\phi_2((P\lef S\rig V)\lef Q\rig(R\lef S\rig V))\\\\
\phi_2((P\lef Q\rig R)\lef Q\rig S)
&=\phi_2(P)\wedge\phi_2(Q)\wedge\phi_2(Q)\\
&=\phi_2(P)\wedge\phi_2(Q)\\
&=\phi_2(P\lef Q\rig S)
\end{align*}

\begin{align*}
\phi_3((P\lef Q\rig R)\lef S\rig V)
&=\begin{cases}
\phi_3(P\lef Q\rig R) & \text{if $\phi_3(S)\le 1$}\\
\phi_3(V) & \text{if $\phi_3(S)>1$}
\end{cases}\\
&=\begin{cases}
\phi_3(P) & \text{if $\phi_3(S)\le 1$ and $\phi_3(Q)\le 1$}\\
\phi_3(R) & \text{if $\phi_3(S)\le 1$ and $\phi_3(Q)>1$}\\
\phi_3(V) & \text{if $\phi_3(S)>1$}
\end{cases}\\
&=\begin{cases}
\phi_3(P) & \text{if $\phi_3(S)\le 1$ and $\phi_3(Q)\le 1$}\\
\phi_3(R) & \text{if $\phi_3(S)\le 1$ and $\phi_3(Q)>1$}\\
\phi_3(V) & \text{if $\phi_3(S)>1$ and $\phi_3(Q)\le 1$}\\
\phi_3(V) & \text{if $\phi_3(S)>1$ and $\phi_3(Q)>1$}
\end{cases}\\
&=\begin{cases}
\phi_3(P\lef S\rig V) & \text{if $\phi_3(Q)\le 1$}\\
\phi_3(R\lef S\rig V) & \text{if $\phi_3(Q)>1$}
\end{cases}\\\\
\phi_3((P\lef Q\rig R)\lef Q\rig S)
&=\begin{cases}
\phi_3(P\lef Q\rig R) & \text{if $\phi_3(Q)\le 1$}\\
\phi_3(S) & \text{if $\phi_3(Q)>1$}
\end{cases}\\
&=\begin{cases}
\phi_3(P) & \text{if $\phi_3(Q)\le 1$}\\
\phi_3(S) & \text{if $\phi_3(Q)>1$}
\end{cases}\\
&=\phi_3(P\lef Q\rig S)
\end{align*}

Showing independence of CPstat requires that we define the following subvariety of $fr$ which consists of all RVAs that satisfy
\[
Q/\dd Q(x)=Q/x
\]
and
\[
\dd P(\dd P(x)) = \dd P(x)
\]
Note that this is a generalization of the variety $cr$. We show that CPstat is not true in this variety by constructing an algebra $\mathbb{A}$ and valuation $H\in\mathbb{A}$ such that $a/H=F$, $b/H=T$, $a/\dd b(H)=T$ and $b/\dd a(H)=T$. Then
\begin{align*}
((F\lef a\rig T)\lef b\rig F)/H
&=\begin{cases}
F & \text{if $b/H=T$ and $a/\dd b(H)=T$}\\
T & \text{if $b/H=T$ and $a/\dd b(H)=F$}\\
F & \text{if $b/H=F$}
\end{cases}\\
&=F\\
&\ne T\\
&=\begin{cases}
F & \text{if $a/H=T$ and $b/\dd a(H)=T$}\\
F & \text{if $a/H=T$ and $b/\dd a(H)=F$}\\
T & \text{if $a/H=F$ and $b/\dd a(H)=T$}\\
F & \text{if $a/H=F$ and $b/\dd a(H)=F$}
\end{cases}\\
&=((F\lef b\rig F)\lef a\rig(T\lef b\rig F))/H
\end{align*}
Hence, CPstat is not true. Since this is a subvariety of $fr$, by Theorem \ref{soundness cp} it suffices to prove that CPcontr is true in this model.
\begin{align*}
((P\lef Q\rig R)\lef Q\rig S)/H
&=\begin{cases}
P/\dd Q(\dd Q(H)) & \text{if $Q/H=T$ and $Q/\dd Q(H)=T$}\\
R/\dd Q(\dd Q(H)) & \text{if $Q/H=T$ and $Q/\dd Q(H)=F$}\\
S/\dd Q(H) & \text{if $Q/H=F$}
\end{cases}\\
&=\begin{cases}
P/\dd Q(H) & \text{if $Q/H=T$ and $Q/H=T$}\\
R/\dd Q(H) & \text{if $Q/H=T$ and $Q/H=F$}\\
S/\dd Q(H) & \text{if $Q/H=F$}
\end{cases}\\
&=\begin{cases}
P/\dd Q(H) & \text{if $Q/H=T$}\\
S/\dd Q(H) & \text{if $Q/H=F$}
\end{cases}\\
&=(P\lef Q\rig S)/H\\\\
\dd{((P\lef Q\rig R)\lef Q\rig S)}(H)
&=\begin{cases}
\dd P(\dd Q(\dd Q(H))) & \text{if $Q/H=T$ and $Q/\dd Q(H)=T$}\\
\dd R(\dd Q(\dd Q(H))) & \text{if $Q/H=T$ and $Q/\dd Q(H)=F$}\\
\dd S(\dd Q(H)) & \text{if $Q/H=F$}
\end{cases}\\
&=\begin{cases}
\dd P(\dd Q(H)) & \text{if $Q/H=T$}\\
\dd S(\dd Q(H)) & \text{if $Q/H=F$}
\end{cases}\\
&=\dd{(P\lef Q\rig S)}(H)
\end{align*}
So by Proposition \ref{equiv congr prop} CPcontr is true in this model. Since CP and CPcontr are true and CPstat is not, it follows that CPstat is independent in $\text{CP}_{st}$.

The model that shows independence of CPcontr has the integers as its domain. Similar to the independence-model for CP3, we assume that the set $A$ is enumerated i.e., $A=\{a_1,a_2,\ldots,a_n\}$.
\begin{align*}
\phi(T) &= 1\\
\phi(F) &= 0\\
\phi(a_i) &= i+1\\
\phi(P\lef Q\rig R) &= \phi(Q)\cdot\phi(P)+(1-\phi(Q))\cdot\phi(R)
\end{align*}
CPcontr is not true in this model:
\begin{align*}
\phi((T\lef a_1\rig F)\lef a_1\rig F)
&=\phi(a_1)\phi(T\lef a_1\rig F)+(1-\phi(a_1))\phi(F)\\
&=2(\phi(a_1)\phi(T)+(1-\phi(a_1))\phi(F))\\
&=4\\
&\ne 2\\
&=\phi(a_1)\phi(T)+(1-\phi(a_1))\phi(F)\\
&=\phi(T\lef a_1\rig F)
\end{align*}

The following derivations show that this is a model for CP1-4 and CPstat. The first three derivations are easy.
\begin{align*}
\phi(P\lef T\rig Q) &= \phi(T)\phi(P) + (1-\phi(T))\phi(Q)\\
&= 1\cdot\phi(P) + 0\cdot\phi(Q)\\
&= \phi(P)
\end{align*}

\begin{align*}
\phi(P\lef F\rig Q) &= \phi(F)\phi(P) + (1-\phi(F))\phi(Q)\\
&= 0\cdot\phi(P) + 1\cdot\phi(Q)\\
&= \phi(Q)
\end{align*}

\begin{align*}
\phi(T\lef P\rig F) &= \phi(P)\phi(T) + (1-\phi(P))\phi(F)\\
&=\phi(P)\cdot 1+ (1-\phi(P))\cdot 0\\
&=\phi(P)
\end{align*}
Checking whether CP4 and CPstat are true in this model requires some bookkeeping.
\begin{align*}
\phi(P\lef(Q\lef R\rig S)\rig V)
&= \phi(Q\lef R\rig S)\phi(P)+(1-\phi(Q\lef R\rig S))\phi(V)\\
&= (\phi(R)\phi(Q)+(1-\phi(R))\phi(S))\phi(P)\\
&\quad +(1-(\phi(R)\phi(Q)+(1-\phi(R))\phi(S)))\phi(V)\\
&= \phi(R)\phi(Q)\phi(P) + \phi(S)\phi(P) - \phi(R)\phi(S)\phi(P)\\
&\quad  + \phi(V) - \phi(R)\phi(Q)\phi(V) - \phi(S)\phi(V)\\
&\quad  + \phi(R)\phi(S)\phi(V)\\
&= \phi(R)\phi(Q)\phi(P)+\phi(R)\phi(V)-\phi(R)\phi(Q)\phi(V)\\
&\quad + \phi(S)\phi(P) + \phi(V) - \phi(S)\phi(V)\\
&\quad - \phi(R)\phi(S)\phi(P) - \phi(R)\phi(V) + \phi(R)\phi(S)\phi(V)\\
&= \phi(R)(\phi(Q)\phi(P)+(1-\phi(Q))\phi(V))\\
&\quad  + (1-\phi(R))(\phi(S)\phi(P)+(1-\phi(S))\phi(V))\\
&= \phi(R)\phi(P\lef Q\rig V)+(1-\phi(R))\phi(P\lef S\rig V)\\
&= \phi((P\lef Q\rig V)\lef R\rig(P\lef S\rig V))
\end{align*}

\begin{align*}
\phi((P\lef Q\rig R)\lef S\rig V)
&=\phi(S)(\phi(Q)\phi(P)+(1-\phi(Q))\phi(R))+(1-\phi(S))\phi(V)\\
&=\phi(S)\phi(Q)\phi(P)+\phi(S)\phi(R)-\phi(S)\phi(Q)\phi(R)\\
&\quad +\phi(V)-\phi(V)\phi(S)\\
&=\phi(Q)\phi(S)\phi(P)+\phi(S)\phi(R)+\phi(V)-\phi(S)\phi(V)\\
&\quad -\phi(Q)\phi(S)\phi(R)\\
&=\phi(Q)\phi(S)\phi(P)+\phi(Q)\phi(V)-\phi(Q)\phi(S)\phi(V)\\
&\quad +\phi(S)\phi(R)+\phi(V)-\phi(S)\phi(V)-\phi(Q)\phi(S)\phi(R)\\
&\quad -\phi(Q)\phi(V)+\phi(Q)\phi(S)\phi(V)\\
&=\phi(Q)(\phi(S)\phi(P)+(1-\phi(S))\phi(V))\\
&\quad +(1-\phi(Q))(\phi(S)\phi(R)+(1-\phi(S))\phi(V))\\
&=\phi((P\lef S\rig V)\lef Q\rig(R\lef S\rig V))
\end{align*}
Since CP and CPstat are true in the model and CPcontr is not, CPcontr is independent.

\chapter{Conclusion}
In this final chapter we give a short summary of the previous chapters, thereby listing some of the main results. Furthermore, we discuss the open issues mentioned in the previous chapters, and finally give some suggestions for future work.

\section{Summary}
Sentential logic is limited by the static behaviour of its valuations. In Chapter~1 we introduced the reader to reactive valuations. These reactive valuations, first defined by Bergstra and Ponse in \cite{main}, are an extension of normal valuations because they allow us to take the evaluation order of the expression into account. By means of a few examples we illustrated the advantages of using reactive valuations instead of normal valuations. Similarly, we also revealed some of the limitations of using reactive valuations. At the end of the introduction section, we showed that it is possible to define several classes of reactive valuations depending on their behaviour.

As motivation for looking at reactive valuations we argued that these are interesting because they can be used to model a variety of sequential processes. In the section on motivation we provided a few examples of sequential behaviour from e.g. computer science and linguistics.

Since reactive valuations are a recent invention by Berstra and Ponse, there is no directly related work on reactive valuations, besides the main paper \cite{main}. We, therefore, opted to list some broad areas which might pertain to reactive valuations e.g., non-monotonic reasoning, program semantics and many-valued logics, thereby giving a few specific examples. 

In Chapter~2 we defined the axiomatization of reactive valuation congruence, called CP, and its semantics. The underlying semantics consists of several parts. In the first part we described reactive valuation algebras (RVAs). In the second part we showed how we can compute the value of closed term $P$, given an RVA $\mathbb{A}$ and reactive valuation $H\in\mathbb{A}$. The value is denoted as $P/H$. By imposing limitations on the RVAs and their valuations we can define several varieties of RVAs, namely the varieties with free reactive valuations, repetition-proof valuations, contractive valuations and static valuations. The signature of all these varieties is the same, and consists of constants $T$ and $F$, an infinite set of variables, a finite set of atomic proposition symbols and the ternary operator $\_\lef\_\rig\_$ denoting conditional composition. Each variety has its own axiomatization, where the axiomatization CP corresponds to the variety with free reactive valuations. Given a variety $K$, we say two closed terms $P$ and $Q$ are $K$-equivalent if $P/H=Q/H$ for RVAs $\mathbb{A}\in K$ and valuations $H\in\mathbb{A}$. Unfortunately, $K$-equivalence does not necessarily have the congruence property. We, therefore, define $K$-congruence as the largest congruence contained in the $K$-equivalence relation. 

With the aim of showing soundness for the various varieties it sufficed to show that each axiom is also true under semantical congruence. So if $A=B$ is an axiom of variety $K$, we need only show that $A$ and $B$ are $K$-congruent. In these proofs we took advantage of the fact that if an axiom is sound with respect to a variety $K$ then the axiom is also sound in all subvarieties of $K$.

In order to show completeness we introduced basic forms, where each sentence is provably equal to a unique basic form. The main advantage of basic forms is that their syntactic structure is quite simple. Consequently, structural induction on the set of basic forms is relatively easy. By showing that for basic forms syntactical equality and semantic congruence coincide we were able to prove completeness. Each variety requires its own set of basic forms.

Given an axiomatization, if for all closed substitutions $\sigma$ and terms $s$ and $t$ we can derive $\sigma(s)=\sigma(t)$, we can also derive $s=t$, we say that this axiomatization is $\omega$-complete. In Chapter~3 we discussed $\omega$-completeness of CP and $\text{CP}_{st}$, the axiomatization of $st$-congruence.

Similar to the completeness proofs in Chapter~2, we defined a special set of terms, namely the set of open basic forms. As opposed to the various sets of basic forms used to prove completeness, the open basic forms may contain variables. Using these open basic forms we were able to prove $\omega$-completeness for CP.

For $\text{CP}_{st}$, we used a different approach. This approach does not rely on a set of open basic forms but on a translation between $\text{CP}_{st}$ and a specification of boolean algebra for which we know that it is $\omega$-complete. Using this translation we show that the $\omega$-completeness property transfers to $\text{CP}_{st}$. We did not show $\omega$-completeness for the other varieties.

Independence of an axiom with respect to a particular axiomatization entails that the axiom is not redundant in that set of axioms i.e., it is not derivable from the other axioms. In Chapter~4 we showed independence of axioms with respect to the various axiomatizations. Showing that an axiom $P$ is independent of a set $A$ of other axioms requires that we construct a model such that $P$ is not true in this model but the axioms in $A$ are. There were a few axioms for which we did not succeed in proving that they were independent e.g., CP4 for contractive valuations.

\section{Open issues and future work}
In the past chapters several specific open issues were mentioned. We will list and discuss these issues in this section. Afterwards, we give some general suggestions for future work on the subject of reactive valuations.

\subsection{Open issues}

The first set of open issues is mentioned in the chapter on $\omega$-completeness. We were not able to establish $\omega$-completeness for $\text{CP}_{rp}$ and $\text{CP}_{cr}$. In Chapter~3 we showed $\omega$-completeness for CP and $\text{CP}_{st}$ using two different methods.

The method we used for $\text{CP}_{st}$ involved a translation between an equivalent $\omega$-complete axiomatization and $\text{CP}_{st}$. Unfortunately, this will not work for $\text{CP}_{rp}$ and $\text{CP}_{cr}$ because finding an equivalent $\omega$-complete axiomatization for these axiomatizations is unlikely.

Consequently, the best approach seems to be the one we used for CP, where we used open basic forms. However, first attempts at using this method failed to yield a positive result. The open basic forms as defined in Definition \ref{open basic form def} have the nice property that if we substitute each variable in an open basic form with an atomic proposition we end up with a closed basic form i.e., a member of \BF\ (see Definition \ref{basic form def}). We use this property and the results we already have for \BF\ in the lemmas leading up to the $\omega$-completeness theorem. The problems arise when we define similar open basic forms for $\text{CP}_{rp}$ and $\text{CP}_{cr}$. For example, take the term
\[
(T\lef a\rig F)\lef x\rig T
\]
At first glance this seems like an excellent candidate for the set of open basic forms of both $\text{CP}_{rp}$ and $\text{CP}_{cr}$. However, if we substitute the $x$ with atomic proposition $a$, the resulting term is neither a member of $\BF_{rp}$ nor of $\BF_{cr}$. Whether it is possible to work around this problem remains to be seen. Additional tools for proving $\omega$-completeness can be found in \cite{omega}.

The second set of open issues concerns the independence of axioms. We failed to show independence of CP4 in $\text{CP}_{cr}$. For the axiomatization $\text{CP}_{st}$ we only showed independence of CP2, CP3, CPstat and CPcontr, which leaves CP1 and CP4.

We showed the independence of an axiom with respect to a particular axiomatization by constructing a model such that the axiom in question is not true but the rest of the axioms in the axiomatization are. Constructing such a model is regrettably a hit-or-miss affair and becomes increasingly more difficult as the number of axioms multiply. We can, however, eliminate some options. In Chapter~3 we used roughly three methods of model constructions. 

The first method involves using normal valuations as we know them from sentential logic. The constants $T$ and $F$ are mapped to \emph{true} and \emph{false}, and similarly the atomic propositions are mapped to either \emph{true} or \emph{false}. Conditional composition is interpreted using a combination of connectives, thereby ignoring evaluation order, e.g., $P\lef Q\rig R$ is mapped to $(Q\to P)\wedge(\neg Q\to R)$. We used such a method to prove that CP1 and CP2 are independent with respect to CP. Since there are only a small finite number of interpretations of conditional composition it is easy to check them all. We, therefore, wrote a small Prolog program that checks these interpretations given a set of axioms to model and the one it should not model. There were no interpretations that proved the independence of the aforementioned axioms.

The second method relies on constructing a variety of RVAs. We used this method to show independence of e.g. CPcr1 and CPcr2. The problem is that such a variety is by definition a subvariety of $fr$. Since we proved that the axioms in CP are sound in $fr$ (see Theorem \ref{soundness cp}), they are sound in all subvarieties of $fr$. The axioms for which we need to prove independence are all in CP, and therefore this method will fail automatically.

This leaves us with the third option of constructing an interpretation in the natural numbers with the usual operations of addition and multiplication. Whether or not this method will work remains an open question. Of course, there are many other possibilities that are not listed here e.g. an interpretation in an n-valued model that takes the evaluation history into account.

\subsection{Future work}
In Chapter~1 we discussed some possible application areas for reactive valuations. However, we did not go into great detail as to the specifics of such applications, and more importantly what is to be gained by the use of reactive valuations. This search for specific applications might also yield new and interesting varieties of RVAs. 

Chapter~1 also contained a discussion on related work in which we mentioned that besides the main reference \cite{main} there is no related work directly pertaining to reactive valuations. We, therefore, listed some areas of interest with possible connections to reactive valuations. These and other areas warrant a more in-depth study, which may involve translations between varieties and other logics. For example, in Chapter~3 we showed a translation between the variety with static valuations and boolean algebra.

Proposition \ref{equiv congr prop} was, despite its apparent simplicity and usefulness, discovered towards the end of writing this thesis. This proposition which given variety $K$ describes the relationship between $K$-equivalence and $K$-congruence, is used extensively in the various soundness proofs throughout this thesis. We, however, have not yet fully taken advantage of this proposition in proving completeness. We feel that this proposition will in all likelihood simplify the proofs of some of the crucial lemmas we need for completeness.

Furthermore, we used basic forms to prove completeness. Alternatively, we can define a term rewriting system with convenient normal forms for each variety, and use that to prove completeness. In Appendix B we give an example of such a term rewriting system. 

In this thesis we gave some suggestions for future research. Of course, these suggestions are not exhaustive as there are many other options for further research not mentioned here e.g., results in expressivity and complexity.

On a final note, during the writing of this thesis a new paper on reactive valuations by Bergstra and Ponse appeared\footnote{In fact the main reference \cite{main} was a prior draft for this paper.}, see \cite{future}, which would be a fitting starting point for any future research.

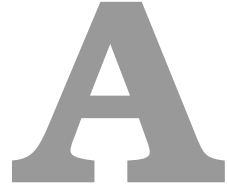
\appendix

\chapter{Characterization of CP+CP5}

\section{Non-replicating valuations}
In this appendix we define an additional variety, which uses non-replicating valuations. It is axiomatized by CP plus the CP5 axiom,
\[
\begin{array}{lrcl}
\text{(CP5)} & x\lef y\rig (z\lef u \rig(v\lef y\rig w)) & = & x\lef y\rig (z\lef u\rig w)
\end{array}
\]
The CP5 axiom is taken from the appendix of \cite{main}. In that appendix Bergstra and Ponse show that the symmetric version of this axiom can be derived from CP+CP5:
\begin{align*}
((x\lef y\rig z)\lef u\rig v)\lef y\rig w
&= w\lef\neg y\rig (v\lef \neg u\rig(z\lef\neg y\rig x))\\
&= w\lef\neg y\rig(v\lef \neg u\rig x)\\
&=(x\lef u\rig v)\lef y\rig w
\end{align*}
with $\neg x\syn F\lef x\rig T$.

The variety of RVAs with non-replicating valuations consists of all RVAs that satisfy the
equations:
\[
y_a(x) = y_a(\dd P(\dd a(x)))
\]
and
\[
\dd a(\dd P(\dd a(x))) = \dd P(\dd a(x))
\]
for all $a\in A$ and $P\in A\cup\{T,F\}$. We call this variety $nr$.

The following lemmas show that the above equations imply their more general versions. Note that the $P$ in the next lemma ranges over all closed terms not just $A\cup\{T,F\}$.

\begin{lem}\label{char lem one}
For all closed terms $P$,
\[
\dd a(\dd P(\dd a(x))) = \dd P(\dd a(x))
\]
\end{lem}
\begin{proof} By structural induction on $P$:
\begin{itemize}
\item If $P\in A\cup\{T,F\}$ then $\dd a(\dd T(\dd a(H)))=\dd T(\dd a(H))$ follows directly from the definition of $L$.
\item Suppose $P\syn P_1\lef P_2\rig P_3$. Then
\begin{align*}
\dd a(\dd{(P_1\lef P_2\rig P_3)}(\dd a(H)))
&=\begin{cases}
\dd a(\dd{P_1}(\dd{P_2}(\dd a(H)))) & \text{if $P_2/\dd a(H)=T$}\\
\dd a(\dd{P_3}(\dd{P_2}(\dd a(H)))) & \text{if $P_2/\dd a(H)=F$}\\
\end{cases}\\
&=_{IH}\begin{cases}
\dd a(\dd{P_1}(\dd a(\dd{P_2}(\dd a(H))))) & \text{if $P_2/\dd a(H)=T$}\\
\dd a(\dd{P_3}(\dd a(\dd{P_2}(\dd a(H))))) & \text{if $P_2/\dd a(H)=F$}\\
\end{cases}\\
&=_{IH}\begin{cases}
\dd{P_1}(\dd a(\dd{P_2}(\dd a(H)))) & \text{if $P_2/\dd a(H)=T$}\\
\dd{P_3}(\dd a(\dd{P_2}(\dd a(H)))) & \text{if $P_2/\dd a(H)=F$}\\
\end{cases}\\
&=_{IH}\begin{cases}
\dd{P_1}(\dd{P_2}(\dd a(H))) & \text{if $P_2/\dd a(H)=T$}\\
\dd{P_3}(\dd{P_2}(\dd a(H))) & \text{if $P_2/\dd a(H)=F$}\\
\end{cases}\\
&=\dd{P_1\lef P_2\rig P_3}(\dd a(H))
\end{align*}
\end{itemize}
\end{proof}

\begin{lem}\label{char lem two}
For all closed terms $P$ and $Q$,
\[
P/\dd Q(\dd P(x)) = P/x
\]
\end{lem}
\begin{proof}
By structural induction on $P$:
\begin{itemize}
\item The case for $P\syn T$ or $P\syn F$ is trivial because $T/H=T$ and
$F/H=F$ for any valuation $H$.
\item Suppose $P\syn a$ for $a\in A$. Then by structural induction on $Q$.
\begin{itemize}
\item The case for $Q\in A\cup\{T,F\}$ follows directly from
the definition of variety $L$.
\item Suppose $Q\syn Q_1\lef Q_2\rig Q_3$. Then
\begin{align*}
a/\dd{(Q_1\lef Q_2\rig Q_3)}(\dd a(H))
&=\begin{cases}
a/\dd{Q_1}(\dd{Q_2}(\dd a(H))) & \text{if $Q_2/\dd a(H)=T$}\\
a/\dd{Q_3}(\dd{Q_2}(\dd a(H))) & \text{if $Q_2/\dd a(H)=F$}\\
\end{cases}\\
&=\begin{cases}
a/\dd{Q_1}(\dd a(\dd{Q_2}(\dd a(H)))) & \text{if $Q_2/\dd a(H)=T$}\\
a/\dd{Q_3}(\dd a(\dd{Q_2}(\dd a(H)))) & \text{if $Q_2/\dd a(H)=F$}\\
\end{cases}\\
&=_{IH}\begin{cases}
a/\dd{Q_2}(\dd a(H)) & \text{if $Q_2/\dd a(H)=T$}\\
a/\dd{Q_2}(\dd a(H)) & \text{if $Q_2/\dd a(H)=F$}\\
\end{cases}\\
&=_{IH}\begin{cases}
a/H & \text{if $Q_2/\dd a(H)=T$}\\
a/H & \text{if $Q_2/\dd a(H)=F$}\\
\end{cases}\\
&= a/H\\
\end{align*}
In the second step of this derivation Lemma \ref{char lem
one} is applied i.e., we substitute $\dd{Q_2}(\dd a(H))$ with $\dd
a(\dd{Q_2}(\dd a(H)))$.
\end{itemize}
\item Suppose $P\syn P_1\lef P_2\rig P_3$. Then
\begin{align*}
&(P_1\lef P_2\rig P_3)/\dd Q(\dd{(P_1\lef P_2\rig P_3)}(H))\\
&=\begin{cases}
P_1/\dd{P_2}(\dd Q(\dd{(P_1\lef P_2\rig P_3)}(H))) & \text{if $P_2/\dd Q(\dd{(P_1\lef P_2\rig P_3)}(H))=T$}\\
P_3/\dd{P_2}(\dd Q(\dd{(P_1\lef P_2\rig P_3)}(H))) & \text{if $P_2/\dd Q(\dd{(P_1\lef P_2\rig P_3)}(H))=F$}\\
\end{cases}\\\\
&=\begin{cases}
P_1/\dd{P_2}(\dd Q(\dd{P_1}(\dd{P_2}(H)))) & \text{if $P_2/\dd Q(\dd{P_1}(\dd{P_2}(H)))=T$ and $P_2/H=T$}\\
P_1/\dd{P_2}(\dd Q(\dd{P_3}(\dd{P_2}(H)))) & \text{if $P_2/\dd Q(\dd{P_3}(\dd{P_2}(H)))=T$ and $P_2/H=F$}\\
P_3/\dd{P_2}(\dd Q(\dd{P_1}(\dd{P_2}(H)))) & \text{if $P_2/\dd Q(\dd{P_1}(\dd{P_2}(H)))=F$ and $P_2/H=T$}\\
P_3/\dd{P_2}(\dd Q(\dd{P_3}(\dd{P_2}(H)))) & \text{if $P_2/\dd Q(\dd{P_3}(\dd{P_2}(H)))=F$ and $P_2/H=F$}\\
\end{cases}\\\\
&=\begin{cases}
P_1/\dd{(Q\circ P_2)}(\dd{P_1}(\dd{P_2}(H))) & \text{if $P_2/\dd{(P_1\circ Q)}(\dd{P_2}(H))=T$ and $P_2/H=T$}\\
P_1/\dd{P_2}(\dd Q(\dd{P_3}(\dd{P_2}(H)))) & \text{if $P_2/\dd{(P_3\circ Q)}(\dd{P_2}(H))=T$ and $P_2/H=F$}\\
P_3/\dd{P_2}(\dd Q(\dd{P_1}(\dd{P_2}(H)))) & \text{if $P_2/\dd{(P_1\circ Q)}(\dd{P_2}(H))=F$ and $P_2/H=T$}\\
P_3/\dd{(Q\circ P_2)}(\dd{P_3}(\dd{P_2}(H))) & \text{if $P_2/\dd{(P_3\circ Q)}(\dd{P_2}(H))=F$ and $P_2/H=F$}\\
\end{cases}\\\\
&=_{IH}\begin{cases}
P_1/\dd{P_2}(H) & \text{if $P_2/H=T$ and $P_2/H=T$}\\
P_1/\dd{P_2}(\dd Q(\dd{P_3}(\dd{P_2}(H)))) & \text{if $P_2/H=T$ and $P_2/H=F$}\\
P_3/\dd{P_2}(\dd Q(\dd{P_1}(\dd{P_2}(H)))) & \text{if $P_2/H=F$ and $P_2/H=T$}\\
P_3/\dd{P_2}(H) & \text{if $P_2/H=F$ and $P_2/H=F$}\\
\end{cases}\\\\
&=\begin{cases}
P_1/\dd{P_2}(H) & \text{if $P_2/H=T$}\\
P_3/\dd{P_2}(H) & \text{if $P_2/H=F$}\\
\end{cases}\\\\
&=(P_1\lef P_2\rig P_3)/H
\end{align*}
Note that sequential composition $\circ$ is defined in Chapter~2, and $P\circ Q=Q\lef P\rig Q$.
\end{itemize}
\end{proof}

\begin{lem}\label{char lem three}
For all closed terms $P$ and $Q$,
\[
\dd P(\dd Q(\dd P(x))) = \dd Q(\dd P(x))
\]
\end{lem}
\begin{proof}
By structural induction on $P$:
\begin{itemize}
\item $P$ is $T$ or $F$; trivial.
\item $P\in A$; see Lemma \ref{char lem one}.
\item Suppose $P\syn P_1\lef P_2\rig P_3$. Then
\begin{align*}
&\dd{(P_1\lef P_2\rig P_3)}(\dd Q(\dd{(P_1\lef P_2\rig P_3)}(H)))\\
&=\begin{cases}
\dd{P_1}(\dd{P_2}(\dd Q(\dd{(P_1\lef P_2\rig P_3)}(H)))) & \text{if $P_2/\dd Q(\dd{(P_1\lef P_2\rig P_3)}(H))=T$}\\
\dd{P_3}(\dd{P_2}(\dd Q(\dd{(P_1\lef P_2\rig P_3)}(H)))) & \text{if $P_2/\dd Q(\dd{(P_1\lef P_2\rig P_3)}(H))=F$}\\
\end{cases}\\\\
&=\begin{cases}
\dd{P_1}(\dd{P_2}(\dd Q(\dd{P_1}(\dd{P_2}(H))))) & \text{if $P_2/\dd Q(\dd{P_1}(\dd{P_2}(H)))=T$ and $P_2/H=T$}\\
\dd{P_1}(\dd{P_2}(\dd Q(\dd{P_3}(\dd{P_2}(H))))) & \text{if $P_2/\dd Q(\dd{P_3}(\dd{P_2}(H)))=T$ and $P_2/H=F$}\\
\dd{P_3}(\dd{P_2}(\dd Q(\dd{P_1}(\dd{P_2}(H))))) & \text{if $P_2/\dd Q(\dd{P_1}(\dd{P_2}(H)))=F$ and $P_2/H=T$}\\
\dd{P_3}(\dd{P_2}(\dd Q(\dd{P_3}(\dd{P_2}(H))))) & \text{if $P_2/\dd Q(\dd{P_3}(\dd{P_2}(H)))=F$ and $P_2/H=F$}\\
\end{cases}\\\\
&=\begin{cases}
\dd{P_1}(\dd{Q\circ P_2}(\dd{P_1}(\dd{P_2}(H)))) & \text{if $P_2/\dd{P_1\circ Q}(\dd{P_2}(H))=T$ and $P_2/H=T$}\\
\dd{P_1}(\dd{P_2}(\dd Q(\dd{P_3}(\dd{P_2}(H))))) & \text{if $P_2/\dd{P_3\circ Q}(\dd{P_2}(H))=T$ and $P_2/H=F$}\\
\dd{P_3}(\dd{P_2}(\dd Q(\dd{P_1}(\dd{P_2}(H))))) & \text{if $P_2/\dd{P_1\circ Q}(\dd{P_2}(H))=F$ and $P_2/H=T$}\\
\dd{P_3}(\dd{Q\circ P_2}(\dd{P_3}(\dd{P_2}(H)))) & \text{if $P_2/\dd{P_3\circ Q}(\dd{P_2}(H))=F$ and $P_2/H=F$}\\
\end{cases}
\end{align*}
Now it is not only possible to apply the induction hypothesis but also
Lemma \ref{char lem two} which will result in a reduction in the
number of cases:
\begin{align*}
&=_{IH}\begin{cases}
\dd{Q\circ P_2}(\dd{P_1}(\dd{P_2}(H))) & \text{if $P_2/H=T$ and $P_2/H=T$}\\
\dd{P_1}(\dd{P_2}(\dd Q(\dd{P_3}(\dd{P_2}(H))))) & \text{if $P_2/H=T$ and $P_2/H=F$}\\
\dd{P_3}(\dd{P_2}(\dd Q(\dd{P_1}(\dd{P_2}(H))))) & \text{if $P_2/H=F$ and $P_2/H=T$}\\
\dd{Q\circ P_2}(\dd{P_3}(\dd{P_2}(H))) & \text{if $P_2/H=F$ and $P_2/H=F$}\\
\end{cases}\\\\
&=\begin{cases}
\dd{Q\circ P_2}(\dd{P_1}(\dd{P_2}(H))) & \text{if $P_2/H=T$}\\
\dd{Q\circ P_2}(\dd{P_3}(\dd{P_2}(H))) & \text{if $P_2/H=F$}\\
\end{cases}\\\\
&=\begin{cases}
\dd{P_2}(\dd{P_1\circ Q}(\dd{P_2}(H))) & \text{if $P_2/H=T$}\\
\dd{P_2}(\dd{P_3\circ Q}(\dd{P_2}(H))) & \text{if $P_2/H=F$}\\
\end{cases}\\\\
&=_{IH}\begin{cases}
\dd{P_1\circ Q}(\dd{P_2}(H)) & \text{if $P_2/H=T$}\\
\dd{P_3\circ Q}(\dd{P_2}(H)) & \text{if $P_2/H=F$}\\
\end{cases}\\\\
&=_{IH}\begin{cases}
\dd Q(\dd{P_1}(\dd{P_2}(H))) & \text{if $P_2/H=T$}\\
\dd Q(\dd{P_3}(\dd{P_2}(H))) & \text{if $P_2/H=F$}\\
\end{cases}\\\\
&=\dd Q(\dd{(P_1\lef P_2\rig P_3)}(H))
\end{align*}
\end{itemize}
\end{proof}

These three lemmas will demonstrate their usefulness in the next section where we will prove soundness.

\section{Soundness}
\begin{thm}
$\text{CP+CP5}\vdash P=Q$ implies that $P=_{nr}Q$.
\end{thm}
\begin{proof}
Since the axioms in CP are sound for the variety $fr$ of RVAs with free reactive valuations, and $nr$ is a subvariety of $fr$, it suffices to check CP5:
\begin{align*}
&(P\lef Q\rig(R\lef S\rig(V\lef Q\rig W)))/H\\
&=\begin{cases}
P/\dd Q(H) & \text{if $Q/H=T$}\\
R\lef S\rig(V\lef Q\rig W))/\dd Q(H) & \text{if $Q/H=F$}\\
\end{cases}\\
&=\begin{cases}
P/\dd Q(H) & \text{if $Q/H=T$}\\
R/\dd S(\dd Q(H)) & \text{if $Q/H=F$ and $S/\dd Q(H)=T$}\\
(V\lef Q\rig W)/\dd S(\dd Q(H)) & \text{if $Q/H=F$ and $S/\dd Q(H)=F$}\\
\end{cases}\\
&=\begin{cases}
P/\dd Q(H) & \text{if $Q/H=T$}\\
R/\dd S(\dd Q(H)) & \text{if $Q/H=F$ and $S/\dd Q(H)=T$}\\
V/\dd Q(\dd S(\dd Q(H))) & \text{if $Q/H=F$, $S/\dd Q(H)=F$ and $Q/\dd S(\dd Q(H))$=T}\\
W/\dd Q(\dd S(\dd Q(H))) & \text{if $Q/H=F$, $S/\dd Q(H)=F$ and $Q/\dd S(\dd Q(H))$=F}\\
\end{cases}\\
\end{align*}
Now Lemma \ref{char lem two} and Lemma \ref{char lem three} can be
applied:
\begin{align*}
&=\begin{cases}
P/\dd Q(H) & \text{if $Q/H=T$}\\
R/\dd S(\dd Q(H)) & \text{if $Q/H=F$ and $S/\dd Q(H)=T$}\\
V/\dd S(\dd Q(H)) & \text{if $Q/H=F$, $S/\dd Q(H)=F$ and $Q/H$=T}\\
W/\dd S(\dd Q(H)) & \text{if $Q/H=F$, $S/\dd Q(H)=F$ and $Q/H$=F}\\
\end{cases}\\
&=\begin{cases}
P/\dd Q(H) & \text{if $Q/H=T$}\\
R/\dd S(\dd Q(H)) & \text{if $Q/H=F$ and $S/\dd Q(H)=T$}\\
W/\dd S(\dd Q(H)) & \text{if $Q/H=F$ and $S/\dd Q(H)=F$}\\
\end{cases}\\
&= (P\lef Q\rig (R\lef S\rig W))/H
\end{align*}
Furthermore, we have
\[
\dd{(P\lef Q\rig(R\lef S\rig(V\lef Q\rig W)))}(H)
\]
\begin{align*}
&=\begin{cases}
\dd P(\dd Q(H)) & \text{if $Q/H=T$}\\
\dd R(\dd S(\dd Q(H))) & \text{if $Q/H=F$ and $S/\dd Q(H)=T$}\\
\dd V(\dd Q(\dd S(\dd Q(H)))) & \text{if $Q/H=F$, $S/\dd Q(H)=F$ and $Q/\dd S(\dd Q(H))=T$}\\
\dd W(\dd Q(\dd S(\dd Q(H)))) & \text{if $Q/H=F$, $S/\dd Q(H)=F$ and $Q/\dd S(\dd Q(H))=F$}\\
\end{cases}\\
&=\begin{cases}
\dd P(\dd Q(H)) & \text{if $Q/H=T$}\\
\dd R(\dd S(\dd Q(H))) & \text{if $Q/H=F$ and $S/\dd Q(H)=T$}\\
\dd V(\dd S(\dd Q(H))) & \text{if $Q/H=F$, $S/\dd Q(H)=F$ and $Q/H=T$}\\
\dd W(\dd S(\dd Q(H))) & \text{if $Q/H=F$, $S/\dd Q(H)=F$ and $Q/H=F$}\\
\end{cases}\\
&=\begin{cases}
\dd P(\dd Q(H)) & \text{if $Q/H=T$}\\
\dd R(\dd S(\dd Q(H))) & \text{if $Q/H=F$ and $S/\dd Q(H)=T$}\\
\dd W(\dd S(\dd Q(H))) & \text{if $Q/H=F$ and $S/\dd Q(H)=F$}\\
\end{cases}\\
&=\dd{(P\lef Q\rig(R\lef S\rig W))}(H)
\end{align*}
Hence, by Proposition \ref{equiv congr prop}, CP5 is sound.
\end{proof}

\chapter{Term rewriting system}

\section{Term rewriting for CP}
In this appendix we define a term rewriting system for CP and prove that it is convergent. For more information on term rewriting see \cite{rewriting}. We call the term rewriting system $R_{fr}$ and it is defined as follows:
\begin{align*}
x\lef T\rig y &\to x\\
x\lef F\rig y &\to y\\
T\lef x\rig F &\to x\\
x\lef(y\lef z\rig v)\rig w &\to (x\lef y\rig w)\lef z\rig(x\lef v\rig w)
\end{align*}

In the following lemma we show that $R_{fr}$ is terminating i.e., all terms can be reduced to a normal form. Note that this in itself does not guarantee that there is a unique normal form for each term.

\begin{lem}
$R_{fr}$ is terminating.
\end{lem}
\begin{proof}
In order to show that $R_{fr}$ is terminating, we are required to prove that an infinite derivation $t_1\to t_2\to t_3\to\ldots$ does not exist. First we define the norm on terms as follows
\begin{align*}
|T| &= 1\\
|F| &= 1\\
|a| &= 1\\
|x| &= 1\\
|t\lef r\rig s| &= 2|r|+max(|t|, |s|)
\end{align*}
Subsequently, we show that for each rewrite rule the norm of the left-hand side 
is strictly greater than the norm of the right-hand side.
\begin{align*}
|t\lef T\rig s|
&= 2|T| + max(|t|, |s|)\\
&= 2 + max(|t|, |s|)\\
&> |t|\\\\
|t\lef F\rig s|
&= 2|F| + max(|t|, |s|)\\
&= 2 + max(|t|, |s|)\\
&> |s|\\\\
|T\lef s\rig F|
&= 2|s| + max(|T|, |F|)\\
&= 2|s| + 1\\
&> |s|\\\\
|t\lef(r\lef s\rig v)\rig w|
&= 2|r\lef s\rig v| + max(|t|, |w|)\\
&= 2(2|s| + max(|r|, |v|)) + max(|t|, |w|)\\
&= 4|s| + max(2|r|, 2|v|) + max(|t|, |w|)\\
&= 4|s| + max(2|r| + max(|t|, |w|), 2|v| + max(|t|, |w|))\\
&= 4|s| + max(|t\lef r\rig w|, |t\lef v\rig w|)\\
&> 2|s| + max(|t\lef r\rig w|, |t\lef v\rig w|)\\
&= |(t\lef r\rig w)\lef s\rig(t\lef v\rig w)|
\end{align*}
Consequently, if $t\to r$ using these rules, it follows that $|t|>|r|$.

Suppose that there is an infinite rewrite sequence $t_1\to t_2\to t_3\to\ldots$. Then we know that if $i < j$, $|t_i|>|t_j|$. Since the norm is always positive and finite (we do not allow terms with an infinite number of symbols), this sequence must end. Hence, such an infinite rewrite sequence does not exist and $R_{fr}$ is terminating.
\end{proof}

The following lemma shows that $R_{fr}$ is locally confluent. If there is a term $u$ mutually derivable from terms $s$ and $t$, we write $s\downarrow t$ and say that $s$ and $t$ are \emph{joinable}. A binary relation $\to$ is then \emph{locally confluent} if for all terms $r$, $s$ and $t$, if $s\downarrow t$ whenever $r\to s$ and $r\to t$.

In order to prove local confluence, we first need to define the concept of critical pair. Let $s\to t$ and $u\to v$ be two rewrite rules with variables renamed such that they do not share variables. Furthermore, let $\mu$ be the most general unifier of $u$ and a nonvariable subterm $s'$ of $s$. A \emph{critical pair} is then the term $\mu(t)$ combined with the term resulting from replacing $\mu(s')$ with $\mu(v)$ in $\mu(s)$. A critical pair $(u_1, u_2)$ is joinable if $u_1$ and $u_2$ are joinable.

\begin{lem}
$R_{fr}$ is locally confluent.
\end{lem}
\begin{proof}
$R_{fr}$ is locally confluent if all its critical pairs are joinable (see Lemma 5.15 in \cite{rewriting}). We first rename the variables in the rules so they are distinct:
\begin{align*}
x_1\lef T\rig y_1 &\to x_1\\
x_2\lef F\rig y_2 &\to y_2\\
T\lef x_3\rig F &\to x_3\\
x_4\lef(y_4\lef z_4\rig v_4)\rig w_4 &\to (x_4\lef y_4\rig w_4)\lef z_4\rig(x_4\lef v_4\rig w_4)
\end{align*}
Then we identify the critical pairs and check whether they are joinable.

\begin{itemize}

\item Let $\mu_1$ be a substitution such that
\begin{align*}
\mu_1(y_4) &= x_1\\
\mu_1(z_4) &= T\\
\mu_1(v_4) &= y_1\\
\end{align*}
and the rest of the variables map to themselves e.g., $\mu_1(x_1)=x_1$ and $\mu_1(x_4)=x_4$. Then we have the following critical pair
\[
(\mu_1(x_4\lef x_1\rig w_4),~ \mu_1((x_4\lef y_4\rig w_4)\lef z_4\rig(x_4\lef v_4\rig w_4)))
\]
which is joinable
\begin{align*}
\mu_1((x_4\lef y_4\rig w_4)\lef z_4\rig(x_4\lef v_4\rig w_4))
&=(x_4\lef x_1\rig w_4)\lef T\rig (x_4\lef y_1\rig w_4)\\
&\to x_4\lef x_1\rig w_4\\
&= \mu_1(x_4\lef x_1\rig w_4)
\end{align*}

\item Let $\mu_2$ be a substitution such that
\begin{align*}
\mu_2(y_4) &= x_2\\
\mu_2(z_4) &= F\\
\mu_2(v_4) &= y_2\\
\end{align*}
and the rest of the variables map to themselves. Then we have the following critical pair
\[
(\mu_2(x_4\lef y_2\rig w_4),~ \mu_2((x_4\lef y_4\rig w_4)\lef z_4\rig(x_4\lef v_4\rig w_4)))
\]
which is joinable
\begin{align*}
\mu_2((x_4\lef y_4\rig w_4)\lef z_4\rig(x_4\lef v_4\rig w_4))
&=(x_4\lef x_2\rig w_4)\lef F\rig (x_4\lef y_2\rig w_4)\\
&\to x_4\lef y_2\rig w_4\\
&= \mu_2(x_4\lef y_2\rig w_4)
\end{align*}

\item Let $\mu_3$ be a substitution such that
\begin{align*}
\mu_3(y_4) &= T\\
\mu_3(z_4) &= x_3\\
\mu_3(v_4) &= F\\
\end{align*}
and the rest of the variables map to themselves. Then we have the following critical pair
\[
(\mu_3(x_4\lef x_3\rig w_4),~ \mu_3((x_4\lef y_4\rig w_4)\lef z_4\rig(x_4\lef v_4\rig w_4)))
\]
which is joinable
\begin{align*}
\mu_3((x_4\lef y_4\rig w_4)\lef z_4\rig(x_4\lef v_4\rig w_4))
&=(x_4\lef T\rig w_4)\lef x_3\rig (x_4\lef F\rig w_4)\\
&\to x_4\lef x_3\rig (x_4\lef F\rig w_4)\\
&\to x_4\lef x_3\rig w_4\\
&= \mu_3(x_4\lef x_3\rig w_4)
\end{align*}

\item Let $\mu_4$ be a substitution such that
\begin{align*}
\mu_4(x_1) &= T\\
\mu_4(x_3) &= T\\
\mu_4(y_1) &= F\\
\end{align*}
and the rest of the variables map to themselves. Then we have the following critical pair
\[
(T,~T)
\]
which is joinable.

\item Let $\mu_5$ be a substitution such that
\begin{align*}
\mu_5(x_2) &= T\\
\mu_5(x_3) &= F\\
\mu_5(y_2) &= F\\
\end{align*}
and the rest of the variables map to themselves. Then we have the following critical pair
\[
(F,~F)
\]
which is joinable.

\item Let $\mu_6$ be a substitution such that
\begin{align*}
\mu_6(x_3) &= y_4\lef z_4\rig v_4\\
\mu_6(x_4) &= T\\
\mu_6(w_4) &= F\\
\end{align*}
and the rest of the variables map to themselves. Then we have the following critical pair
\[
(\mu_6(x_3),~\mu_6((x_4\lef y_4\rig w_4)\lef z_4\rig(x_4\lef v_4\rig w_4)))
\]
which is joinable
\begin{align*}
\mu_6((x_4\lef y_4\rig w_4)\lef z_4\rig(x_4\lef v_4\rig w_4))
&=(T\lef y_4\rig F)\lef z_4\rig(T\lef v_4\rig F)\\
&\to y_4\lef z_4\rig(T\lef v_4\rig F)\\
&\to y_4\lef z_4\rig v_4\\
&=\mu_6(x_3)\\
\end{align*}

\item The last critical pair requires that we rename the variables in the fourth rule a second time
\[
x'_4\lef(y'_4\lef z'_4\rig v'_4)\rig w'_4 \to (x'_4\lef y'_4\rig w'_4)\lef z'_4\rig(x'_4\lef v'_4\rig w'_4)
\]
Let $\mu_7$ be a substitution such that
\begin{align*}
\mu_7(y_4) &= x'_4\\
\mu_7(z_4) &= y'_4\lef z'_4\rig v'_4\\
\mu_7(v_4) &= w'_4\\
\end{align*}
and the rest of the variables map to themselves. Then we have the following critical pair
\[
(\mu_7(x_4\lef((x'_4\lef y'_4\rig w'_4)\lef z'_4\rig(x'_4\lef v'_4\rig w'_4))\rig w_4),~ \mu_7((x_4\lef y_4\rig w_4)\lef z_4\rig(x_4\lef v_4\rig w_4)))
\]
which is joinable
\[
\mu_7((x_4\lef y_4\rig w_4)\lef z_4\rig(x_4\lef v_4\rig w_4))
\]
\begin{align*}
&=(x_4\lef x'_4\rig w_4)\lef(y'_4\lef z'_4\rig v'_4)\rig(x_4\lef w'_4\rig w_4)\\
&\to x_4\lef(x'_4\lef(y'_4\lef z'_4\rig v'_4)\rig w'_4)\rig w_4\\
&\to x_4\lef((x'_4\lef y'_4\rig w'_4)\lef z'_4\rig(x'_4\lef v'_4\rig w'_4))\rig w_4\\
&=\mu_7(x_4\lef((x'_4\lef y'_4\rig w'_4)\lef z'_4\rig(x'_4\lef v'_4\rig w'_4))\rig w_4)\\
\end{align*}

\end{itemize}
Since every critical pair is joinable, the rewrite system $R_{fr}$ is locally confluent.
\end{proof}

By Lemma 5.13 (the so called Diamond Lemma, see \cite{newman}) in \cite{rewriting}, a terminating binary relation is Church-Rosser iff it is locally confluent. Hence, $R_{fr}$ is Church-Rosser. Furthermore, by Theorem 5.4 in \cite{rewriting} a binary relation is confluent iff it is Church-Rosser. Consequently, $R_{fr}$ is also confluent, which leads us to the final theorem.

\begin{thm}
$R_{fr}$ is convergent.
\end{thm}
\begin{proof}
By Definition 5.6 in \cite{rewriting}.
\end{proof}

This means that term rewriting system $R_{fr}$ has unique normal forms. In the next section, we list a program that uses this result to determine whether two terms are provably equal in CP.

\section{Theorem prover for CP}

In this section we list the code of a small theorem prover for CP based on the term-rewriting system $R_{fr}$ in the previous section. The program is written in Prolog, and has been tested in the \href{http://www.swi-prolog.org}{SWI-Prolog} interpreter.

Within the program we use \verb|1| and \verb|0| to denote $T$ and $F$, respectively. Atomic propositions have the same notation i.e. atomic proposition $a$ is represented by \verb|a|. Conditional composition is represented as a three-place predicate \verb|c(_,_,_)| in which the middle argument is the antecedent and the first and third argument are the left- and right-consequent, respectively. So, for example, the term $a\lef (T\lef b\rig c)\rig F$ is represented by the term \verb|c(a, c(1, b, c), 0)|.

After loading the program in the Prolog interpreter we can check whether two terms are equal as follows:
\begin{verbatim}
-? equals(Term1, Term2).
\end{verbatim}
where you should replace \verb|Term1| and \verb|Term2| with two terms using the notation we just discussed. The \verb|equals| predicate will compute the normal form of each term and determine if they are equal or not.

In the code below we see several other predicates. We now give a short description of each of these predicates.

The \verb|rule| predicate describes the term rewriting rules i.e., in this case the rules of $R_{fr}$.

The \verb|normal_form(+Term, -NormalForm)| predicate computes the normal form \verb|NormalForm| of the term \verb|Term|.

The \verb|subterm(-Subterm, +Term)| predicate returns a subterm \verb|Subterm| of the term \verb|Term|. 

The \verb|substitute(+Subterm1, +Subterm2, +Term1, -Term2)| predicate replaces all occurrences of \verb|Subterm1| in \verb|Term1| with \verb|Subterm2|, and returns the resulting term as \verb|Term2|.
 
\begin{verbatim}
rule(c(X, 1, _), X).
rule(c(_, 0, Y), Y).
rule(c(1, X, 0), X).
rule(c(X, c(Y, Z, U), V), c(c(X, Y, V), Z, c(X, U, V))).

normal_form(Term, Term) :-
   findall(Subterm, subterm(Subterm, Term), Subterms),
   forall(member(X, Subterms), \+ rule(X, _)).
normal_form(Term1, NormalForm) :-
   subterm(Subterm1, Term1),
   rule(Subterm1, Subterm2),
   substitute(Subterm1, Subterm2, Term1, Term2),
   normal_form(Term2, NormalForm).

subterm(T, T).
subterm(T1, T2) :-
   T2 =.. [_|T],
   member(T3, T),
   subterm(T1, T3).

substitute(Term1, Term2, Term1, Term2) :- !.
substitute(_, _, Term, Term) :-
   Term \= c(_,_,_).
substitute(Term1, Term2, c(X, Y, Z), c(NewX, NewY, NewZ)) :-
   substitute(Term1, Term2, X, NewX),
   substitute(Term1, Term2, Y, NewY),
   substitute(Term1, Term2, Z, NewZ).

equal(Term1, Term2) :-
   normal_form(Term1, NormalForm),
   normal_form(Term2, NormalForm).
\end{verbatim}

\bibliographystyle{plain}
\bibliography{thesis}

\begin{thebibliography}{10}

\bibitem{non-monotonic2}
G.~A. Antonelli.
\newblock Non-monotonic logic.
\newblock In Edward~N. Zalta, editor, {\em The Stanford Encyclopedia of
  Philosophy}. Summer edition, 2010.

\bibitem{connectives}
J.~A. Bergstra, I.~Bethke, and P.~H. Rodenburg.
\newblock A propositional logic with 4 values: true, false, divergent and
  meaningless.
\newblock {\em Journal of Applied Non-Classical Logics}, 5(2):199--218, 1995.

\bibitem{omega_bergstra}
J.~A. Bergstra and J.~Heering.
\newblock Which data types have $\omega$-complete initial algebra
  specifications?
\newblock {\em Theoretical Computer Science}, 124(1):149--168, 1994.

\bibitem{future}
J.~A. Bergstra and A.~Ponse.
\newblock Proposition algebra.
\newblock {\em To appear in Transactions on Computational Logic}.
\newblock Available at http://tocl.acm.org/accepted.html.

\bibitem{main}
J.~A. Bergstra and A.~Ponse.
\newblock Proposition algebra with projective limits.
\newblock {\em Available under arXiv:0807.3648v3 at arXiv}, September 2008.

\bibitem{rewriting}
N.~Dershowitz and D.~A. Plaisted.
\newblock Rewriting.
\newblock In {\em Handbook of Automated Reasoning}. Elsevier, 2001.

\bibitem{fokkink}
W.~Fokkink.
\newblock {\em Introduction to Process Algebra}.
\newblock Texts in theoretical computer science. Springer, 2000.

\bibitem{hoare}
C.~A.~R. Hoare.
\newblock A couple of novelties in the propositional calculus.
\newblock {\em Zeitschrift fur Mathematische Logik und Grundlagen der
  Mathematik}, 31(2):173--178, 1985.

\bibitem{DHT}
H.~Kamp.
\newblock A theory of truth and semantic representation.
\newblock In J.~A.~G. Groenendijk, T.~M.~V. Janssen, and M.~B.~J. Stokhof,
  editors, {\em Formal Methods in the Study of Language}, Mathematical Centre
  Tracts 135, pages 277--322, 1981.

\bibitem{defeasible}
R.~Koons.
\newblock Defeasible reasoning.
\newblock In Edward~N. Zalta, editor, {\em The Stanford Encyclopedia of
  Philosophy}. Winter edition, 2009.

\bibitem{kracht}
M.~Kracht.
\newblock Logic and control: How they determine the behaviour of
  presuppositions.
\newblock In {\em Logic and Information Flow}, pages 89--111. MIT Press, 1994.

\bibitem{omega}
A.~Lazrek, P.~Lescanne, and J.~Thiel.
\newblock Tools for proving inductive equalities, relative completeness and
  $\omega$-completeness.
\newblock {\em Information and Computation}, 84:47--70, 1990.

\bibitem{prolog_expr}
J.~A. Makowsky, J.~C. Gr\'egoire, and S.~Sagiv.
\newblock The expressive power of side effects in prolog.
\newblock {\em The Journal of Logic Programming}, 12(1-2):179--188, January
  1992.

\bibitem{newman}
M.~H.~A. Newman.
\newblock On the theories with a combinatorial definition of equivalence.
\newblock {\em Annals of Mathematics}, 43(2):223--243, 1942.

\bibitem{prolog_semantics}
T.~Nicholson and N.~Foo.
\newblock A denotational semantics for prolog.
\newblock {\em ACM Transactions on Programming Languages and Systems (TOPLAS)},
  11(4):650--665, 1989.

\bibitem{belnap_cond}
A.~Ponse and M.~B. van~der Zwaag.
\newblock Belnap's logic and conditional composition.
\newblock {\em Theoretical Computer Science}, 388(1-3):319--336, 2007.

\bibitem{non-monotonic}
H.~Tompits.
\newblock A survey of non-monotonic reasoning.
\newblock {\em Open Systems \& Information Dynamics}, 3(3):369--395, 1995.

\bibitem{PDL}
J.~van Eijck and M.~Stokhof.
\newblock The gamut of dynamic logics.
\newblock In {\em The Handbook of the History of Logic}, volume~7, pages
  499--600. Elsevier, 2006.

\end{thebibliography}

\end{document}